\documentclass[a4paper,10pt]{article}
\usepackage[bbgreekl]{mathbbol}
\usepackage[small]{titlesec}
\usepackage[utf8]{inputenc}
\usepackage[T1]{fontenc}  
\usepackage[]{mdframed}
\usepackage[backend=biber,style=alphabetic,natbib=false,giveninits=true,doi=true,isbn=false,url=false,date=year,maxbibnames=99,sorting=ynt,defernumbers=true]{biblatex}
\DeclareNameAlias{default}{family-given}

\DeclareFieldFormat[article]{title}{#1}
\renewbibmacro{in:}{}
\DeclareFieldFormat[article]{pages}{#1}
\DeclareFieldFormat{journal}{#1}
\renewcommand\bf\bfseries

\renewbibmacro*{volume+number+eid}{%
	\printfield{volume}%
	\setunit*{\addnbspace}
	\printfield{number}%
	\setunit{\addcomma\space}%
	\printfield{eid}}
\DeclareFieldFormat[article]{number}{\mkbibparens{#1}}
\DeclareFieldFormat[article]{volume}{\textbf{#1}}
\DeclareFieldFormat{year}{\mkbibparens{#1}}
\DeclareBibliographyDriver{article}{%
	\printnames{author}:%
	\newunit\newblock
	\printfield{title}%
	\newunit\newblock
	\printfield{journaltitle}
	\newunit
	\iffieldundef{number}{\printfield{volume}}{\printfield{volume}\addspace\printfield{number}}
	\addcomma\addspace\printfield{pages}\addspace
	\printfield{year}
}
\AtEveryBibitem{\clearfield{month}}
\AtEveryBibitem{\clearfield{day}}
\usepackage[english]{babel}
\usepackage[T1]{fontenc}
\usepackage{csquotes}
\usepackage{bbm}
\usepackage[leqno]{amsmath}
\usepackage{amsfonts,amsthm,amsbsy,amssymb,dsfont,stmaryrd}
\usepackage[dvipsnames]{xcolor}
\usepackage{braket}

\makeatletter
\newcommand{\leqnomode}{\tagsleft@true\let\veqno\@@leqno}
\newcommand{\reqnomode}{\tagsleft@false\let\veqno\@@eqno}
\makeatother

\usepackage{mathtools}
\usepackage[makeroom]{cancel}

\numberwithin{equation}{section}

\newcommand\myshade{85}
\colorlet{mylinkcolor}{violet}
\colorlet{mycitecolor}{YellowOrange}
\colorlet{myurlcolor}{Aquamarine}

\usepackage[left=2.5cm,right=2.5cm,top=2.5cm,bottom=2.5cm]{geometry}
\usepackage[unicode=true,pdfusetitle,
bookmarks=true,bookmarksnumbered=false,bookmarksopen=false,
breaklinks=false,pdfborder={0 0 1},backref=false,
linkcolor  = mylinkcolor!\myshade!black,
citecolor  = mycitecolor!\myshade!black,
urlcolor   = myurlcolor!\myshade!black,
colorlinks = true,
]
{hyperref}
\usepackage[nameinlink]{cleveref}

\usepackage{tikz}
\usetikzlibrary{arrows}
\usetikzlibrary{intersections}

\usepackage{graphicx}
\usepackage{caption}
\usepackage{slashed}
\definecolor{ct_black}{HTML}{000000}
\definecolor{ct_orange}{HTML}{ED872D}
\definecolor{ct_purple}{HTML}{7A68A6}
\definecolor{ct_blue}{HTML}{348ABD}
\definecolor{ct_turquoise}{HTML}{188487}
\definecolor{ct_red}{HTML}{E32636}
\definecolor{ct_pink}{HTML}{CF4457}
\definecolor{ct_green}{HTML}{467821}

\definecolor{ct2_green}{HTML}{9FF781}
\definecolor{ct2_green_dark}{HTML}{088A08}

\theoremstyle{plain}
\newtheorem{thm}{\protect\theoremname}[section]
\theoremstyle{plain}
\newtheorem{lem}[thm]{\protect\lemmaname}
\theoremstyle{plain}
\newtheorem{cor}[thm]{\protect\corollaryname}
\theoremstyle{plain}

\theoremstyle{plain}

\newtheorem{assumption}[thm]{\protect\assumptionname}

\theoremstyle{remark}
\newtheorem{rem}[thm]{\protect\remarkname}

\theoremstyle{definition}
\newtheorem{defn}[thm]{\protect\definitionname}
\theoremstyle{plain}
\newtheorem{example}[thm]{\protect\examplename}
\providecommand{\assumptionname}{Assumption}
\providecommand{\claimname}{Claim}
\providecommand{\corollaryname}{Corollary}
\providecommand{\definitionname}{Definition}
\providecommand{\lemmaname}{Lemma}
\providecommand{\propositionname}{Proposition}
\providecommand{\remarkname}{Remark}
\providecommand{\theoremname}{Theorem}
\providecommand{\examplename}{Example}

\crefname{section}{Section}{Sections}
\crefname{appendix}{Appendix}{Appendices}
\crefname{figure}{Figure}{Figures}
\crefname{assumption}{Assumption}{Assumptions}
\crefname{thm}{Theorem}{Theorems}
\crefname{lem}{Lemma}{Lemmas}
\crefformat{equation}{(#2#1#3)}
\crefname{table}{Table}{Tables}

\crefrangelabelformat{equation}{(#3#1#4--#5#2#6)}

\crefmultiformat{equation}{(#2#1#3}{, #2#1#3)}{#2#1#3}{#2#1#3}
\Crefmultiformat{equation}{(#2#1#3}{, #2#1#3)}{#2#1#3}{#2#1#3}

\newtheorem*{lem*}{\protect\lemmaname}

\newcommand{\ee}{\operatorname{e}}
\newcommand{\ii}{\operatorname{i}}
\newcommand{\Mat}{\operatorname{Mat}}

\newcommand{\ZZ}{\mathbb{Z}}

\newcommand{\bS}{\mathbb{S}}

\newcommand{\NN}{\mathbb{N}}
\newcommand{\RR}{\mathbb{R}}
\newcommand{\CC}{\mathbb{C}}
\newcommand{\FF}{\mathbb{F}}

\newcommand{\calA}{\mathcal{A}}
\newcommand{\calB}{\mathcal{B}}
\newcommand{\calC}{\mathcal{C}}
\newcommand{\calF}{\mathcal{F}}

\newcommand{\calG}{\mathcal{G}}

\newcommand{\calU}{\mathcal{U}}
\newcommand{\calSU}{\mathcal{S}}
\newcommand{\calH}{\mathcal{H}}

\newcommand{\calK}{\mathcal{K}}
\newcommand{\calL}{\mathcal{L}}
\newcommand{\calI}{\mathcal{I}}

\newcommand{\calP}{\mathcal{P}}
\newcommand{\calS}{\mathcal{S}}

\newcommand{\bbLambda}{\mathbb{\Lambda}}
\newcommand{\bbDelta}{\mathbb{\Delta}}
\newcommand{\ti}[1]{\widetilde{#1}} 
\newcommand{\LamNT}{\Lambda\mathrm{nt}}
\newcommand{\LamNTSAU}{\calSU^{\LamNT}}

\newcommand\norm[1]{\left\lVert#1\right\rVert}
\newcommand{\ip}[2]{\langle #1, #2 \rangle}

\DeclareMathOperator*{\slim}{s-lim}
\newcommand{\dif}{\operatorname{d}}

\newcommand{\szpan}{\operatorname{span}}

\newcommand{\ve}{\varepsilon}
\newcommand{\vf}{\varphi}
\newcommand{\Id}{\mathds{1}}
\newcommand{\HH}{\mathbb{H}}

\newcommand{\Closed}[1]{\mathrm{Closed}(#1)}

\newcommand{\quotes}[1]{``#1''}

\newcommand{\sgn}{\operatorname{sgn}}
\newcommand{\findex}{\operatorname{ind}}

\newcommand{\im}{\operatorname{im}}

\usepackage{environ}

\NewEnviron{malign}{%
	\begin{align}\begin{split}
			\BODY
	\end{split}\end{align}
}

\newcommand{\eq}[1]{\begin{align*}#1\end{align*}}
\newcommand{\eql}[1]{\begin{align}#1\end{align}}

\newcommand{\polar}{\operatorname{pol}}

\newcommand{\calT}{\mathcal{T}}

\newcommand{\DD}{\mathbb{D}}
\newcommand{\Schur}{Z}

\title{Topological Classification of Insulators:\\ I. Non-interacting Spectrally-Gapped One-Dimensional Systems}
\author{\href{mailto:jc1220@math.princeton.edu}{Jui-Hui Chung}\\
	{\footnotesize Department of Applied Mathematics, Princeton University }\\
	\href{mailto:jacobshapiro@princeton.edu}{Jacob Shapiro}\\
	{\footnotesize Department of Mathematics, Princeton University}
}

\begin{document}
	\reqnomode
	
	\maketitle
	\begin{abstract}
		We study non-interacting electrons in disordered one-dimensional materials which exhibit a spectral gap, in each of the ten Altland-Zirnbauer symmetry classes. We define an appropriate topology on the space of Hamiltonians so that the so-called strong topological invariants become complete invariants yielding the one-dimensional column of the Kitaev periodic table, but now derived \emph{without} recourse to K-theory. We thus confirm the conjecture regarding a one-to-one correspondence between topological phases of gapped non-interacting 1D systems and the respective Abelian groups $\Set{0},\ZZ,2\ZZ,\ZZ_2$ in the spectral gap regime. The main tool we develop is an equivariant theory of homotopies of \emph{local} unitaries and orthogonal projections. Moreover, we discuss an extension of the unitary theory to partial isometries, to provide a perspective towards the understanding of strongly-disordered, mobility-gapped materials.
	\end{abstract}
	\tableofcontents
	\section{Introduction}
	Topological insulators \cite{Hasan_Kane_2010} are exotic materials which insulate in their bulk, but may be excellent conductors along their boundary. The quintessential example is Galium-Arsenic in two dimensions, at very low temperatures and strong perpendicular magnetic fields, which exhibits the celebrated integer quantum Hall effect (IQHE) \cite{vKDP1980_PhysRevLett.45.494}. Beyond the aforementioned typical bulk-boundary behavior \cite{Graf07}, another defining feature of these materials is that they exhibit observables which are quantized and experimentally stable--a manifestation of macroscopic quantum mechanical effects. Mathematically this phenomenon suggests a global, topological description and indeed Nobel prizes have been awarded \cite{Nobel2016} for the association of the integer quantum Hall effect with the mathematical theory of algebraic topology, see e.g. \cite{TKNN_1982_PhysRevLett.49.405,ASS_1983_PhysRevLett.51.51}. A decisive step was taken by Kitaev \cite{Kitaev2009} who devised a periodic \cref{table:Kitaev} of insulators organized by the Altland-Zirnbauer symmetry classes \cite{AltlandZirnbauer1997} and patterned after K-theoretic Bott periodicity. The classification problem which is in present focus here enjoyed much attention recently in the mathematics literature, from various perspectives, see e.g. \cite{FreedMoore2013,DeNittisGomi2015,Kubota2016,Thiang2016,BourneCareyRennie2016,PSB_2016,GrossmannSchulz-Baldes2016,KatsuraKoma2018,Kellendonk2019,AlldridgeMaxZirnbauer2020,BourneSchulz-Baldes2020,BourneOgata2021,AvronTurner2022,GontierMonacoSolal2022}.
	
	\begin{table}[b]
		\begin{center}
			\begin{tabular}{c|ccc|cccccccc}
				\multicolumn{4}{c|}{Symmetry } & \multicolumn{8}{c}{dimension} \\
				\multicolumn{1}{c}{AZ} &$\hspace{1.5mm}\Theta\hspace{1.5mm} $ &
				$\hspace{1.5mm} \Xi\hspace{1.5mm} $ &
				$\hspace{1.5mm} \Pi\hspace{1.5mm} $ &
				$1$   &  $2$ &  $3$ &  $4$ &  $5$ & $6$ & $7$& $8$ \\
				\hline
				A & $0$ & $0$ & $0$  &$0$& $\mathbb{Z}$ &$0$& $\mathbb{Z}$ &$0$& $\mathbb{Z}$ &$0$& $\mathbb{Z}$\\
				AIII & $0$ & $0$ & $1$ & $\mathbb{Z}$ &$0$& $\mathbb{Z}$ &$0$& $\mathbb{Z}$ &$0$& $\mathbb{Z}$& $0$\\
				\hline
				AI & $1$ & $0$ & $0$  &$0$&$0$&$0$&$\mathbb{Z}$&$0$&$\mathbb{Z}_2$&$\mathbb{Z}_2$& $\mathbb{Z}$ \\
				BDI & $1$ &$1$ &$1$ & $\mathbb{Z}$ &$0$&$0$&$0$&$\mathbb{Z}$&$0$&$\mathbb{Z}_2$& $\mathbb{Z}_2$\\
				D & $0$ &$1$ &$0$ & $\mathbb{Z}_2$& $\mathbb{Z}$ &$0$&$0$&$0$&$\mathbb{Z}$&$0$&$\mathbb{Z}_2$\\
				DIII&$-1$ &$1$ &$1$ &$\mathbb{Z}_2$& $\mathbb{Z}_2$& $\mathbb{Z}$ &$0$&$0$&$0$&$\mathbb{Z}$&$0$\\
				AII & $-1$ & $0$ & $0$ &$0$&$\mathbb{Z}_2$& $\mathbb{Z}_2$& $\mathbb{Z}$ &$0$&$0$& $0$&$\mathbb{Z}$\\
				CII & $-1$ &$-1$ & $1$&$\mathbb{Z}$ & $0$&$\mathbb{Z}_2$& $\mathbb{Z}_2$& $\mathbb{Z}$ &$0$&$0$&$0$ \\
				C & $0$ & $-1$& $0$ & $0$ &$\mathbb{Z}$ &$0$&$\mathbb{Z}_2$& $\mathbb{Z}_2$& $\mathbb{Z}$ &$0$& $0$\\
				CI & $1$ & $-1$ & $1$& $0$ & $0$&$\mathbb{Z}$&$0$&$\mathbb{Z}_2$& $\mathbb{Z}_2$& $\mathbb{Z}$& $0$ \\
			\end{tabular}
		\end{center}
		\caption{The Kitaev periodic table. The entries stand for the respective K-theory groups in a given dimension and symmetry class. The present paper focuses on the one-dimensional column. See \cref{sec:bulk insulators classification} for explanation of $\Theta,\Xi,\Pi$.}
		\label{table:Kitaev}
	\end{table}

	A first presentation of the association of quantum mechanics of insulators with algebraic topology would assume periodicity of the materials involved, which leads very naturally to the theory of equivariant vector bundles and their classification via K-theory, culminating in, e.g., \cite{FreedMoore2013}. However, while vector bundle theory is mathematically classical, a periodic model cannot describe strong-disorder, an important feature of topological insulators (see below). 
	This has been recognized early on by Bellissard and collaborators \cite{Bellissard_1994JMP....35.5373B} who have laid important ground work in the 1990s to build bridges from physics into K-theory of C-star algebras and use Connes' tools from non-commutative geometry to study what they refer to as the non-commutative Brillouin zone. And yet using K-theory bears a price: it allows homotopies to explore additional internal degrees of freedom, and it only studies \emph{relative} phases. These two points mean the classification is more fuzzy than one would hope for (this point receives some attention in \cite{DeNittisGomi2015}). For this reason one might argue that K-theoretic classifications do not offer a one-to-one correspondence between topological phases of gapped systems and the respective Abelian K-theory groups. More severely, there does not seem to be a way to extend it to the strongly-disordered mobility gap regime--the description remains in the disordered spectral gap regime. Moreover, K-theory of C-star algebras with real or quaternionic structures is difficult to handle since (as far as we are aware) its dual, which is necessary to study index pairings, is not defined. These latter two points are somewhat addressed by Kasparov's KK-theory, which is however vastly more complicated and (as far as we are aware) still cannot address the mobility gap problem.
	
	Let us expand on the mobility gap regime briefly. The physical situation of materials being insulators is encoded mathematically by operators that have a certain gap. In the simplest scenario this is a spectral gap about the Fermi energy. But it turns out that when strong disorder is present (i.e., under Anderson localization) this spectral gap closes and the Fermi energy is immersed in an interval of localized states which cannot contribute to electric conductance, a situation referred to as the mobility gap regime \cite{Aizenman_Graf_1998}. These localized states are however essential in order to explain many important features of topological insulators, e.g., why plateaus emerge in the integer quantum Hall effect; see \cite{EGS_2005,Graf_Shapiro_2018_1D_Chiral_BEC,Shapiro2019,Shapiro20,BSS23} for further discussion of the mobility gap regime. In the spectral gap regime, the Fermi projection is a continuous function of the Hamiltonian and thus belongs to the C-star algebra generated by it. This makes the spectral gap regime amenable to analysis by K-theory of C-star algebras. On the other hand, in the mobility gap regime, the Fermi projection is merely a measurable function of the Hamiltonian.
	
	It is mainly for the study of the mobility gap problem that it is important to be able to build alternative perspectives to the classification problem that do not rely on algebraic topology of classical manifolds (as in the periodic case) or on K-theory of C-star algebras (as in the disordered but spectrally gapped case), and this is the main point of the present paper: we present a first K-theoretic-free classification of disordered materials to our knowledge. Moreover, in our approach the question of ``which topology to define on the space of `insulators'\thinspace'' becomes explicit and is brought to the foreground, since without it one cannot even start the analysis.
	
	It remains unclear just what physical (better yet, experimentally relevant) role the choice of this topology bears, and it is also interesting to ask whether this choice is necessarily unique (we presume it is not). Be that as it may, since topological insulators are presumed useful for quantum computing \cite{KempeKitaevRegev2006}, where it is precisely the topological stability properties that lend themselves to be of great utility, it seems that exploring the foundations and boundaries of these stability properties could maybe help answer edge cases of quantum engineering problems.
	
	In this paper we build the first step of this research program, which is the most straight-forward, namely, understanding non-interacting one-dimensional spectrally-gapped disordered systems via homotopies and without K-theory. This has the appeal that it is simpler--though this is a matter of taste--than the existing K-theoretic classifications, but also, that it allows us to start working on the next steps in the aforementioned program:
	\begin{enumerate}
		\item Higher dimensions in the spectral-gap, non-interacting case, and a more detailed study of higher dimensional locality (see \cref{sec:higher odd d}).
		\item The strongly-disordered mobility gap regime (see \cref{sec:mobility gap}).
		\item The interacting case (and within it the fractional quantum Hall effect).
		\item Understanding interactions in the strongly-disordered regime, and hence also many-body localization (MBL).
	\end{enumerate}
	It is mainly the second item which we feel is amenable to the methods developed here.
	
	Let us briefly describe the mathematical novelty of this paper, to be presented in \cref{sec:local unitaries classification,sec:local SA unitaries classification}. Quantum mechanical Hamiltonians, beyond being self-adjoint, must obey a certain kind of \emph{locality constraint} which is central in the present paper. Indeed it is that constraint which elevates the analysis from pure functional analysis into physics. This constraint roughly corresponds to the fact that there is no action at a distance. Geometrically this can be understood as a non-commutative analog of a regularity constraint on symbols, since, if our systems were translation invariant it would correspond to continuity of the symbol via a Riemann-Lebesgue lemma. Hence we are concerned with spaces of \emph{local} operators. Under the various symmetry constraints these operators break down into two main classes depending on the presence or absence of a so-called \emph{chiral symmetry}: unitaries or self-adjoint projections. These two broad categories are then broken into five additional ones: complex, real, quaternionic, and so-called $\star$-real or $\star$-quaternionic (see \cref{sec:locality and symmetry}). Hence all together we find ten possible classes. Let us consider, then, the simplest case: that of complex unitaries. Without the locality constraint, it is a result that goes back to Kuiper \cite{Kuiper1965} (see \cref{thm:Kuiper} below) that the set of unitaries on a separable Hilbert space is path-connected. Indeed, a path from any unitary $U$ to $\Id$ is given by \eql{\label{eq:Kuiper's path} [0,1]\ni t\mapsto \exp\left(\ii (1-t)\left(-\ii \log\left(U\right)\right)\right) } where $-\ii\log\left(U\right)$ is a self-adjoint operator to be understood via the bounded measurable functional calculus of normal operators. In contrast to Kuiper's situation, the space of \emph{local} unitaries turns out to be very much disconnected: the components are indexed by a non-commutative analog of the winding number, which under the assumption of translation invariance indeed collapses to the classical winding number (this is the Krein-Widom-Devinatz theorem \cite[pp. 185]{Douglas1998}). The winding number requires the continuity of the map to be meaningful, which is analogous to the present locality constraint. The main issue to be dealt with is, then: given two local unitaries $U,V$ of the same index, construct a continuous \emph{local} path between them, or equivalently, given a local unitary of zero index, connect it locally to $\Id$. It is a theorem that if a local unitary has non-zero index then its spectrum is the whole $\mathbb{S}^1$ \cite{AschBourgetJoye2020}. Naively one might expect that unitaries of zero index always have a spectral gap on $\mathbb{S}^1$ and hence the above logarithm may actually be interpreted via the \emph{holomorphic} functional calculus, in which case it preserves locality (this is the Combes-Thomas estimate for unitaries, see e.g. \cite{hamza2009dynamical,Shapiro2019}). This is unfortunately false: take as a counter-example any continuous map $\bS^1\to\bS^1$ which has zero winding number but whose range is $\bS^1$. Its Fourier series will correspond to a local unitary of zero index which has $\sigma(U)=\bS^1$. The solution is then to factorize $U=AB$ where $A,B$ are two local unitaries, one of which has a gap and the other diagonal in a left-right decomposition of the Hilbert space, and is hence amenable to a (not necessarily local) usual Kuiper path on each side of space separately. The homotopies of local complex unitaries were first studied in \cite{CareyHurstOBrien1982}, although there a different proof was presented. The non-complex local unitary homotopies are, to our knowledge, new. For self-adjoint projections the local homotopies are somewhat different; to this end we make equivariant extensions of the work of \cite{AndruchowChiumientoLucero2015}. It turns out that in the complex case, all self-adjoint local projections of a certain non-trivial class are path-connected.
	
	In two of the symmetry classes, the index is $\ZZ_2$-valued, corresponding to the Atiyah-Singer skew-adjoint Fredholm index \cite{AtiyahSinger1969}, see \cref{sec:Atiyah-Singer Z2 index theory} for an introduction. For these symmetry classes, the analysis becomes more complicated due to the absence of a logarithmic law for the $\ZZ_2$ index, leading us to connect directly two arbitrary operators of odd index. The application of Atiyah and Singer's skew-adjoint Fredholm index in the context of topological insulators was pioneered in \cite{Schulz-Baldes2015} but then studied also in \cite{KatsuraKoma2016,FSSWY22,BSS23}.
	
	In regards to existing literature, almost exclusively, classification results of topological insulators rely on K-theory and it is in this sense that they do not provide a complete homotopy classification. Of the ones listed in the first paragraph above, we mention the paper by Thiang \cite{Thiang2016} who provides a K-theoretic classification of disordered spectrally gapped systems in all dimensions. On a more pedestrian note, if one assumes translation invariance, the classification problem is of course classical and reduces to studying homotopies of continuous maps $\mathbb{T}^d\to\mathrm{Gr}_k(\CC^N)$ under various symmetry constraints where $\mathrm{Gr}_k(\CC^N)$ is the Grassmannian: the space of $k$-dimensional subspaces within $\CC^N$. This classification is in fact known to ``contradict'' \cref{table:Kitaev} due to: (1) low $N$ problems, and (2) the existence of \emph{weak topological invariants}. These are, roughly speaking, indices which do not explore all $d$ dimensions of real space and are not stable under strong disorder. Recently Avron and Turner \cite{AvronTurner2022} presented a full classification of these translation invariant systems in the special case $d=k=1=N/2$.
	
	This paper is organized as follows. In \cref{sec:locality and symmetry} we present the abstract mathematical setting of odd-dimensional locality, symmetries and the associated indices. This section is mainly intended to set up the terminology and notation for \cref{sec:local unitaries classification,sec:local SA unitaries classification} in which we calculate $\pi_0$ of various symmetry-constrained local unitaries and self-adjoint projections. We make use of this theory in \cref{sec:bulk insulators classification} by connecting it to the problem of classifying bulk one-dimensional spectrally-gapped insulators. Within this section, we single out \cref{subsec:exp locality 1d chiral classification} where operators with the more common form of \emph{exponential locality} are studied using an entirely separate scheme. After making some brief remarks about edge systems in \cref{sec:1d edge classification}, we present a negative result about the classification in the mobility gap result in \cref{sec:mobility gap} and conclude in \cref{sec:higher odd d} with a few words about the classification problem in higher dimensions. We shall argue there that even though in some sense one may wish to draw conclusions from our work on the classification problem in all odd dimensions, the notion of locality we employ here and which makes sense in one-dimension, is rather unsatisfactory in higher dimensions, which warrants that not only the even-dimensional but also the odd-dimensional problem be revisited in future work.
	
	\paragraph{Notations and conventions}\begin{itemize}
		\item $\pi_0$ is the path-components functor acting on the category of topological spaces.
		\item We use $|A|^2\equiv A^\ast A$ and $\polar(A)$ for the polar part in the polar decomposition $A=\polar(A)|A|$, made unique by the convention that $\ker U = \ker A$.
		\item $\calH$ is a separable Hilbert space, $\calB$ is its Banach algebra of bounded linear operators, and $\calU,\calG,\calF,\calK$ are the subspaces of unitary, invertible, Fredholm, and compact operators respectively. We shall also use the space $\calP$ of self-adjoint (orthogonal) projections and (the equivalent) $\calSU$, the space of self-adjoint unitary operators. Sometimes we also use $\calF^{\mathrm{sa}}$ for the space of self-adjoint Fredholm operators.
		\item For us $A$ is idempotent iff $A^2=A$ and $A$ is a self-adjoint (orthogonal) projection iff $A^2=A=A^\ast$. We generally try to avoid the term \quotes{projection} by itself since some authors use it for idempotent and others for \quotes{self-adjoint projection}.
		\item $\calL$ is the C-star algebra of local operators, those operator having a compact commutator with a fixed projection $\Lambda$.
		\item $C$ is a real structure (anti-unitary) which squares to $+\Id$ and $J$ is a quaternionic structure (anti-unitary) which squares to $-\Id$.
		\item By the word ``essentially'' we generally mean that an algebraic condition holds up to compacts, i.e., in the Calkin algebra. With this, we have essentially unitary operators ($\Id-|A|^2,\Id-|A^\ast|^2\in\calK$), essentially projections ($A-A^\ast,A-A^2\in\calK$), etc.
		\item We use $\bS^1$ for the unit circle and $\DD$ for the open unit disc, both understood as subsets of $\CC$.
		
		
	\end{itemize}
	
	\section{Abstract locality, indices and symmetry constraints}\label{sec:locality and symmetry}
	In this section, $\calH$ is some fixed abstract separable Hilbert space.
	
	\begin{defn}[non-trivial projections]\label{def:non trivial projections} We call a self-adjoint projection $P\in\calP$ \emph{non-trivial} iff its range and kernel are both infinite dimensional.
	\end{defn}
	Essential projections are classical objects which go back to \cite{Calkin1941}: $A\in\calB$ is called essentially a projection iff $A-A^\ast,A-A^2\in\calK$. Actually this implies, via \cref{lem:essential projection is compact away from genuine projection}, that there exists some $P\in\calP$ such that $A-P\in\calK$. Less common is the notion of essentially non-trivial projections:
	\begin{defn}[essentially a non-trivial projection]
		We call a bounded linear operator $A\in\calB$ ``essentially a non-trivial projection'' iff $A$ is essentially a projection and $\Set{0,1}\subseteq\sigma_{\mathrm{ess}}(A)$. 
	\end{defn}
	It will be useful to have another criterion for essentially non-trivial projections:
	\begin{lem}\label{lem:stability of essentially non trivial projections}
		$A\in\calB$ is essentially a non-trivial projection iff there exists a non-trivial projection $P$ such that $A-P\in\calK$. Furthermore, if $B\in\calB$ is essentially a projection and $A-B$ is sufficiently small in norm or compact, then $B$ is also essentially a non-trivial projection.
	\end{lem}
	\begin{proof}
		If there exists a non-trivial projection $P$ such that $A-P\in\calK$, write $A-P=K$, then $A-A^*=K-K^*\in\calK$ and $A-A^2=P+K-(P+K)^2=K-PK-KP-K^2\in\calK$ and hence $A$ is essentially a projection. In particular, $\Set{0,1}\subseteq \sigma_{\mathrm{ess}}(P)=\sigma_{\mathrm{ess}}(P+K)=\sigma_{\mathrm{ess}}(A)$. Thus $A$ is essentially a non-trivial projection.
		
		Conversely, if $A$ is essentially a projection, using \cref{lem:essential projection is compact away from genuine projection}, there exists a self-adjoint projection $P$ such that $A-P\in\calK$. In particular, $\Set{0,1}\subseteq \sigma_{\mathrm{ess}}(A)=\sigma_{\mathrm{ess}}(P+K)=\sigma_{\mathrm{ess}}(P)$, so the kernel and image of $P$ are infinite and thus $P$ is non-trivial.
		
		Now, if $A-B\in\calK$ the statement is trivial. Moreover, we note that if $A$ is essentially a projection, then, it is non-trivial iff $A,\Id-A$ are not Fredholm (by the Fredholm definition of the essential spectrum). Hence, if $\norm{A-B}$ is sufficiently small, it can't be that $B$ is Fredholm whereas $A$ is not, since the Fredholms are open, and same with $(\Id-A)-(\Id-B)$. Thus $\Set{0,1}\in \sigma_{\mathrm{ess}}(B)$.
	\end{proof}
	
	\begin{rem}\label{rem:connection between projections and self-adjoint unitaries}Given a projection $P\in\calP$, $\Id-2P$ is a self-adjoint unitary, so that the space $\calP$ is identified with \eq{\calSU\equiv\Set{U\in\calU|U=U^\ast}\,,} the space of self-adjoint unitaries, and all the notions discussed above of non-triviality of projections carry over to self-adjoint unitaries. We shall refer to both spaces interchangeably, and the classification of local self-adjoint projections or local self-adjoint unitaries is the same.
	\end{rem}

	\begin{defn}[$\Lambda$-local operators]\label{def:Lambda-local operators} For a fixed non-trivial self-adjoint projection $\Lambda$, an operator $A\in\calB$ is termed \emph{$\Lambda$-local} iff it essentially commutes with $\Lambda$, i.e., \eql{[\Lambda,A]\equiv\Lambda A-A\Lambda\in\calK\,.} The space of all local operators is denoted by $\calL_\Lambda$. Clearly if a projection is trivial, the condition is vacuous, and hence the restriction. Sometimes we use the phrase \emph{hyper-local} if $[A,\Lambda]=0$.
	\end{defn} 
	
	Unless otherwise specified (mainly relevant in \cref{sec:mobility gap}) we shall always use the subspace topology induced by the operator norm topology on $\calB$ unless otherwise specified. With respect to this topology, we use $\pi_0$ as the path-components functor. 
	
	For most of what follows, we shall not have occasion to consider different $\Lambda$'s for locality, and so, let us fix once and for all one self-adjoint projection $\Lambda$ and omit this choice entirely from the notation. If a space $\calA$ carries the superscript $\calL$ we mean by it the intersection: \eql{\calA^\calL \equiv \calA\cap\calL} and the prefix $\calS$ means the subset of self-adjoint operators within $\calA$.
	
	\begin{lem}
		$\calL$ is a C-star algebra with respect to the operator norm and adjoint inherited from $\calB$.
	\end{lem}
	\begin{proof}
		The only thing to verify is the compact commutator condition is closed. However, the norm limit of compact operators is compact, and hence the statement follows.
	\end{proof}
	\begin{rem}
		To the extent that commutators may be considered as non-commutative derivatives, locality may be thought of as a certain regularity condition analogous to differentiability. This is essentially Bellissard et al's non-commutative Sobolev spaces \cite{Bellissard_1994JMP....35.5373B}.
	\end{rem}

	\begin{lem}\label{lem:continuous functional calculus of normals is local}The continuous functional calculus on normal operators maps $\calL$ to $\calL$.
	\end{lem}
	\begin{proof}
		Let $A\in\calL$ be normal and $f:\CC\to\CC$ continuous. Since $A$ is bounded, its spectrum is restricted to some compact set $S\subseteq\CC$ and hence we may assume WLOG that $f$ has support $S$. Let now $\Set{p_k:\CC\to\CC}_k$ be a sequence polynomials converging \emph{uniformly} to $f$ on $S$. Then $p_k(A)\to f(A)$ in operator norm, and hence, since each $[p_k(A),\Lambda]$ is compact (recall $\calL$ is a C-star algebra) its norm limit is too. 
	\end{proof}
	We note in passing that for holomorphic functions (which may be desired when dealing with non-normal operators) this can be deduced by a Combes-Thomas type argument: the resolvent of a local operator is clearly local by $[\Lambda,(A-z\Id)^{-1}]=-(A-z\Id)^{-1}[\Lambda,A](A-z\Id)^{-1}$.

	We now define the so-called \quotes{super} operator $\bbLambda:\calB\to\calB$ given by \eql{\calB\ni A \mapsto \bbLambda A \equiv \Lambda A \Lambda + \Id-\Lambda\,.} With it we may define an index for local unitaries as follows.
	\begin{lem}[Fredholm property of local unitaries] The image of $\calU^\calL$ under $\bbLambda$ is Fredholm, i.e., \eq{\bbLambda(\calU^\calL)\subseteq\calF\,.} 
	\end{lem}
	\begin{proof}
		Let $U\in\calU^\calL$. Then using Atkinson's theorem \cite{Booss_Topology_and_Analysis} it suffices to exhibit $(\bbLambda U)^\ast$ as the parametrix of $\bbLambda U$, and to that end, we note that \eq{ \Id-(\bbLambda U)^\ast(\bbLambda U)=\Id-(\bbLambda U^\ast)(\bbLambda U) = \Lambda(\Id-U^\ast\Lambda U)\Lambda 
			= \Lambda U^\ast \Lambda^\perp U\Lambda 
			= \Lambda U^\ast \Lambda^\perp [U,\Lambda]\,.
		} Now, since $U\in\calL$ this last commutator is compact, and so by the ideal property of $\calK$ the entire expression.
	\end{proof}
	
	It thus makes sense to define $ \findex_\Lambda : \calU^\calL \to \ZZ$ via \eql{\findex_\Lambda(U) := \findex(\bbLambda U)\equiv\dim\ker \bbLambda U - \dim\ker \bbLambda U^\ast\label{eq:index of local unitary}\,.} This index reduces to the winding number, if the unitary happens to be a Toeplitz operator on $\ell^2(\ZZ)$ (this statement is the aforementioned Krein-Widom-Devinatz theorem \cite[pp. 185]{Douglas1998}, sometimes also referred to as the Krein-Gohberg theorem).
	It is comforting to know that this index inherits the logarithm law from the Fredholm index:
	\begin{lem}[logarithmic law]\label{lem:log law for index_Lam} If $U,V\in\calU^\calL$ then \eql{\findex_\Lambda U + \findex_\Lambda V = \findex_\Lambda (UV)}
	\end{lem}
	\begin{proof}
		Using the logarithmic law of $\findex:\calF\to\ZZ$ \cite{Booss_Topology_and_Analysis}, it remains to show $\bbLambda(UV)-(\bbLambda U)\bbLambda V$ is compact. This follows from \eq{\bbLambda(UV)-(\bbLambda U)\bbLambda V
			=\Lambda UV\Lambda-\Lambda U\Lambda V\Lambda
			=\Lambda U\Lambda^\perp V\Lambda
			=\underbrace{[\Lambda, U]}_{\in\calK}\Lambda^\perp V\Lambda\,.
		}
	\end{proof}

	We now turn to symmetry constraints. Let $C,J:\calH\to\calH$ be two fixed \emph{anti-unitary} operators on $\calH$ such that \eq{C^2=\Id\,,\qquad\,J^2=-\Id.}
	As such, $C$ and $J$ define \emph{real} and \emph{quaternionic} structures respectively on $\calH$: $C$ should be understood as complex conjugation and $J$ as the $j$th quaternionic basis vector, so that $\Id,\ii\Id,J$ and $\ii J$ build the quaternionic basis vectors \cite{Baez2012}. It is thus natural to consider the subspace of real and quaternionic bounded operators, those which respect that structure: \eql{\label{eq:real and quaternionic operators} \calB_\RR := \Set{A\in\calB | AC=CA} \,,\qquad \calB_\HH := \Set{A\in\calB | AJ=JA}\,. }
	We note that in the latter case, unitary operators $U$ obeying $[U,J]=0$ may also be understood as \emph{Hermitian}-symplectic operators (discussed e.g. in \cite[(3.7)]{Shapiro2021}) with respect to the symplectic bi-linear form given by $\ip{\cdot}{J\cdot}$, since then one has $U^\ast J U = J$ and hence the bi linear form $\ip{\cdot}{J\cdot}$ is preserved by such $U$.
	
	We shall also need the following somewhat more exotic symmetry constraints. For lack of better terminology, we call them $\star$-real and $\star$-quaternionic operators: \eql{\label{eq:star real and quaternionic operators} \calB_{\star\RR} := \Set{A\in\calB | AC=CA^\ast} \,,\qquad \calB_{\star\HH} := \Set{A\in\calB | AJ=JA^\ast}\,. }
	In \cite{FSSWY22} we used the terminology $J$-odd for the same constraint (only $\calB_{\star\HH}$ was used there), but in the current abstract mathematical setting it is more natural to use the real and quaternionic structures. We caution the reader that our naming is not standard, e.g., in \cite{garcia2006complex} the name \emph{$C$-symmetric} was used for $\calB_{\star\RR}$.
	
	The following purely imaginary classes are not independent of the ones presented so far, but we introduce them separately nonetheless for notational simplicity: \eql{\label{eq:pure imaginary real or quaternionic operators} \calB_{\ii\RR} := \Set{A\in\calB | AC=-CA} \,,\qquad \calB_{\ii\HH} := \Set{A\in\calB | AJ=-JA}\,. } They may be obtained as $\ii \calB_\RR,\ii\calB_\HH$ respectively.
	
	We shall see below in \cref{sec:bulk insulators classification} that these combinations build together all the necessary Altland-Zirnbauer symmetry classes (the ten fold way) which appear in \cref{table:Kitaev}.
	
	\begin{assumption}[real and quaternionic structures are hyper-local]\label{ass:real and quaternionic structures are hyperlocal}
		We shall assume that $C,J$ are chosen so that \eql{[J,\Lambda]=[C,\Lambda]=0\,.}
		This can probably be weakened to from zero to compact, but we do not need this generalization.
	\end{assumption}

	It is then clear that, for $\FF=\RR,\HH$, restricting $\findex_\Lambda$ to $\calU_{\star\FF}^\calL$, we get the constant zero map. Indeed, this is immediate from the fact $C,J$ are bijections and the logarithmic rule \cref{lem:log law for index_Lam}. The same is true within $\calSU^{\calL}_{\FF}$ for any $\FF$ by self-adjointness. Be that as it may, Atiyah and Singer recognized that another index, a $\ZZ_2$ index, may sometimes be defined (see \cref{sec:Atiyah-Singer Z2 index theory} below): \eql{\findex_{\Lambda,2}(U) := \findex_{2}\bbLambda U \equiv \left(\dim \ker \bbLambda U \mod 2\right) \in \ZZ_2} where $\findex_{2}:\calF\to\ZZ_2$ is the Atiyah-Singer $\ZZ_2$ Fredholm index. As discussed in \cref{thm:ASZ2 index,thm:ASZ2 index v2}, this index is norm continuous as a map with domain $\calU_{\star\HH}^\calL$ or $\calSU^{\calL}_{\ii\RR}$ respectively.
	
	\section{Equivariant classification of local unitaries}\label{sec:local unitaries classification}
	In this section we shall study $\pi_0(\calU_\FF^\calL)$ where $\FF$ is either $\CC$ (in which case this is just the space of local unitaries) or $\FF$ is one of the four symmetries discussed above: $\RR,\HH,\star\RR,\star\HH$. We group our theorems together based on method of proof. The results are summarized in \cref{table:equivariant local unitaries}.
	
	We start with the main classification statement:
	\begin{thm}[Classification of $\RR,\CC$ and $\HH$ local unitaries]\label{thm:classification of RCH local unitaries} For $\FF=\RR,\CC$ the map $\findex_\Lambda:U_\FF^\calL\to\ZZ$ is norm continuous and ascends to a bijection \eql{\findex_\Lambda:\pi_0(\calU_\FF^\calL)\xrightarrow{\sim}\ZZ}  and analogously for the quaternionic class, we have the bijection\eql{\findex_\Lambda:\pi_0(\calU_\HH^\calL)\xrightarrow{\sim}2\ZZ\,.}
	\end{thm}
	This theorem should be compared with Kuiper's theorem ($\pi_0(\calU)\cong\Set{0}$, see \cref{thm:Kuiper}) and the Atiyah-J\"anich theorem ($\findex:\pi_0(\calF)\cong\ZZ$) \cite{Booss_Topology_and_Analysis}. Strictly speaking, when $\FF=\CC$, it is not new: it may be deduced from the results of \cite{CareyHurstOBrien1982}, where the criterion of locality as a compact commutator is replaced by the commutator belonging to a more general ideal.
	We shall present a different proof, which also covers the cases $\FF=\RR,\HH$ (which as far as we are aware has not appeared previously). We also became aware that \cite{Geib2022} contains ideas of similar spirit.
	
	Next, we have the nullhomotopic result:
	\begin{thm}[Classification of $\star\RR$-local unitaries]\label{thm:star-real local unitaries} The space of $\star\RR$-local unitaries is null-homotopic: \eql{\pi_0(\calU_{\star\RR}^\calL)\cong\Set{0}\,.}
	\end{thm}
	
	Finally, there is the $\ZZ_2$ classification:
	\begin{thm}[Classification of $\star\HH$-local unitaries]\label{thm:star-quaternionic local unitaries} The space of $\star\HH$-local unitaries has two path-components. The map $\findex_{\Lambda,2}:\calU_{\star\HH}^\calL\to\ZZ_2$ is norm continuous and ascends to a bijection \eql{\findex_{\Lambda,2}:\pi_0(\calU_{\star\HH}^\calL)\xrightarrow{\sim}\ZZ_2\,.}
	\end{thm}
	\begin{table}
		\begin{center}
			\begin{tabular}{|c|c|}
				\hline
				$\pi_0(\calU^\calL)$ & $\ZZ$\\\hline
				$\pi_0(\calU_\RR^\calL)$ & $\ZZ$\\\hline
				$\pi_0(\calU_\HH^\calL)$ & $2\ZZ$\\\hline
				$\pi_0(\calU_{\star\RR}^\calL)$ & $\Set{0}$\\\hline
				$\pi_0(\calU_{\star\HH}^\calL)$ & $\ZZ_2$\\
				\hline
			\end{tabular}
		\end{center}
		\caption{The classification of equivariant local unitaries.}
		\label{table:equivariant local unitaries}
	\end{table}
	
	The main technical tool to be used in \cref{thm:classification of RCH local unitaries,thm:star-quaternionic local unitaries,thm:star-real local unitaries} is a factorization principle, which we present and prove before tending to the proofs of the main theorems.
	\begin{lem}[factorization of local unitaries]\label{lem:factorization for local unitaries} For any ${\FF\in\Set{\CC,\RR,\HH}}$, let $U\in\calU_\FF^\calL$ such that \eq{\findex_\Lambda U = 0\,.} Then there exist two unitaries $A,B\in\calU_\FF^\calL$ such that $\Id-A\in\calK$, $[\Lambda,B]=0$ and such that \eql{ U = AB\,.} 
	\end{lem}
	\begin{proof}
		Let $\FF\in\Set{\CC,\RR,\HH}$, and $F\in\Set{\Id,C,J}$ accordingly. Let us decompose $U$ in $\calH=(\im\Lambda)^\perp\oplus \im \Lambda $ as \eql{\label{eq:Lambda decomposition of U}U=\begin{bmatrix}U_{LL} & U_{LR} \\ U_{RL} & U_{RR}\end{bmatrix}\,.} 
		
		$U\in\calU^\calL$ implies both $U_{LR},U_{RL}\in\calK$, and $U_{ii}$ is \emph{essentially unitary}: $\Id -|U_{ii}|^2,\Id-|U_{ii}^\ast|^2\in\calK$ (see \cite{Murphy1990}). By \cref{ass:real and quaternionic structures are hyperlocal}, the subspaces $\im\Lambda,(\im\Lambda)^\perp$ are invariant under the action of $\FF$-structures, i.e., $F$ is diagonal in this grading and (hopefully without confusion) we do not give each block a separate symbol. Thus $UF=FU$ implies $U_{ii}F=FU_{ii}$ for $i\in\Set{L,R}$. Now $\findex_\Lambda U=0$ is equivalent to $\findex U_{RR}=0$.  Using $\findex A\oplus B=\findex A+\findex B$, the fact that $U$ is unitary, and that $\findex$ is invariant under compact perturbations, we have \eq{0=\findex U=\findex U_{LL}\oplus U_{RR}=\findex U_{LL} +\findex U_{RR}=\findex U_{LL}} so that $\findex U_{LL}=0$ as well. Hence $U_{ii}$ is an essential unitary Fredholm operator of zero index, so that applying \cref{lem:extension of zero index Fredholm RCH partial isometries to unitaries} below on $U_{ii}$ we obtain some $B_{ii}\in\calU_\FF$ which differs from $U_{ii}$ by a compact; we point out $B_{ii}$ is a unitary on one of the smaller spaces $\im\Lambda$,$\im\Lambda^\perp$.
		
		Let $B=B_{LL}\oplus B_{RR}$. Define $A=UB^*$ from which $A\in\calU_\FF^\calL$ follows. To see that $\Id-A\in\calK$, write $U_{ii}=B_{ii}+K_{ii}$ for some $K_{ii}\in\calK$. We have \eql{A =\begin{bmatrix}B_{LL}+K_{LL} & U_{LR} \\ U_{RL} & B_{RR}+K_{RR}\end{bmatrix}\begin{bmatrix}B_{LL}^* & 0 \\ 0 & B_{RR}^*\end{bmatrix} = \begin{bmatrix}\Id + K_{LL}B_{LL}^* & U_{LR}B^*_{RR} \\ U_{RL}B^*_{LL} & \Id + K_{RR}B^*_{RR}\end{bmatrix} \label{eq:A is compact perturbation from identity}\,.}
	\end{proof}

	For the $\star$ classes, we need an adjusted factorization statement:  
	\begin{lem}[factorization of local $\star$-unitaries]\label{lem:factorization for star RH local unitaries} For any ${\FF\in\Set{\star\RR,\star\HH}}$ and $F\in\Set{C,J}$ accordingly, let $U\in\calU_\FF^\calL$. If $\FF=\star\HH$, we furthermore assume \eq{\findex_{\Lambda,2} U = 0\,.} Then there exist two unitaries $A,B$ with \eq{ U = AB } such that $B\in\calU^\calL_\FF$ with $[\Lambda,B]=0$, and such that $A$ is local, has $\Id-A\in\calK$, and is furthermore $\star\FF$ with respect to $\widetilde{F}:=BF$ (i.e., $A(BF)=(BF)A^*$).
	\end{lem}
	\begin{proof}
		Let us decompose $U$ in $\calH=(\im\Lambda)^\perp\oplus \im \Lambda $ as in \cref{eq:Lambda decomposition of U} with the properties of $U_{ij}$ listed there. Using \cref{lem:extension of star fredholm partial isometries} below, for $i=L,R$, we have $B_{ii}\in\calU_\FF$ whose difference from $U_{ii}$ is compact. This is justified because $\findex_{\Lambda,2} U_{ii} = 0$. Indeed, by hypothesis, this holds for $i=R$, and since $U$ is unitary, $\findex_{2} U = 0$. Now, using \cref{cor:compact perturbation preserve Z2 index}, we conclude that $\findex_{2} U_{LL}\oplus U_{RR} = 0$. So $\findex_{2} U_{LL} = 0$ too.
		
		Let $B=B_{LL}\oplus B_{RR}\in\calU_F^\calL$ and $A=UB^*\in \calU^\calL$. We have $\Id-A\in\calK$ from \cref{eq:A is compact perturbation from identity}. However, now $\calB_\FF$ is not an algebra. Nonetheless, observe that $(BF)^2=BFBF=FB^*BF=F^2$ so that $\ti{F}:=BF$ (instead of $F$) defines an $\FF$-structure with which $A$ is $\star$-real or quaternionic: \eq{A(BF)=UF=F U^*=F B^*A^*= (BF) A^*\,.}
	\end{proof}

	Using the factorization lemmas, we may tend to our three theorems. In regards to the continuity of $\findex_{(2),\Lambda}\equiv\findex_{(2)}\circ\bbLambda$, it is a consequence of the norm continuity of $\findex_{(2)}:\calF\to\ZZ_{(2)}$ and the trivial fact that $\bbLambda:\calB\to\calB$ is continuous. This statement is true regardless of $\FF$. 
	
	So we merely need to show surjectivity (when applicable), well-definedness (when applicable) and injectivity.

	\begin{proof}[Proof of \cref{thm:classification of RCH local unitaries}] 
		
		We start with surjectivity for $\FF=\CC$. Since $\Lambda$ is non-trivial, there is an ONB for $\calH$, $\Set{\vf_n}_{n\in\ZZ}$, such that $\Set{\vf_n}_{n\leq0}$ spans $\ker\Lambda$ and $\Set{\vf_n}_{n>0}$ spans $\im\Lambda$. Define a unitary operator $R$ on $\calH$ via \eq{R\vf_n := \vf_{n+1}\qquad(n\in\ZZ)\,.} Since $\ip{\vf_n}{[R,\Lambda]\vf_m}=\left(\chi_\NN(n)-\chi_\NN(m)\right)\delta_{n,m+1}$ this commutator is trace-class and hence compact, so that $R\in\calU^\calL$. Moreover, clearly  \eq{\findex_\Lambda R^k = -k\qquad(k\in\ZZ)} so that $\findex_\Lambda:\calU^\calL\to\ZZ$ is surjective.

		We turn to surjectivity in the case  $\FF=\RR$. We can use the same proof as above, being careful to make sure that the shift operator $R$ we define obeys $[R,C]=0$. To do so, using the fact $[\Lambda,C]=0$, we may let $C$ act on each subspace $\ker\Lambda,\im\Lambda$ separately, and we invoke \cref{lem:basis in RCH symmetry} on each to obtain an ONB for each subspace, $\Set{\vf_n}_{n\leq0}$ and $\Set{\vf_n}_{n>0}$ respectively, which is fixed by $C$. Defining again $R\vf_i:=\vf_{i+1}$ then does the job.
		
		We next turn to the case $\FF=\HH$. Let us begin with well-definededness of the index, i.e., we show that $\findex_\Lambda(U)\in 2\ZZ$ for all $U\in\calU_\HH^\calL$. For such a $U$, since $[J,\Lambda]=0$, we may let $J$ act separately on $\ker\Lambda,\im\Lambda$. Moreover, $[U,J]=0$ implies $[U_L,J]=0$ where $U_L\equiv\Lambda^\perp U \Lambda^\perp:\ker\Lambda\to\ker\Lambda$. As such, we have $J\ker U_L=\ker U_L$ and $J \ker U_L^\ast = \ker U_L^\ast$. As a result, these two spaces are even dimensional thanks to \cref{lem:basis in RCH symmetry} and hence the index is even.
		
		Finally we turn to the surjectivity of $\findex_\Lambda:\calU^\calL_\HH\to2\ZZ$. Since $[J,\Lambda]=0$, we may let $J$ act separately on $\ker\Lambda,\im\Lambda$ and invoke \cref{lem:basis in RCH symmetry} separately on them to obtain two orthonormal bases \eql{\label{eq:J equivariant basis for ker im Lambda}\Set{\vf_n,\psi_n}_{n\leq0},\Set{\vf_n,\psi_n}_{n>0}} for $\ker\Lambda,\im\Lambda$ respectively, which obey the property that $J\vf_n=\psi_n$. We then define a unitary operator $R$ on $\calH$ via \eq{\vf_n\mapsto\vf_{n+1}\qquad\psi_n\mapsto\psi_{n+1}\qquad(n\in\ZZ)} and note that \eq{RJ\varphi_i=R\psi_i=\psi_{i+1}=J\varphi_{i+1}=JR\varphi_i \\ RJ\psi_i=-R\varphi_i=-\varphi_{i+1}=J\psi_{i+1}=JR\psi_i\,.} Thus $RJ=JR$ so that $R\in\calU^\calL_\HH$ and $\findex_\Lambda R^k=-2k$ for all $k\in\ZZ$.
		
		We are left with establishing injectivity, which we can do for all three  $\FF=\CC,\RR,\HH$. It is tantamount to the following statement: given any $U,V\in\calU_\FF^\calL$ such that $\findex_\Lambda U = \findex_\Lambda V$, there exists a (norm) continuous path \eq{ \gamma:[0,1]\to \calU_\FF^\calL} such that $\gamma(0)=U$ and $\gamma(1)=V$. Thanks to \cref{lem:log law for index_Lam} we may WLOG assume that $V=\Id$ and hence that $\findex_\Lambda U =0$. But then an application of \cref{lem:factorization for local unitaries} on $U$ yields \eq{ U=AB } for $A,B\in\calU_\FF^\calL$ with $\Id-A\in\calK$ and $[\Lambda,B]=0$. This first of all implies that $A$ can only have accumulation of spectrum at $+1\in\bS^1$. Now the analysis divides according to the value of $\FF$. In the simplest case, if $\FF=\CC$, let $\alpha\in\bS^1\setminus\sigma(A)$. Then $\log_\alpha$, which is the \emph{holomorphic} logarithm with branch cut at $\alpha$, defines a \emph{local} self-adjoint operator $-\ii \log_\alpha(U)$. With that, Kuiper's path \cref{eq:Kuiper's path} passes within $\calU^\calL$ thanks to \cref{lem:continuous functional calculus of normals is local}. For $\FF=\RR,\HH$, the claim is shown via \cref{lem:deform A to Id obeying RCH symmetry} right below.
		
		Next, since $[B,\Lambda]=0$, we may write $B = B_L \oplus B_R$ in the grading $\calH =(\im\Lambda)^\perp\oplus\im\Lambda$, for two unitaries $B_L,B_R\in\calU_\FF$. Now \cref{thm:Kuiper} guarantees paths $B_L,B_R\to\Id$ which pass within $\calU_\FF$. Taking the direct sum of these two paths we obtain a diagonal, and hence local path $B\to\Id$ within $\calU_\FF^\calL$. 
		
		Combining the two separate paths from either $A,B$ to $\Id$ by multiplication yields the desired path.
	\end{proof}
	
	The following result which was used just above shows that when a local unitary has a gap, it may be deformed to the identity in a local way, i.e., Kuiper's path \cref{eq:Kuiper's path} may be taken as local using \cref{lem:continuous functional calculus of normals is local}. Next, we establish this also in the presence of symmetries:
	\begin{lem}\label{lem:deform A to Id obeying RCH symmetry}
		Let $\FF\in\Set{\RR,\HH}$ and $A\in\calU_\FF^\calL$ such that $\Id-A\in\calK$. Then there exists a continuous path $[0,1]\ni t\mapsto A_t\in\calU_\FF^\calL$ such that $A_0=A$ and $A_1=\Id$.
	\end{lem}
	\begin{proof}
		Since $\Id-A\in\calK$, $\sigma(A)\neq\bS^1$, and in fact, all the spectrum of $A$ outside of $1$ consists of finitely degenerate eigenvalues.
		
		Suppose $-1\notin\sigma(A)$, then there is a gap around $-1$ in $\sigma(A)$ and \eq{A_t=\polar((1-t)A+t\Id)} is the path we need. Indeed, the polar part $A\mapsto A|A|^{-1}$ is a norm continuous mapping on operators that have a spectral gap about zero, which $(1-t)A+t\Id$ does for any $t\in[0,1]$ thanks to $-1\notin\sigma(A)$. Moreover, $\polar$ preserves symmetry by \cref{lem:polar part preserve symmetry} and locality by \cref{lem:continuous functional calculus of normals is local}. Indeed, $|(1-t)A+t\Id|^2$ is clearly local, and $\lambda\mapsto\lambda^{-1/2}$ is a continuous function.
		
		Now assume $-1\in \sigma(A)$. Let $V:=\ker (A+\Id)$ denote the $-1$ eigenspace for $A$. For brevity let $F=C,J$ according to the value of $\FF$. Since $AF=FA$, if $A\varphi=-\varphi$, then $AF\varphi=FA\varphi=-F\varphi$, i.e., $F\varphi\in V$ iff $\varphi\in V$. Thus \eq{FV=V,\quad FV^\perp=V^\perp} Note the space $V$ is finite dimensional since $-1$ is in the discrete spectrum of $A$. We decompose $A$ in $\calH=V\oplus V^\perp$ as \eq{A=\begin{bmatrix}-\Id & 0 \\ 0 & A_{V^\perp}\end{bmatrix}=\begin{bmatrix}-\Id & 0 \\ 0 & \Id\end{bmatrix}\begin{bmatrix}\Id & 0 \\ 0 & A_{V^\perp}\end{bmatrix}} Now $\widetilde{A}:=\Id\oplus A_{V^\perp}$ belongs to $\calU_\FF^\calL$ (as $A\in\calU^\calL_\FF$ and $(-\Id)\oplus\Id$ is unitary, local, and commutes with $F$) and $-1\notin \sigma(\widetilde{A})$, one can deform $\widetilde{A}\to \Id$ within $\calU_\FF^\calL$ as shown in the first paragraph.
		
		We now deform $(-\Id_V)\oplus \Id_{V^\perp} $ to $-\Id$ within $\calU_\FF^\calL$. Since $-\Id_V$ is already where it should be, we concentrate on deforming $\Id_{V^\perp}$ to $-\Id_{V^\perp}$ in a $\Lambda$-local way. We decompose $\Lambda$ in $\calH=V\oplus V^\perp$ as 
		\eq{\Lambda=\begin{bmatrix}\Lambda_{11} & \Lambda_{12} \\ \Lambda_{21} & \Lambda_{22}\end{bmatrix}\,.} 
		Here $\Lambda_{22}$ is essentially a projection (as $V$ is finite dimensional), and hence \cref{lem:essential projection is compact away from genuine projection} there exists a self-adjoint projection $P:V^\perp\to V^\perp$ such that $P-\Lambda_{22}\in\calK$. In particular, $[P,F]=0$ too ($F$ is diagonal in the $V$ decomposition). Now, to deform $\Id_{V^\perp}$ to $-\Id_{V^\perp}$ in a $P$-local way, since $\Id_{V^\perp}$ is diagonal in a $P$-grading of $V^\perp$, we may deform each diagonal block separately using Kuiper \cref{thm:Kuiper} and that path is guaranteed to be $P$-local as it is $P$-diagonal. Let $t\mapsto W_t$ denote this deformation from $\Id_{V^\perp}$ to $-\Id_{V^\perp}$. Then $W_t\in\calU_\FF$ and $[W_t,P]=0$. In particular \eq{[W_t,\Lambda_{22}]=[W_t,(\Lambda_{22}-P)+P]=[W_t,\Lambda_{22}-P]\in\calK\,.} 
		Thus $(-\Id)\oplus W_t \in \calU_\FF^\calL$ deforms $(-\Id)\oplus \Id$ to $-\Id$ \emph{in a $\Lambda$-local way}. Now we can deform $-\Id$ to $\Id$ within $\calU_\FF^\calL$ by decomposing $-\Id$ in the $\Lambda$ grading and the argument proceeds similarly.
	\end{proof}

	Next, we turn to the $\star\RR,\star\HH$ classes. 
	\begin{proof}[Proof of \cref{thm:star-quaternionic local unitaries}] 
		We begin by establishing surjectivity (only for the case $\FF=\star\HH$, the other case being nullhomotopic). Clearly, we have $\Id\in\calU_{\star\HH}^\calL$ and $\findex_{\Lambda,2}(\Id)=0$. We are left to construct some $U\in\calU_{\star\HH}^\calL$ with $\findex_{\Lambda,2}(U)=1$. Using the same basis choice as in \cref{eq:J equivariant basis for ker im Lambda} we define a unitary operator $U$ on $\calH$ via \eq{\vf_n\mapsto \vf_{n-1}, \quad \psi_n\mapsto \psi_{n+1}\qquad (n\in \ZZ)} and note that \eq{UJ\vf_n=U\psi_n=\psi_{n+1}=J\vf_{n+1}=JU^*\vf_n \\ UJ\psi_n=-U\vf_n=-\vf_{n-1}=J\psi_{n-1}=JU^*\psi_{n}\,.}
		Thus $UJ=JU^*$ so that $U\in\calU^\calL_{\star\HH}$. In particular, since $U\cong R\oplus R^\ast$ ($R$ being the right shift operator) in the decomposition $\szpan(\Set{\vf_n}_n)\oplus\szpan(\Set{\psi_n}_n)$, $\dim\ker\bbLambda R\oplus R^\ast = 1$ and hence $\findex_{\Lambda,2}(U)=1$.

		Next, we deal with the proof of injectivity for both $\FF=\star\RR,\star\HH$ (with $F=C,J$ accordingly). Let $U\in\calU_\FF^\calL$ and if $\FF=\star\HH$, assume for the moment that $\findex_{\Lambda,2}U=0$, the other case will be dealt with separately. 
		Applying \cref{lem:factorization for star RH local unitaries} on $U$ yields \eq{U=AB} where $A\in\calU^\calL$ with $A(BF)=(BF)A^*$ and $\Id-A\in\calK$, and $B\in\calU_\FF^\calL$ and $[\Lambda,B]=0$. Since $\Id-A\in\calK$, then the spectrum of $A$ can only accumulate at $+1$. We may rotate, e.g., anti-clockwise by $\theta$ degree, the spectrum $\sigma(A)\subset \bS^1$ of $A$ so that there is a gap at $-1$. 
		In particular $e^{i\theta}A\in\calU_\FF$ since $F$ is anti-$\CC$-linear, so \eq{e^{i\theta}A F=F e^{-i\theta}A^\ast=F(e^{i\theta}A)^*} Thus WLOG we may assume there is a gap at $-1$ in the spectrum of $A$. Now consider the path \eq{[0,1]\ni t\mapsto A_t=\polar((1-t)A+t\Id)\,.} 
		Then $A_t\in\calU^\calL_\FF$ with respect to $BF$ by \cref{lem:polar part preserve symmetry}. Consider $t\mapsto A_tB$ which deforms $U$ to $B$. We have \eq{A_tB F=BFA_t^*=FB^*A_t^*=F (A_tB)^*\,.} 
		Thus $A_tB\in\calU_\FF^\calL$ for all $t$. 
		Finally, since $[B,\Lambda]=0$, we may use \cref{star-RH Kuiper} to deform each diagonal block of $B$ in the $\Lambda$ grading to $\Id$, resulting in $B\to \Id$ within $\calU_\FF^\calL$.
		
		Finally, consider $\FF=\star\HH$ assuming that $\findex_{\Lambda,2} U=1$. Unlike in the ordinary $\ZZ$-valued index where we use the logarithmic law \cref{lem:log law for index_Lam} and then we may always connect a zero-index operator to $\Id$, in the present case for the $\ZZ_2$ index, we rather directly argue by connecting any two non-zero index operators together. Hence consider $U,\widetilde{U}\in\calU_{\star\HH}^\calL$ both with non-zero index.   
		
		Let us start by deforming $U$ (and also $\ti U$) to a more convenient non-zero index operator. Decompose $U$ in block form in the $\Lambda$ grading as before in \cref{eq:Lambda decomposition of U}. For $Z=U_{LL}$ for brevity and $X=\polar(Z)\in\calB_{\star\HH}$, let $m$ be defined by  $2m+1:=\dim\ker X=\dim (\im X)^\perp $. We use \cref{cor:compact perturbation preserve Z2 index} to conclude that $\dim\ker U_{RR}$ is odd as well.
		Similar to the proof in \cref{lem:factorization for star RH local unitaries}, we may extend $X$ to a \emph{partial isometry} $Y\in\calB_{\star\HH}$ that has $\dim \ker Y=\dim (\im Y)^\perp = 1$, and such that $Y-Z\in \calK$. Indeed, we cannot extend $Y$ to a unitary in $\calU_{\star\HH}$ precisely since the kernels are not even dimensional (see \cref{lem:extension of star fredholm partial isometries} whose hypothesis is a zero index). 
		Let $B_{L} := Y$ and $B_{R}$ the corresponding construction out of $U_{RR}$. Even though we cannot extend $B_{i}$ to a unitary separately, we will show that this may be done on the bigger space for the direct sum $B:=B_{L}\oplus B_{R}$. Let $\ker B_{i}$ be spanned by the unit vector $\eta_i$. Since $B_i\in \calB_{\star\HH}$, $(\im B_{i})^\perp$ is spanned by the vector $\xi_i:=J\eta_i$. 
		We now define an operator $M$, which is unitary when considered as a map $M:\ker B\to(\im B)^\perp$, written in the $\Lambda$-grading as \eq{M:=\begin{bmatrix}0 & \xi_L\otimes \eta_R^* \\ -\xi_R \otimes \eta_L^* & 0\end{bmatrix}} that satisfies $MJ=JM^*$. Indeed, we have 
		\eq{MJ 
			& =\begin{bmatrix}0 & \xi_L\otimes \eta_R^* \\ -\xi_R \otimes \eta_L^* & 0\end{bmatrix} \begin{bmatrix}J & 0 \\ 0 & J\end{bmatrix} 
			\\ & =\begin{bmatrix}0 & (\xi_L\otimes \eta_R^*) J \\ (-\xi_R \otimes \eta_L^*) J & 0\end{bmatrix} 
			\\ &=\begin{bmatrix}0 & J(-\eta_L\otimes \xi_R^*) \\ J(\eta_R\otimes \xi_L^*) & 0\end{bmatrix} 
			\\ &= JM^*} 
		where in the third equality, we use the fact \eq{(\xi_L\otimes \eta_R^*) J\vf= (\eta_R,J \vf)\xi_L=(\vf,J^*\eta_R)\xi_L=-(\vf,\xi_R)\xi_L=-J(\xi_R,\vf)\eta_L=J(-\eta_L\otimes\xi_R^*)\vf } and similarly for $(-\xi_R\otimes \eta_L^*)J=J(\eta_R\otimes \xi_L^*).$
		With $M$ we define the operator \eql{D:=B+M=\begin{bmatrix}B_{L} & \xi_L\otimes \eta_R^* \\ -\xi_R \otimes \eta_L^* & B_{R}\end{bmatrix}\label{eq:B with off-diagonal entries}\,.} Thus $D$ is unitary by construction, and local, since $\xi_L\otimes \eta_R^*$ and $-\xi_R \otimes \eta_L^*$ are finite-rank, and obeys $DJ=JD^\ast$. 
		We define $A:=UD^*$, write $U=AD$, and deform $A$ away using similar argument as before ($\Id-A= (D-U)D^\ast\in\calK$ and $A(DJ)=(DJ)A^\ast$). Thus there is a path from $U$ to $D$ within $\calU^\calL_{\star\HH}$, and similarly a path from $\ti U$ to $\ti{D}$, where $\ti{D}$ is constructed analogously to $D.$
		
		Hence we are left to deform $D$ to $\ti D$. For each $i\in\Set{L,R}$, we use \cref{lem:intertwines U and tiU in star H} right below to construct unitaries $X_{i},Y_{i}$ such that $B_{i}=X_{i}^*\widetilde{B}_{i}Y_{i}$ and $Y_{i}J=JX_{i}$. Plugging in $X,Y$ into $D$ we find
		\eq{D &=\begin{bmatrix}X^*_{L}\widetilde{B}_{L}Y_{L} & \xi_L\otimes \eta_R^* \\ -\xi_R \otimes \eta_L^* & X_{R}^*\ti{B}_{R}Y_{R}\end{bmatrix} 
			\\ &= \begin{bmatrix}X_{L}^* & 0 \\ 0 & X_{R}^*\end{bmatrix}\begin{bmatrix}\widetilde{B}_{L} & X_{L} (\xi_L\otimes \eta_R^*)Y_{R}^* \\ X_{R}(-\xi_R \otimes \eta_L^*)Y_{L}^* & \ti{B}_{R}\end{bmatrix}\begin{bmatrix}Y_{L} & 0 \\ 0 & Y_{R}\end{bmatrix}\,.} 
		In particular, from how $X_i,Y_i$ are constructed in \cref{eq:how V11 W11 maps}, one has \eq{X_{R}(-\xi_R \otimes \eta_L^*)Y_{L}^* =-(X_{R}\xi_R) \otimes (Y_{L}\eta_L)^*= -\ti{\xi}_R\otimes \widetilde{\eta}_L^*\,.} 
		Similarly $X_{L} (\xi_L\otimes \eta_R^*)Y_{R}^*=\tilde{\xi}_L\otimes \widetilde{\eta}_R^*$. We write $X=X_{L}\oplus X_{R}$ and $Y=Y_{L}\oplus Y_{R}$. Then \eq{D=X^*\widetilde{D}Y=X^*\widetilde{D}JXJ^*} where we used $YJ=JX$ in the last step. 
		Applying \cref{thm:Kuiper} on $X$, there exists a path $[0,1]\ni t\mapsto X_t\in\calU$ connecting $X\rightsquigarrow\Id$. Let $B_t=X_t^*\widetilde{B}JX_tJ^*$. Then \eq{B_tJ=X_t^*\widetilde{B}JX_t=X_t^*J\widetilde{B}^*X_t= J(JX_t^*J^*\widetilde{B}^* X_t)=JB_t^*\,.} Thus $B_t\in\calU_{\star\HH}^\calL$ deforms $B$ to $\widetilde{B}$ as desired.
		
	\end{proof}
	
	\begin{lem}\label{lem:intertwines U and tiU in star H}
		Let $U,\ti{U}\in\calB_{\star\HH}$ be partial isometries such that the dimension of the kernels of $U,\ti{U},U^*,\ti{U}^*$ are all finite and equal. Then there exists $V,W\in\calU$ such that $U=V^*\ti{U}W$ and $WJ=JV.$
	\end{lem}
	\begin{proof}
		We first write 
		\eq{
			U&=\begin{bmatrix}0 & 0 \\ 0 & U_2\end{bmatrix}:\ker U \oplus (\ker U)^\perp \to (\im U)^\perp \oplus \im U 
			\\ \ti{U}&=\begin{bmatrix}0 & 0 \\ 0 & \ti{U}_{2}\end{bmatrix}:\ker \ti{U} \oplus (\ker \ti{U})^\perp \to (\im \ti{U})^\perp \oplus \im \ti{U}\,.
		} 
		Define operators $V,W$ that take the block form 
		\eq{W&=\begin{bmatrix}W_{1} & 0 \\ 0 & W_{2}\end{bmatrix}:\ker U \oplus (\ker U)^\perp \to \ker \ti{U}\oplus (\ker \ti{U})^\perp \\ V&=\begin{bmatrix}V_{1} & 0 \\ 0 &  \ti{U}_{2}W_{2} U_{2}^*\end{bmatrix}: (\im U)^\perp \oplus \im U \to (\im\ti{U})^\perp \oplus \im \ti{U}}
		where $V_{1},W_{1},W_{2}$ are, for now, some unitaries to be constructed explicitly momentarily. With this, it is clear that \eq{U=V^*\ti{U}W\,.} 
		
		We construct the unitaries $V_{1},W_{1},W_{2}$ such that the following two conditions hold \eql{W_{1}J&=JV_{1} \label{eq:V11 W11 symmetry relation}\\ W_{2}J&=J \ti{U}_{2}W_{2}U_{2}^*\label{eq:V22 W22 symmetry relation}\,.} 
		The expressions make sense since $U\in \calB_{\star\HH}$ implies that \eq{J\ker U=(\im U)^\perp\,,\quad  J(\ker U)^\perp =\im U\,;}
		similar expressions holds for $\ti{U}$. Now we let $\ker U$ be spanned by the ONB $\Set{\eta_i}_{i=1}^m$. Then $(\im U)^\perp$ is spanned by the ONB $\Set{\xi_i}_{i=1}^m$, where $\xi_i:=J\eta_i$. We construct analogous tilde version of ONB for $\ker \ti{U}$ and $(\im \ti{U})^\perp$. Define \eql{V_{1}:\sum_{i=1}^m \ti{\xi}\otimes \xi_i^*,\quad W_{1}:\sum_{i=1}^m\ti{\eta}_i\otimes \eta_i^*.\label{eq:how V11 W11 maps}} 
		Thus \cref{eq:V11 W11 symmetry relation} holds. 
		
		We find that $JU_{2}$ defines a $\star\HH$-structure on $(\ker U)^\perp$, so applying \cref{lem:basis in RCH symmetry} gives an ONB consisting of Kramers pairs $\Set{\vf_i,\psi_i}_{i=1}^\infty$ for $(\ker U)^\perp$ such that $\psi_i=JU_{2}\vf_i$. Similarly, let $\Set{\vf_i,\ti{\psi}_i}_{i=1}^\infty$ be an ONB of Kramers pairs for $J\ti{U}_2$ on $(\ker \ti{U})^\perp$. Construct $W_{2}$ as \eql{\vf_i\mapsto -\ti{\psi}_i,\quad \psi_i\mapsto \ti{\vf}_i } One again readily verifies that \cref{eq:V22 W22 symmetry relation} holds. The relations \cref{eq:V11 W11 symmetry relation,eq:V22 W22 symmetry relation} are equivalent to \eq{WJ=JV\,.}
	\end{proof}
	
	\section{Equivariant classification of $\Lambda$-non-trivial self-adjoint unitaries}\label{sec:local SA unitaries classification}

	In this section we turn our attention to equivariant local self-adjoint (orthogonal) projections, and calculate the corresponding set of path-connected components. Now, however, we add a non-triviality condition which is stronger than locality, and moreover, the symmetry classes $\FF$ we consider are slightly different. To explain the difference, let us consider equivariant local self-adjoint unitaries $\calSU^{\calL}_{\FF}$ rather than projections (see \cref{rem:connection between projections and self-adjoint unitaries}); we shall abbreviate \emph{SAU} henceforth. We prefer working with SAUs here for two reasons: (1) the physical symmetry constraints appear naturally at the level of the self-adjoint unitaries, which, as we'll see below, are flat Hamiltonians, and (2) the calculations below somewhat simplify in this way. As for the symmetries, we have still $\CC$ (no constraint), $\RR$ and $\HH$, i.e., the SAU would commute with $F=C,J$. However now we replace $\star\RR,\star\HH$ by $\ii \RR,\ii\HH$ respectively (since these are the conditions which arise from particle-hole symmetry later, see \cref{sec:bulk insulators classification}).
	
	As was mentioned, now we constrain the class of SAUs we study even more beyond locality in a crucial way. We have already seen the notion of a non-trivial SAU in \cref{def:non trivial projections}: this is a SAU operator where both $\pm1$ eigenspaces are infinite dimensional. We shall also need
	\begin{defn}[$\Lambda$-non-trivial SAUs]\label{def:Lambda-non-trivial SAUs} $U\in\calSU$ is called $\Lambda$-non-trivial iff there exists some $V\in\calSU$ such that:
		\begin{enumerate}
			\item $[V,\Lambda]=0$ (hyper-local).
			\item $\Lambda V\Lambda$ and $\Lambda^\perp V \Lambda^\perp$ are both non-trivial SAUs.
			\item $U-V\in\calK$.
		\end{enumerate}
		We note this implies automatically that such a $U$ is $\Lambda$-local since by definition $\Lambda U \Lambda^\perp$ is compact. We denote the space of all $\Lambda$-non-trivial SAUs by $\LamNTSAU$. Hence we have \eq{\LamNTSAU\subsetneq\calSU^\calL\subsetneq\calSU\,.}
	\end{defn}
	
	It turns out that if one attempts to classify the bare, merely $\Lambda$-local space $\calSU^\calL$, the result is \emph{not} nullhomotopic (as one would expect from \cref{table:Kitaev}) due to finite rank problems which roughly correspond to half-infinite systems. So later on, in \cref{sec:bulk insulators classification} we will see that the correct notion to reproduce \cref{table:Kitaev} is rather the more constrained $\Lambda$-non-trivial space $\LamNTSAU$ and that in a sense, unitaries are \emph{automatically} $\Lambda$-non-trivial (see \cref{lem:chiral insulators are automatically bulk non trivial} below), which is why this notion was not necessary in \cref{sec:local unitaries classification}. Another point in support of this notion is that $\Lambda$-non-triviality is well-defined in the sense that it is preserved under small norm and compact perturbations within $\calSU^\calL$, see \cref{lem:lambda-non-triviality is well-defined} at the end of this section.

	\begin{table}
		\begin{center}
			\begin{tabular}{|c|c|}
				\hline
				$\pi_0(\LamNTSAU)$ & $\Set{0}$\\\hline
				$\pi_0(\LamNTSAU_{\RR})$ & $\Set{0}$\\\hline
				$\pi_0(\LamNTSAU_{\HH})$ & $\Set{0}$\\\hline
				$\pi_0(\LamNTSAU_{\ii\RR})$ & $\ZZ_2$\\\hline
				$\pi_0(\LamNTSAU_{\ii\RR})$ & $\Set{0}$\\
				\hline
			\end{tabular}
		\end{center}
		\caption{The classification of equivariant non-trivial self-adjoint unitaries.}
		\label{table:equivariant Lambda-non-trivial SAUs}
	\end{table}
	We finally turn to our main classification theorems. The results of this section are summarized in \cref{table:equivariant Lambda-non-trivial SAUs}. 
	\begin{thm}[Classification of $\Lambda$-non-trivial $\RR,\CC,\HH$ SAUs]\label{thm:classification of RCH LamNT SAUs} The space of $\Lambda$-non-trivial $\RR,\CC,\HH$ SAUs is null-homotopic: \eql{\pi_0(\LamNTSAU_\FF) \cong\Set{0}\qquad(\FF\in\Set{\CC,\RR,\HH})\,.}
	\end{thm} 
	When $\FF=\CC$ (i.e., without any symmetry constraints) this theorem is not new, and appeared relatively recently within \cite{AndruchowChiumientoLucero2015}. Here we extend it also to the cases $\FF=\RR,\HH$ (that extension is straightforward) in a unified proof for all three cases.

	The following theorem is new to our knowledge:

	\begin{thm}[Classification of $\Lambda$-non-trivial $\ii\RR$ and $\ii\HH$ SAUs]\label{thm:classification of iR iH LamNT SAUs} The space of $\Lambda$-non-trivial $\ii\RR$ SAUs has two path components. The map $\findex_{\Lambda,2}:\LamNTSAU_{\ii\RR}\to\ZZ_2$ is norm continuous and ascends to a bijection \eql{\findex_{\Lambda,2}:\pi_0(\LamNTSAU_{\ii\RR})\xrightarrow{\sim}\ZZ_2\,.}
		
		The space of $\Lambda$-non-trivial $\ii\HH$ SAUs is null-homotopic: \eql{\pi_0(\LamNTSAU_{\ii\HH}) \cong\Set{0}\,.}
	\end{thm} 
	
	We now present the proofs of \cref{thm:classification of RCH LamNT SAUs,thm:classification of iR iH LamNT SAUs}.
	\begin{proof}[Proof of \cref{thm:classification of RCH LamNT SAUs}]
		Let $U,\widetilde{U}\in\LamNTSAU_{\FF}$ be given, with $\FF=\CC,\RR,\HH$. Our goal is to construct a continuous path within $\LamNTSAU_{\FF}$ from $U$ to $\widetilde{U}$. The main idea is as follows: using \cref{lem:classification of nontrivial projections} below we know that non-trivial SAUs are null-homotopic without the further $\Lambda$-non-triviality constraint. So if it turned out that both $U,\widetilde{U}$ were diagonal in the $\Lambda$ grading, we would be finished. We thus concentrate on showing how $U$ may be deformed into a $\Lambda$-diagonal element in $\LamNTSAU_\FF$, with the deformation within that space.
		
		Hence, in the $\calH = \im(\Lambda)^\perp\oplus\im(\Lambda)$ grading, let us write \eql{\label{eq:space decomposition of SA unitary}U = \begin{bmatrix}
				X & A \\ A^\ast & Y
		\end{bmatrix}} where $X,Y$ are self-adjoint operators and $A:\im(\Lambda)\to\im(\Lambda)^\perp$ is general. It should be emphasized that $X,Y$ are \emph{not} SAUs since $A\neq0$. By locality however, $A$ is compact, and since for any SAU, $-\Id\leq U\leq\Id$, we also have $-\Id\leq X,Y\leq \Id$. Since $\Id=U^2$, we have $X,Y$ essentially unitary; in particular, $|A|^2 = \Id-Y^2$, $|A^\ast|^2=\Id-X^2$. Finally, the intertwining property \eql{\label{eq:intertwining property XA=minus AY}XA = -A Y} holds. Using it, we can show that \eql{\label{eq:A maps eigenvectors of Y to minus eigenvectors of X}A:\ker(Y-\lambda\Id)\to\ker(X+\lambda\Id)\qquad(\lambda\in(-1,1))} isomorphically. Indeed, let $Y\psi = \lambda \psi$ for $|\lambda|<1$ and $\norm{\psi}=1$. Then $XA\psi=-\lambda A\psi$ and \eql{\label{eq:eigenvectors of Y mapped by A to eigenvectors of X}\norm{A\psi}^2=\ip{\psi}{|A|^2\psi}=\ip{\psi}{(\Id-Y^2)\psi}=(1-\lambda^2)\norm{\psi}^2\neq0\,.} This works similarly to show that $A^\ast$ is injective and hence $A$ is the claimed isomorphism.
		
		Now, $X,Y$ have spectra which may only accumulate at $\pm1$, since \eq{\sigma_{\mathrm{ess}}(U) = \sigma_{\mathrm{ess}}(X)\cup\sigma_{\mathrm{ess}}(Y)} so that on $(-1,1)$ both $X$ and $Y$ have discrete spectrum, and the intertwining property implies that \eql{\label{eq:spectral projections flip via A}\chi_{\Set{\lambda}}(X) A = A \chi_{\Set{-\lambda}}(Y)} for all $\lambda\in(-1,1)$. Indeed, \eq{A \chi_{\Set{-\lambda}}(Y) = \chi_{\Set{\lambda}}(X) A \chi_{\Set{-\lambda}}(Y) = \chi_{\Set{\lambda}}(X) A \left(\Id-\chi_{\Set{-\lambda}}(Y)^\perp\right)
		} where the first inequality is thanks to \cref{eq:A maps eigenvectors of Y to minus eigenvectors of X}. Now, $\chi_{\Set{\lambda}}(X) A \chi_{\Set{-\lambda}}(Y)^\perp=0$ because $\ker(A)=\ker(|A|^2)=\ker(\Id-Y^2)$, so $A$ is zero on any eigenvector of $Y$ of eigenvalue of modulus $1$, and on any other eigenvalue $-\mu$, \cref{eq:A maps eigenvectors of Y to minus eigenvectors of X} shows that it maps to the range of $\chi_{\Set{-\mu}}(X)$.
		
		We will use the definition of $\sgn:\RR\to\RR$ given by \eql{\label{eq:sgn function}\sgn(\lambda) \equiv -\chi_{(-\infty,0)}(\lambda) +\chi_{(0,\infty)}(\lambda)\qquad(\lambda\in\RR)} and with it define $f_\pm:\RR\to\RR$ via \eql{\label{eq:fpm function}f_\pm(\lambda) := \sgn(\lambda) \pm \chi_{\Set{0}}(\lambda)\qquad(\lambda\in\RR)\,.}
		
		We define the diagonal SAU operator 
		\eql{\label{eq:V diagonal from SAU U}V:=f_+(X)\oplus f_-(Y)\,.}
		Here $V:=f_+(X)\oplus f_+(Y)$ also works.
		This is the operator we shall deform $U$ into. To do so, we shall use the conjugating self-adjoint operator due to \cite{AndruchowChiumientoLucero2015}: \eql{\label{eq:G invertible from U and V}G := \frac{1}{2}\left(U + V\right)\,.}
		
		Let us note the effect of different symmetries in this context, i.e., we assume that $[F,U]=0$ for $F=\Id,C,J$. Since the symmetry operators are hyper-local by \cref{ass:real and quaternionic structures are hyperlocal}, we have $X,Y\in\calB_\FF$ too. In particular $\sgn(X),\sgn(Y)$ and $\chi_{\Set{0}}(X),\chi_{\Set{0}}(Y)$ also belong to $\calB_\FF.$ Indeed, even though $F$ is anti-unitary, $X,Y$ are self-adjoint and hence the anti-unitarity does not interfere. Hence $G\in\calB_\FF$ too, and so we can essentially forget about the symmetry constraints as long as we have symmetric versions of \cref{thm:Kuiper,lem:classification of nontrivial projections}, which we do.
		
		Next, we note that that since $U^2=V^2=\Id$, $GU=\frac{1}{2}(U+V)U=\frac{1}{2}V(U+V)=VG$. Also, $G$ is clearly Fredholm of zero $\Lambda$-index (in the sense of \cref{eq:index of local unitary}). Indeed, since $\lambda\mapsto\lambda+f_\pm(\lambda)=:g_\pm(\lambda)$ has $\im(g_\pm)\cap(-1,1)=\varnothing$, $G$ is, up to compact, a direct sum of invertible, and hence zero-index Fredholm, operators: \eq{G - \frac{1}{2}\left( g_+(X)\oplus g_-(Y)\right) \in\calK\,.} We find \eq{\findex_\Lambda G = \findex g_+(X) = 0\,.}
		In fact $G=\frac12\begin{bmatrix} g_+(X) & A \\ A^\ast &  g_-(Y)\end{bmatrix}$ is invertible. Indeed, since $\findex G=0$, it suffices to check that ${\ker G=\Set{0}}$. Suppose $G \begin{bmatrix}\vf \\ \psi\end{bmatrix}=0$, then \eql{\label{eq:G phi psi = 0 (1)}g_+(X)\vf + A\psi = 0 \\ \label{eq:G phi psi = 0 (2)}A^*\vf + g_-(Y)\psi = 0\,.}
		Multiply the first equation by $A^*$ to get \eq{A^* g_+(X)\vf + |A|^2\psi = 0\,.} Now, using the intertwining property \cref{eq:intertwining property XA=minus AY}, we have \eq{A^*g_+(X) &= A^*(X + \sgn(X)+\chi_{\Set{0}}(X)) \\ &= (-Y-\sgn(Y) + \chi_{\Set{0}}(Y))A^* = -g_-(Y)A^*\,.} 
		Thus \eq{0=A^* g_+(X)\vf + |A|^2\psi=-g_-(Y)A^*\vf+|A|^2\psi = (g_-^2(Y)+\Id-Y^2)\psi} where, in the last equality, we have used \cref{eq:G phi psi = 0 (2)}, and $|A|^2=\Id-Y^2.$ Note that $\lambda\mapsto g_-^2(\lambda) + 1-\lambda^2 = 2|\lambda|+2$ has range in $[2,\infty)$, and hence the above implies that $\psi=0.$ Similarly, one can show that $\vf=0.$
		
		We now use the self-adjoint invertible operator $G$ as follows. Since $GU = VG$, we have $G^2 U = U G^2$ and thus also $|G| U = U |G|$. This implies that we also have $\polar(G) U = V \polar(G)$, i.e., \eql{\label{eq:V and U intertwines} V = \polar(G)^\ast U \polar(G)\,.} But $\polar(G)\in\calU^\calL_\FF$ (we are using \cref{lem:polar part preserve symmetry}) and, $\findex_\Lambda \polar(G) = 0$ (the index is invariant under the taking the polar part \cite[Lemma 6]{Graf_Shapiro_2018_1D_Chiral_BEC}). Thus using \cref{thm:classification of RCH local unitaries} $\polar(G)$ may be deformed to $\Id$ within the space $\calU^\calL_\FF$. This yields an equivariant local path within $\calU_\FF$ \eq{U \rightsquigarrow V \overset{\mathrm{\cref{lem:classification of nontrivial projections}}}{\rightsquigarrow} \ti{V}\rightsquigarrow \ti{U}\,.} So far we have only exhibited this path as $\Lambda$-local, but it is in fact  $\Lambda$-non-trivial via \cref{lem:lambda-non-triviality is well-defined} below. 
		
	\end{proof}

	We turn to the two remaining, more exotic symmetry classes, $\ii \RR$ and $\ii \HH$, one of which is not nullhomotopic.
	\begin{proof}[Proof of \cref{thm:classification of iR iH LamNT SAUs}] Since we claim that $\LamNTSAU_{\ii \HH}$ is nullhomotopic, we only need to establish surjectivity for $\findex_{2,\Lambda}:\LamNTSAU_{\ii \RR}\to\ZZ_2$. We first construct some $U\in\LamNTSAU_{\ii \RR}$ with a zero index. 
		
		To this end, we invoke \cref{lem:basis in RCH symmetry} separately on $\ker\Lambda,\im\Lambda$ to obtain an ONB for $\calH$ $\Set{\vf_i,\psi_i}_{i\in\ZZ}$ such that $i\leq0$ spans the kernel and $i>0$ spans the image, and moreover, $C\vf_i=\psi_i$ for $i\in\ZZ$. Let now $P^\pm_\vf,P^\pm_\psi$ be self-adjoint projections onto $\Set{\vf_i}_{i>0},\Set{\psi_i}_{i>0}$ and $\Set{\vf_i}_{i\leq0},\Set{\psi_i}_{i\leq0}$ respectively, and define \eq{U=(P_\vf^--P_\psi^-)\oplus (P_{\vf}^+-P_{\psi}^+)\,.} Then $U$ is a $\Lambda$-non-trivial SAU that belongs to $\calB_{\ii\RR}.$
		
		We turn to the construction of a $U\in\LamNTSAU_{\ii\RR}$ that has a non-trivial index. Using \cref{lem:basis in RCH symmetry} separately on $\ker\Lambda,\im\Lambda$, we have an orthonormal basis fixed by $C$ for each of these spaces. Pick out a vector $\eta^\pm$ out of each, and re-apply the lemma on $\ker\Lambda\ominus\szpan(\eta^-),\im\Lambda\ominus\szpan(\eta^+)$ to obtain $\Set{\vf_i,\psi_i}_{i\in\ZZ}$ such that $i\leq0$ spans the kernel minus $\eta^-$ and $i>0$ spans the image minus $\eta^+$, and moreover, $C\vf_i=\psi_i$ for $i\in\ZZ$. We define $P^\pm_\vf,P^\pm_\psi$ similarly as in the previous paragraph and define \eq{U=\begin{bmatrix} P_\vf^--P_\psi^- & -\ii\eta^-\otimes \left(\eta^+\right)^\ast \\ \ii\eta^+\otimes \left(\eta^-\right)^\ast & P_{\vf}^+-P_{\psi}^+\end{bmatrix}} Then $U\in\LamNTSAU_{\ii\RR}$ and $\findex_{\Lambda,2}(U)=\dim\ker (P_\vf^+-P_\psi^+)=1$ as $\ker (P_\vf^+-P_\psi^+)=\szpan(\eta^+)$.
		
		Let us now establish injectivity. Let $U,\ti U\in\LamNTSAU_\FF$ be given, with $\FF=\ii\RR,\ii\HH$ and $F$ accordingly, and if $\FF=\ii\RR$, we assume that both operators have zero index (the non-trivial will be dealt with later on). We will deform $U\rightsquigarrow\ti U$ by a similar procedure to \cref{thm:classification of RCH LamNT SAUs}: first we deform $U$ to a $\Lambda$-diagonal SAU $V$, and then we deform $V\rightsquigarrow \ti V$. So now we write $U$ in \cref{eq:space decomposition of SA unitary} and we will reuse the properties proven and the construction defined in the proof of \cref{thm:classification of RCH LamNT SAUs}. Since $F$ is hyper-local, $\Set{U,F}=0$ implies that \eql{\label{eq:X Y A are iiR iiH symmetric}X,Y,A\in \calB_\FF\,.} 
		We would have liked to use $V\equiv f_+(X)\oplus f_-(Y)$ from \cref{eq:V diagonal from SAU U}; that construction, however, is unsatisfactory since now $V$ does not respect the $\FF$ symmetry constraint: \eq{f_+(X)F=(\sgn(X)+\chi_{\Set{0}}(X))F=F(-\sgn(X)+\chi_{\Set{0}}(X))=-F f_-(X)\neq -F f_+(X)\,.} 
		To fix this, we will decompose the zero eigenspace of $X$ into two disjoint parts of the same dimension $\chi_{\Set{0}}(X)=E + L$ such that \eql{E F=F L\,.\label{eq:symmetry intertwines two parts of kernels at zero}} 
		To do so, we should have $\ker X$ even dimensional. In the case $\FF=\ii\HH$, the kernel is always even dimensional. Indeed, since $X\vf=0$ iff $XF\vf=-FX\vf=0$, it follows that $\chi_{\Set{0}}(X)F=F\chi_{\Set{0}}(X)$. Hence $F=J$ is a symmetry operator on $\im\chi_{\Set{0}}(X)\equiv \ker X$ and hence $\dim\ker X\in 2\NN$ from \cref{lem:basis in RCH symmetry}, and we get an ONB of Kramers pairs $\Set{\vf_i,\psi_i}_{i=1}^m$ for $\ker X$. Let \eql{\label{eq:definition of two parts of kerX}E :=\sum_{i=1}^m \vf_i\otimes \vf_i^*,\quad L:=\sum_{i=1}^m\psi_i\otimes \psi_i^*\,.}
		Then \eq{LJ\xi=\sum_{i=1}^m \ip{\psi_i}{J\xi}\psi_i = \sum_{i=1}^m \ip{\xi,J^*\psi_i}\psi_i = \sum_{i=1}^m \ip{\xi,\vf_i}J\vf_i= \sum_{i=1}^m J \ip{\vf_i,\xi_i}\vf_i=J E\xi\,.}
		For the case $\FF=\ii\RR$ we need to further impose that $\dim\ker X\in 2\NN$ (which is equivalent to $\findex_{\Lambda,2}(U)=0$, using the fact that the total operator is unitary, and stability of the index under compacts, see \cref{cor:compact perturbation preserve Z2 index})--we deal with the odd case in the end. Hence for $\FF=\ii\RR$ and $F=C$, apply \cref{lem:basis in RCH symmetry} on $\ker X$ to obtain an ONB $\Set{\vf_i,\psi_i}_{i=1}^m$ for  with $C\vf_i=\psi_i$. Now we define similarly $E$ and $L$ as in \cref{eq:definition of two parts of kerX}.
		
		Note that \cref{eq:spectral projections flip via A} implies $\chi_{\Set{0}}(Y)=A^{-1}\chi_{\Set{0}}(X)A$ (where we mean $A$ as the isomorphism in \cref{eq:A maps eigenvectors of Y to minus eigenvectors of X}). In fact, using \cref{eq:eigenvectors of Y mapped by A to eigenvectors of X} with $\lambda=0$ we see that $A$ maps $\ker Y$ \emph{unitarily} onto $\ker X$. Thus we can write $\chi_{\Set{0}}(Y)=A^*\chi_{\Set{0}}(X)A$. Let us therefore define the diagonal SAU that \emph{does} satisfy the symmetry constraint \eq{V=\begin{bmatrix}\sgn(X)+E-L & 0 \\ 0 & \sgn(Y)-A^*(E-L)A \end{bmatrix}=:V_L\oplus V_R\,.} 
		Indeed, $V\in\calB_\FF$ since \eq{V_LF=(\sgn(X)+E-L)F=F(-\sgn(X)+L-E)=-FV_L} and similarly for $V_R$. Moreover, $V$ is a SAU since a short calculation (which uses the intertwining property \cref{eq:intertwining property XA=minus AY}) shows $V_L^2=\Id_{\Lambda^\perp},V_R^2=\Id_\Lambda$.
		
		Now that we have a diagonal symmetric $V$, we may define the conjugation operator $G$ as before in \cref{eq:G invertible from U and V}. One has to be slightly careful since the different definition of $V$ leads to a different $G$ compared with \cref{eq:G invertible from U and V}, however, the two differ by a compact. Hence, all properties of $G$ from before still hold, and in particular, it is Fredholm of zero index. Similarly we can show that $G$ is invertible, and again it suffices to check that $\ker G=\Set{0}$ since $\findex G=0$. To that end, suppose $G\begin{bmatrix}\vf \\ \psi\end{bmatrix}=0$, then \eq{(X+V_L)\vf + A\psi = 0 \\ A^*\vf + (Y+V_R)\psi=0\,.} Using the intertwining property \cref{eq:intertwining property XA=minus AY}, we have \eq{A^*(X+V_L)=A^*(X+\sgn(X)+E-L)=(-Y-\sgn(Y)+A^*(E-L)A)A^*=-(Y+V_R)A^*\,.}
		Hence when we multiply the first equation by $A^*$ we get \eq{0=A^*(X+V_L)\vf+A^*A\psi=-(Y+V_R)A^*\vf + |A|^2\psi = ((Y+V_R)^2+\Id-Y^2)\psi} where in the last step we use the second equation and $|A|^2=\Id-Y^2.$ Now \eq{(Y+V_R)^2+\Id-Y^2=Y^2+V_R^2 + YV_R+V_RY + \Id-Y^2=2\Id + YV_R+V_RY\,.} 
		But observe that  \eq{YV_R=Y(\sgn(Y)-A^*(E-L)A)=Y\sgn(Y)+A^*\underbrace{X(E-L)}_{=0}A=Y\sgn(Y)=|Y|\,.} Thus $(Y+V_R)^2+\Id-Y^2=2\Id + 2|Y|$, which implies that $\psi=0$. Similarly, one can show that $\vf=0$.
		
		Since $U,V\in\calB_\FF^\calL$, it follows that $G\in\calB_\FF^\calL$, and hence $W:=\polar(G)\in\calU^\calL_\FF$, and we have $V=W^*UW$ similar to \cref{eq:V and U intertwines}. The rest of the arguments follow analogously to the proof in \cref{thm:classification of RCH LamNT SAUs}, with the exception, however of an equivariant version (adapted to $\ii\RR$ and $\ii\HH$) of \cref{thm:Kuiper}. To that end, let us rather apply \cref{thm:classification of RCH local unitaries} on $-\ii W$. Indeed, since \eq{(-\ii W)F=-\ii (-FW)=F(-\ii W)} it follows that $-\ii W\in\calU^\calL_{-\ii \FF}$ with $\FF=\ii\RR,\ii\HH$, i.e., $-\ii W$ is a standard real or quaternionic operator, and there is a path from $-\ii W$ to $\Id$ within $\calU^\calL_{-\ii \FF}$. We multiply by $\ii$ again to obtain a path from $W$ to $\ii \Id$ within $\calU^\calL_\FF.$ Thus $V$ can be deformed to $(\ii \Id)^* U (\ii \Id)=U$ in $\LamNTSAU_\FF.$

		Finally, we tackle the problem of connecting two operators $U,\ti{U}\in\LamNTSAU_{\ii \RR}$ both of whom have non-zero index. However, we can't exactly follow the strategy above since it turns out that it is not possible to deform $U$ to a diagonal SAU that obeys that symmetry constraint. Since, by definition, $\findex_{2,\Lambda} U = 1$ means $\dim \ker X \in 2\NN+1$ (still using \cref{eq:space decomposition of SA unitary}), we decompose this kernel into \emph{three} disjoint parts \eq{\chi_{\Set{0}}(X)=E+L+\eta\otimes \eta^*} where $E,L$ satisfies \cref{eq:symmetry intertwines two parts of kernels at zero}, and $\eta\in(\im\Lambda)^\perp$ is fixed by the symmetry $C\eta=\eta$. It is possible to find such $\eta.$ Indeed, apply \cref{lem:basis in RCH symmetry} to construct an ONB of $\ker X$ fixed by $C$, and pick one $\eta$ in this collection. Now apply \cref{lem:basis in RCH symmetry} again on $\ker X \ominus\CC\eta$, and obtain $\Set{\vf_i,\psi_i}_{i=1}^m$ such that $C\vf_i=\psi_i$, and we construct $E,L$ similar as before \cref{eq:definition of two parts of kerX}. 
		Define \eql{\label{eq:define xi in iiR case}\xi:=  A^*\eta \in\im \Lambda} so that  $C\xi=CA^*\eta = -A^*C\eta =-A^*\eta =-\xi$, where we use \cref{eq:X Y A are iiR iiH symmetric} to show that $CA^*=-A^*C$.
		Construct a SAU \eql{\label{eq:V non diagonal}V:=\begin{bmatrix}\sgn(X)+E-L & \eta\otimes \xi^* \\ \xi\otimes \eta^* & \sgn(Y)-A^*(E-L)A\end{bmatrix}=:\begin{bmatrix}V_L & \eta\otimes \xi^* \\ \xi\otimes \eta^* & V_R\end{bmatrix}\,.} 
		We have $V^2=\Id$ from the following computation: \eq{(\sgn(X)+E-L)^2+(\eta\otimes\xi^*) (\xi\otimes \eta^*)=\sgn(X)^2+(E+L+\eta\otimes \eta^*)=(\sgn(X)+\chi_{\Set{0}}(X))^2=\Id} and \eq{(\xi\otimes\eta^*)(\sgn(X)+E-L)+(\sgn(Y)-A^*(E-L)A)(\xi\otimes\eta^*)=0\,.}
		
		To verify that $V\in\calB_{\ii\RR}$, we compute 
		\eq{
			VC&=\begin{bmatrix}\sgn(X)+E-L & \eta\otimes \xi^* \\ \xi\otimes \eta^* & \sgn(Y)-A^*(E-L)A\end{bmatrix} \begin{bmatrix}C & 0 \\ 0 & C\end{bmatrix}
			\\&=\begin{bmatrix}(\sgn(X)+E-L)C & (\eta\otimes \xi^*)C \\ (\xi\otimes \eta^*)C & (\sgn(Y)-A^*(E-L)A)C\end{bmatrix}
			\\&=\begin{bmatrix}C(-\sgn(X)+L-E) & -C(\eta\otimes \xi^*) \\ -C(\xi\otimes \eta^*) & C(-\sgn(Y)-A^*(L-E)A)\end{bmatrix} 
			\\&=-CV
		}
		where we use the fact that $AC=-CA$ and $C\eta=\eta$. Now that we have an appropriate SAU $V$, we follow \cref{eq:G invertible from U and V} and define $G:=\frac12(U+V)$. This self-adjoint operator enjoys all the properties discussed above of having $\findex_\Lambda G = 0$ and invertible. To see the invertibility, one follows a similar calculation to the one performed already twice above, so we omit it here.
				
		Using the invertibility of $G$ we may construct now a path from $U\rightsquigarrow V$ within $\LamNTSAU_{\ii\RR}$. Having deformed $U,\ti U$ into $V,\ti{V}$ respectively we now seek a unitary operator that conjugates $V$ into $\ti{V}$, and could be deformed to $\Id$ in a certain symmetric way. Using the RHS of \cref{eq:V non diagonal}, decompose the space $\ker\Lambda$ as \eq{\ker\Lambda = \im \chi_{\Set{1}}(V_{L}) \oplus \im \chi_{\Set{-1}}(V_{L}) \oplus \chi_{\Set{0}}(V_{L})=:\calG} where, in fact, $\im \chi_{\Set{0}}(V_{L})=\szpan(\eta)$. We may decompose $\ker \Lambda$ similarly according to $\ti{V}_{L}$, and we denote this grading as $\ti\calG$; note $\calG$ and $\ti\calG$ are \emph{the same space $\ker\Lambda$} which is merely graded differently. Since $V_LC=-CV_L$, it follows that $\chi_{\Set{1}}(V_{L})C=C \chi_{\Set{-1}}(V_{L})$. Since $\chi_{\Set{1}}(V_{L}),\chi_{\Set{1}}(\ti {V_{L}})$ are both infinite dimensional, there is a unitary $Z:\chi_{\Set{1}}(V_{L})\to\chi_{\Set{1}}(\ti{V_{L}})$. Using it, we define $W_L:\calG\to\ti\calG$ as \eq{W_L:=\begin{bmatrix}Z & 0 & 0 \\ 0 & CZC & 0 \\ 0 & 0 & \ti\eta\otimes\eta^\ast\end{bmatrix}\,.}
		It is clear that $W_L$ is unitary. In this grading, we can write $C,V_L:\calG\to\calG$ as \eq{C=\begin{bmatrix}0 & C_{\im \chi_{\Set{-1}}(V_{L})\to\im \chi_{\Set{1}}(V_{L})} & 0 \\ C_{\im \chi_{\Set{1}}(V_{L})\to\im \chi_{\Set{-1}}(V_{L})} & 0 & 0 \\ 0 & 0 & \Id\end{bmatrix},\quad V_L=\begin{bmatrix}\Id & 0 & 0 \\ 0 & -\Id & 0 \\ 0 & 0 & 0\end{bmatrix}\,.} A similar expression holds for $\ti \calG$ and $\ti{V_L},C:\ti\calG\to\ti\calG$. Now a direct computation shows that \eq{V_{L}=W_{L}^\ast\ti{V}_{L}W_{L},\quad W_{L}C=C W_{L}\,.} 
		
		We define $W_{R}:\im\Lambda\to\im\Lambda$ similarly. Plugging this into the RHS of \cref{eq:V non diagonal}, we find \eq{V &=\begin{bmatrix}W_{L}^\ast\ti{V}_{L}W_{L} & \eta\otimes \xi^* \\ \xi\otimes\eta^* & W_{R}^\ast\ti{V}_{R}W_{R}\end{bmatrix} \\ &=\begin{bmatrix}W_{L}^* & 0 \\ 0 & W_{R}^*\end{bmatrix} \begin{bmatrix}\ti{V}_{L} & W_{L} \eta\otimes \xi^* W_{R}^* \\ W_{R} \xi\otimes\eta^* W_{L}^* & \ti{V}_{R}\end{bmatrix}\begin{bmatrix}W_{L} & 0 \\ 0 & W_{R}\end{bmatrix}\,.} In fact, it holds that $W_{L} (\eta\otimes \xi^* )W_{R}^*=\ti{\eta}\otimes\ti{\xi}^*$. Thus $V=W\ti{V}W^*$ where $W:=W_{L}\oplus W_{R}$.
		
		We can use \cref{thm:Kuiper} to deform each $W_{L},W_{R}$ to $\Id$ such that the path commutes with $C$, which yields a path connecting $V\rightsquigarrow \ti V$ within $\LamNTSAU_{\ii \RR}$.
	\end{proof}

	\begin{lem}[$\Lambda$-non-triviality is well-defined]\label{lem:lambda-non-triviality is well-defined}
		Let $U\in\LamNTSAU$ and $W\in \calSU^\calL$. If $U-W$ is compact or sufficiently small in norm, then $W\in\LamNTSAU$ too. In particular, any continuous path in $\calSU^\calL$ starting within $\LamNTSAU$ is entirely contained within $\LamNTSAU$.
	\end{lem}
	\begin{proof}
		Let $U\in\LamNTSAU$. Then there exists $V$ as in \cref{def:Lambda-non-trivial SAUs} such that $[V,\Lambda]=0$, and $V$ is non-trivial when restricted to either $\im\Lambda$ or $(\im \Lambda)^\perp$, and $U-V\in\calK$. Let $W\in\calSU^\calL$. Decompose $W$ in $\calH= (\im\Lambda)^\perp \oplus \im\Lambda$ as \eq{\begin{bmatrix}W_{LL} & W_{LR} \\ W_{RL} & W_{RR}\end{bmatrix}\,.} Here $W_{ii}$ is self-adjoint and essentially unitary for $i=L,R$. By \cref{lem:stability of essentially non trivial projections} and the fact that $\norm{W_{ii}-U_{ii}}\leq\norm{W-U}$ for $i=L,R$, we conclude that $W_{ii}$ is essentially a non-trivial SAU.

		The compact statement is trivial.
	\end{proof}

	\section{Classification of bulk one-dimensional spectrally-gapped insulators}\label{sec:bulk insulators classification}
	We now come to the classification of one-dimensional insulators with a spectral gap. Let us begin with the general setup. We are interested in describing quantum mechanical systems of non-interacting electrons on a lattice, and hence we choose the Hilbert space $$ \calH := \ell^2(\ZZ^d)\otimes\CC^N $$ where $d$ is the space dimension and $N$ is the (fixed) number of internal degrees of freedom on each lattice site. The choice of the cubic lattice is made for simplicity of notation, since changing $N$ we may encode any graph via redimerization. What is however of importance is the fact $\ZZ^d$ has no boundary, which corresponds physically to \emph{bulk} systems. Later we comment briefly on edge systems in \cref{sec:1d edge classification}. We note that a classification of continuum systems with Hilbert space $L^2(\RR^d)\otimes\CC^N$ would also be interesting, especially since some of the features presented here seem to only emerge in the tight-binding setting,  see  \cite{ShapiroWeinstein2020,ShapiroWeinstein2022}.
	
	As was mentioned already above, locality plays a crucial role in our analysis. Physically it corresponds to the decaying probability of quantum mechanical transition between farther and farther points in space. There are various ways to encode locality of a quantum mechanical operator; Let $\Set{\delta_x}_{x\in\ZZ^d}$ be the singled-out position basis of Hilbert space, so that for any $A\in\calB$ and $x,y\in\ZZ^d$, the expression $A_{xy}$ corresponds to an $N\times N$ matrix whose matrix elements are $$ (A_{xy})_{ij} \equiv \ip{\delta_x\otimes e_i}{A \delta_y \otimes e_j}\qquad(i,j\in\Set{1,\dots,N}) $$  with $\Set{e_i}_{i=1}^N$ the standard basis for $\CC^N$. Now, the most straight forward way which is common in physics to specify locality is the nearest-neighbor constraint, i.e., $A\in\calB$ is local iff $$ A_{xy} = A_{xy}\chi_{\Set{0,1}}(\norm{x-y}) \qquad(x,y\in\ZZ^d)$$ where we take, say, the Euclidean norm on $\ZZ^d$ and $\chi$ is the characteristic function. Sometimes one prefers to consider \emph{finite hopping} operators, which are those operators $A\in\calB$ for which there exists some $R>0$ such that \eq{A_{xy} = A_{xy}\chi_{\RR_{\leq R}}(\norm{x-y}) \qquad(x,y\in\ZZ^d)\,.} In mathematics it is customary to consider the locality constraint as exponential decay of the off-diagonal matrix elements, i.e., that there exists some $C,\mu<\infty$ such that \eql{\label{eq:exp locality} \norm{A_{xy}}\leq C \exp\left(-\mu\norm{x-y}\right)\qquad(x,y\in\ZZ^d)\,. } Here we may choose any matrix norm for the LHS. This definition of locality is very natural and also facilitates the analysis on many occasions, it has appeared in various papers on topological insulators, e.g. \cite{EGS_2005,Shapiro20,FSSWY22,BSS23,Shapiro2019}.
	
	In choosing the correct definition of locality there is a certain art. If we were to insist on the above definition via exponential decay \cref{eq:exp locality}, the analysis becomes tedious and inelegant. Indeed, to drive this point further, and out of general interest, we explore this idea later in \cref{subsec:exp locality 1d chiral classification}. On the other hand one may define locality as that property of operators so that (together with the gap condition), topological indices are well-defined, which might lead to rather abstract topological analysis. Here we choose a middle ground which on the one hand leads to relatively natural functional analytic proofs and on the other hand is somewhat of a shadow of \cref{eq:exp locality}. We formulate it only in one and two dimensions here so as to avoid additional notational overhead which is anyway not necessary in the present paper, but see \cref{sec:higher odd d} below for the construction in higher dimensions.
	
	\begin{defn}[locality in $d=1$]\label{def:1d locality} Define an operator to be local iff it is $\Lambda$-local as in \cref{def:Lambda-local operators}, now with the particular choice $\Lambda:=\chi_\NN(X)$ where $X$ is the position operator on $\ell^2(\ZZ)$. Hence, $A\in\calB$ is local iff $[A,\Lambda]\in\calK$. 
	\end{defn}
	\begin{defn}[locality in $d=2$]\label{def:2d locality} Let $X_1,X_2$ be the two position operators on $\ell^2(\ZZ)$, with which $\Phi = \arg(X_1+\ii X_2)$ is the angle-position operator and $\ee^{\ii \Phi}$ is the phase position operator. An operator $A\in\calB$ is termed local now iff $[\ee^{\ii \Phi},A]\in\calK$.
	\end{defn}
	It is a fact that \cref{eq:exp locality} implies the compact commutator locality criterion: indeed, this is proven e.g. in \cite[Lemma 2 (b)]{Graf_Shapiro_2018_1D_Chiral_BEC} and \cite[Lemma A.1]{BSS23} for $d=1,2$ respectively. On the other hand it is certainly clear that these compact commutator notions of locality are strictly weaker than \cref{eq:exp locality}. From now on in this section $\calL$ stands for local operators with the compact commutator condition (very soon we will specify to $d=1$ and then we mean \cref{def:1d locality}).
	
	\begin{rem}[Compact commutator locality and the role of $N$]\label{rem:the role of N with the compact commutator condition} In our presentation so far the parameter $N$ is the internal fiber dimension, which physically could stand for spin, isospin, sub-lattice, or any other on-site internal degree of freedom of electrons. By requiring that operators are local via \cref{def:1d locality} instead of \cref{eq:exp locality}, we in principle allow $N$ to vary as we perform homotopies between operators. Indeed, by re-dimerization, given any operator presented on a Hilbert space with one given $N$ we may obtain another operator with any other $\widetilde{N}$ and clearly both would obey the compact commutator condition. This is thus a counter point of criticism on our K-theoretic-free analysis: why go through so much trouble to avoid K-theory if in the end anyway $N$ may effectively vary during homotopies? One response would be that unlike in K-theory our construction still calculates absolute rather than relative phases (we avoid the Grothendiek construction) and moreover, as explained, the calculation brings the topology defined on the set of operators to the foreground and as such may allow us to deal with the mobility gap regime.
	\end{rem}
	
	\begin{defn}[material]
		A material is then specified as a \emph{local} quantum mechanical Hamiltonian $H$ on $\calH$, i.e., some self-adjoint bounded linear operator $H=H^\ast\in\calL$.
	\end{defn}
	
	\subsection{Insulators}
	The space of all materials is too big to be topologically interesting (it is clearly nullhomotopic with straight-line homotopies). To further restrict it, we concentrate on \emph{insulators}: materials which exhibit zero direct current if electric voltage is applied. This statement needs to be qualified: due to the Pauli exclusion principle, electrons in a solid are characterized by a \emph{Fermi energy} $\mu\in\RR$, and so the same material could be both an insulator and a conductor when probed at different values of $\mu$. It turns out that for the purpose of conductivity, at a given $\mu$, it is equivalent to consider either $H$ at Fermi energy $\mu$, or $H-\mu\Id$ at Fermi energy $0$; clearly the latter operator is local too. Hence for the sake of simplicity we shall henceforth assume, without loss of generality, that the Fermi energy is always fixed at $\mu=0$. We note in passing that this assumption is not entirely benign when coupled with symmetries: further below we will see that certain symmetric operators have spectral symmetry about zero and then if one sets the Fermi energy at values other than zero one may obtain a different classification. 
	
	We identify two ways to encode the insulator (at $\mu=0$) condition: the spectral gap and the mobility gap criterions. The spectral gap condition is a simple constraint on the operator ${0\notin\sigma(H)}$, i.e., $H$ is an invertible operator. Since $\sigma(H)\in\Closed{\RR}$, this implies the existence of an open interval about zero which is not in the spectrum. The mobility gap condition is rather a constraint on the quantum dynamics associated with $H$, and is a set of almost-sure consequences for random ensembles of operators exhibiting Anderson localization. This condition was first presented in \cite{EGS_2005}. Since we will discuss the mobility gap regime specifically later in \cref{sec:mobility gap} let us continue with the general progression here and accept that insulators are
	\begin{defn}[insulators] A material $H=H^\ast\in\calL$ is an insulator iff it is invertible, i.e., if $0\notin\sigma(H)$. The space of all insulators is denoted by $\calI\equiv\calI_N$ (we mostly keep the fiber dimension $N$ implicit since it is fixed) and is endowed, as all other spaces, with the subspace topology from the operator norm topology on $\calB$.
	\end{defn}

	To each insulator $H$ we associate a \emph{Fermi projection} \eq{P\equiv P(H)  := \chi_{(-\infty,0)}(H)} which physically speaking corresponds to the Fermionic many-body ground state (density matrix) within the single-particle Hilbert space. Importantly, $P$ inherits locality from $H$: This is a consequence of \cref{lem:continuous functional calculus of normals is local} and the fact that under the assumption of a spectral gap, $\chi_{(-\infty,0)}(H)=f(H)$ with $f$ a continuous function differing from $\chi_{(-\infty,0)}$ on $\CC\setminus\sigma(H)$.

	At this point we specify to $d=1$. The task at hand is to calculate $\pi_0(\calI)$. According to the Kitaev table \cref{table:Kitaev} we should recover $\pi_0(\calI)=\Set{0}$. This is however not true at the level of generality we are working. Indeed, this is clear even without locality constraints: just take any insulator that has spectrum only above zero and another insulator that has spectrum only below zero: these two cannot be connected without passing with spectrum through zero and hence exiting $\calI$. A remedy would be to constrain to the space of insulators such that their Fermi projection is non-trivial as in \cref{def:non trivial projections}. But actually even this is still not enough: locality in one-dimension divides the system into left and right halves, and we should insist that our system is non-trivial on each side separately--this is the notion of $\Lambda$-non-trivial projections from \cref{def:Lambda-non-trivial SAUs} (adapted from SAUs to self-adjoint projections in an obvious way)--so that we are speaking about genuine \emph{bulk} systems rather than domain walls or edge systems. 
	
	\begin{example}[The necessity of $\Lambda$-non-triviality]
		Let $H := \Lambda - \Lambda^\perp$ and $\widetilde{H} := -\Lambda + \Lambda^\perp$. Both of these (flat) Hamiltonians are local (indeed, diagonal in space and in energy) and each has a Fermi projection which is non-trivial in the sense of \cref{def:non trivial projections}, because it has an infinite kernel and infinite range. However, on each half of space separately, the Fermi projections are trivial (just $\Id$ or $0$). 
		
		We claim that $H$ cannot be deformed into $\widetilde{H}$ without either closing the gap or violating locality.
	\end{example}
	
	\begin{proof}
		We prove the claim by contradiction: suppose there exists a continuous path $t\mapsto H_t$ that deforms $H$ to $\widetilde{H}$ such that $H_t$ is self-adjoint, invertible and local. Then $t\mapsto \frac{1}{2}(\Id-\sgn(H_t))$ is a continuous path of local self-adjoint projections that connects the Fermi projection of $H$ to $\widetilde{H}$, which we denote as $P$ and $\widetilde{P}$, respectively. Let us recall \cite[Proposition 2.2.6]{rordam2000introduction}, which says that for any C-star algebra $\calA$, if $A,B\in\calA$ are projections that are path-connected, then there exists a unitary in $\calA$ conjugating them. We apply this lemma on the C-star algebra $\calL$ to conclude that there exists some $U\in\calU^\calL$ such that $P=U^*\widetilde{P}U$. Decompose $U$ in $\calH=(\im \Lambda)^\perp\oplus \im \Lambda$ as \cref{eq:Lambda decomposition of U}. Writing out the equation $P=U^*\widetilde{P}U$ in this decomposition, we find \eq{\begin{bmatrix}\Id & 0 \\ 0 & 0\end{bmatrix} = \begin{bmatrix}U^*_{LL}& U^*_{RL} \\ U^*_{LR} & U^*_{RR}\end{bmatrix}\begin{bmatrix}0 & 0 \\ 0 & \Id \end{bmatrix} \begin{bmatrix}U_{LL}& U_{LR} \\ U_{RL} & U_{RR}\end{bmatrix} = \begin{bmatrix}U^*_{RL}U_{RL}& U^*_{RL}U_{RR} \\ U^*_{RR}U_{RL} & U^*_{RR}U_{RR}\end{bmatrix} \,.} Thus $U_{RL}^*U_{RL}=\Id.$ Now by assumption $U$ is local, which implies that $U_{RL}$ is a compact operator. However, $\Id$ is not compact on the infinite-dimensional space $(\im\Lambda)^\perp$. This leads to a contradiction.
	\end{proof}
	
	We thus define 
	
	\begin{defn}[bulk insulators]\label{def:bulk insulators} A material $H\in\calI$ is a bulk-insulator iff its Fermi projection is $\Lambda$-non-trivial, i.e., $\sgn(H)$ is a $\Lambda$-non-trivial SAU in the sense of \cref{def:Lambda-non-trivial SAUs}. We denote the space of bulk-insulators with $\calI_B\equiv\calI_{N,B}$: \eql{\calI_B := \Set{H\in\calI|\sgn(H)\in\LamNTSAU}} and furnish it also with the subspace topology.
	\end{defn}

	It will indeed emerge that in one space dimension, $\pi_0(\calI_B)=\Set{0}$ as stipulated by \cref{table:Kitaev}; this is one case of \cref{thm:main 1d insulators classification theorem} below.

	\subsection{The Altland-Zirnbauer symmetry classes} 
	Next we discuss the Altland-Zirnbauer symmetry classes \cite{AltlandZirnbauer1997} (AZ classes henceforth). The idea is that by restricting to a subspace, we could obtain non-trivial topology. From context of physics, naturally the subspaces of operators are those which obey certain symmetries. According to Wigner's theorem \cite{Bargmann1964}, a symmetry is a unitary or anti-unitary operator on $\calH$. Two basic operations coming from quantum field theory are time-reversal $\Theta$ and charge conjugation $\calC$ (which, in the context of solid state physics should be considered as particle-hole $\Xi$); the third one is parity which we do not need here. Naturally since the time evolution in quantum mechanics is implemented via $\exp(-\ii t H)$, $\Theta$ should be anti-unitary and $H$ is deemed ``time-reversal invariant'' iff it \emph{commutes} with $\Theta$. It was Dyson \cite{Dyson1962} who identified the two important cases $\Theta^2=\pm\Id$ which eponymously became known as Dyson's three-fold way (no $\Theta$ constraint or $[H,\Theta]=0$ with $\Theta=\pm\Id$). Altland and Zirnbauer \cite{AltlandZirnbauer1997} combined the three-fold way together with the charge-conjugation operator to form what is now known as \emph{the ten-fold way}. They considered many-body systems and Bogoliubov-de-Gennes (BdG) Hamiltonian description of superconductors, and in the context of which, one may think of particle hole $\Xi$ again as an anti-unitary operator which may square to $\pm\Id$, and commutes with $\Theta$. However, now, a Hamiltonian is deemed particle-hole symmetric iff it \emph{anti-commutes} with $\Xi$: \eql{\{H,\Xi\}\equiv H\Xi+\Xi H = 0\,.} The idea that a symmetry anti-commutes with a Hamiltonian may appear unnatural and at odds with basic notions of quantum mechanics--this is not how Altland and Zirnbauer phrased their many-body theory where all symmetries commute with the Hamiltonian; see \cite{Zirnbauer2021} for further discussion. Nonetheless it became quite established in modern condensed matter physics to use the anti-commutation condition as a convenient way to deal with particle-hole symmetry, \emph{and we will follow suit}. They then defined the chiral symmetry operator as the composition of the two \eq{\Pi := \Theta \Xi\,.} Since both $\Theta$ and $\Xi$ are anti-unitary, $\Pi$ is actually unitary and its square is of no consequence in the sense that $\{H,\Pi\}=0$ iff $\{H,\ii \Pi\}=0$. An interesting point is that one may consider a system which is chiral-symmetric (so it obeys $\{H,\Pi\}=0$ even though it has no further symmetries). Taking into account all possibilities (presence or absence of each symmetry constraint with each $\pm\Id$ version) we arrive at ten possibilities which are depicted in the first column of \cref{table:Kitaev}. These ten possibilities correspond to well-known structures in mathematics, such as the ten Morita equivalence classes of Clifford algebras \cite{AtiyahBottShapiro1964}, Cartan's ten infinite families of compact symmetric spaces \cite{Cartan1926,Cartan1927} and the ten associative real super division algebras \cite{Wall1964,Deligne1999}. The AZ labels themselves, by the way, come from Cartan.

	\begin{assumption}[Symmetries are strictly local] We shall assume that $\Theta,\Xi$ and $\Pi$ are strictly local, i.e., they \emph{commute} with the position operator $X$. Hence they can be considered as (anti-)unitary operators on $\CC^N$.
	\end{assumption}
	It would appear that most of the analysis should probably go through if it is only assumed that the commutator is compact: redimerization could make it hold if the symmetry operators have finite range.
	\begin{rem}\label{rem:conventional choices for Pi}
		In the foregoing discussion, we merely remarked that the sign of $\Pi^2$ is of no consequence to the analysis, and usually, when one presents the Kitaev table \cref{table:Kitaev} (as we did) one does not write out what $\Pi^2$ is, but rather only whether it is present or not. 
		
		It is however clear that if $\Theta$ and $\Xi$ are presumed to commute (as we indeed assume) then $\Pi^2=\Theta^2 \Xi^2$ and hence according to \cref{table:Kitaev} once $\Theta^2$ and $\Xi^2$ disagree, $\Pi^2=-\Id$. This however contradicts the ubiquitous convention of taking $\Pi=\Id\otimes\sigma_3$ which always squares to $+\Id$. Thus there are two possibilities: either take $\widetilde{\Pi}=\ii \Id\otimes\sigma_3$ for those AZ symmetry classes where $\Xi^2$ and $\Theta^2$ disagree ($\Set{H,\Pi}=0$ iff $\Set{H,\widetilde{\Pi}}=0$), or equivalently, for those AZ symmetry classes, take $\Set{\Theta,\Pi}=0$ instead of $[\Theta,\Pi]=0$.
		
		To preserve notational simplicity, we found it more convenient to \emph{always} assume that $\Pi=\Id\otimes\sigma_3$ and when necessary, employ $\Set{\Theta,\Pi}=0$; this convention follows, e.g., \cite{KatsuraKoma2018}. This explains the following assumption.
	\end{rem}

	\begin{assumption}\label{ass:chirality has same dimensions for its eigenspaces}
		We assume that $\Pi$ has $\pm1$ eigenspaces of the same dimension, and that there is a unitary mapping between the two $\Pi$ eigenspaces which commutes with both $\Theta$ or $\Xi$.
	\end{assumption}
	\begin{defn}[symmetric insulators] To each of the AZ symmetry classes \eq{\mathrm{AZ}:=\Set{\mathrm{A},\mathrm{AI},\mathrm{AII},\mathrm{AIII},\mathrm{BDI},\mathrm{D},\mathrm{DIII},\mathrm{C},\mathrm{CI},\mathrm{CII}}} we define the class of bulk-insulators which obey that symmetry and label it by \eq{\calI_{B,\Sigma}\equiv\calI_{N,B,\Sigma}\qquad(\Sigma\in\mathrm{AZ})\,.}
	\end{defn}
	
	The main result of this section is
	\begin{thm}[The one-dimensional column of the Kitaev table]\label{thm:main 1d insulators classification theorem} At each fixed $N$, for any $\Sigma\in\mathrm{AZ}$, the path-connected components of $\calI_{N,B,\Sigma}$ considered with the subspace topology associated with the operator norm topology, agree with the set appearing in the first column of \cref{table:Kitaev}.
	\end{thm}
	We stress that while \cref{table:Kitaev} was derived using K-theory of C-star algebras, here we make no recourse to K-theory and rely entirely on homotopies of operators. In particular, the classification we derive is not relative and does not rely on extended degrees of freedom (for us $N$ in $\calH=\ell^2(\ZZ)\otimes\CC^N$ is fixed once and for all throughout the analysis). While these two points might not exactly appeal to specialists in K-theory, what is perhaps more interesting is the perspective on the mobility gap regime, see \cref{sec:mobility gap}.
	
	The rest of this section is dedicated to proving \cref{thm:main 1d insulators classification theorem} using the results presented in \cref{sec:local unitaries classification,sec:local SA unitaries classification}. In \cref{subsec:exp locality 1d chiral classification} we present a completely different approach which assumes a different mode of locality via \cref{eq:exp locality}.
	
	\paragraph{Examples of concrete physical models.} In order to connect with concrete literature in physics, we point out that
	\begin{enumerate}
		\item In class AIII, the Hamiltonian is of the form \eq{H = \begin{bmatrix}
				0 & S^\ast \\
				S & 0
		\end{bmatrix}} and the associated index is \eq{\findex\bbLambda \polar(S)\in\ZZ\,.} This index is widely known as the \emph{Zak phase} \cite{Zak1989}, and in the translation-invariant setting reduces to a winding number. A widely popular model which exhibits a non-trivial Zak phase is the SSH model \cite{SSH1979} of polyacetylene.
		\item In class D, the associated index is \eq{\findex_2 \bbLambda \sgn(H) \in \ZZ_2\,.} This index is widely known as the \emph{Majorana number} and a quintessential model which exhibits it is the \emph{Kitaev chain} \cite{Kitaev2001}.
	\end{enumerate}
	
	\subsection{Flat Hamiltonians}
	Sometimes in physics there is a distinction between classifying Hamiltonians and classifying \emph{ground states}, which, in the single-particle context correspond to the associated Fermi projections. As we will see now, for us this distinction does not exist since we are working in the spectral gap regime.
	
	We say a Hamiltonian $H$ is \emph{flat} iff $\sgn(H)=H$ where $\sgn$ is the sign function (its value at zero is of no consequence since our Hamiltonians have no spectrum there). We denote the space of all flat bulk-insulators by $\calI_B^\flat$. We note that if $H$ is flat then its Fermi projection $P$ is given by $ P = \frac{1}{2}\left(\Id-H\right)$ so flat Hamiltonians are algebraically related to their Fermi projections.
	\begin{lem}\label{lem:flat Hamiltonians are strong deformation retraction} Flat insulators are a strong deformation retract of insulators. This statement remains true if we add the bulk-insulator constraint as well as any of the ten AZ symmetry constraints: $\calI_{B,\Sigma}^\flat$ is a strong deformation retraction of $\calI_{B,\Sigma}$ for any $\Sigma\in\mathrm{AZ}$. 
	\end{lem}
	\begin{proof}
		The desired retraction \emph{is} in fact $\sgn$, which (via the functional calculus) may be considered a map $\calL\to\calL$ (see \cref{lem:continuous functional calculus of normals is local}). 
		
		Hence, given $H\in\calI_{B,\Sigma}$, one has \eq{ \sgn(H) = \Id + \frac{1}{\pi \ii}\oint_\Gamma R(z)\dif{z}} where $\Gamma$ is any CCW path encircling $\sigma(H)\cap(-\infty,0)$ and $R(z)\equiv(H-z\Id)^{-1}$. From this formula and the resolvent identity norm continuity easily follows. Since $\sgn\circ\sgn=\sgn$, this is indeed a retraction; note that since $\sgn$ is odd, $\sgn(H)$ would obey the same AZ constraint that $H$ would.
		
		Next, define $F:\calI_{B,\Sigma}\times[0,1]\to\calI_{B,\Sigma}$ via \eq{F(H,t) := (1-t)H+t\sgn(H)\qquad(H\in\calI_{B,\Sigma},t\in[0,1])\,.} It is well-defined since \eq{\chi_{(-\infty,0)}(F(t,H)) = \chi_{(-\infty,0)}(H)\qquad(t\in[0,1])} and $F(H,0)=H$, $F(H,1)=\sgn(H)$ and $F(\sgn(H),1)=\sgn(H)$. Since the bulk-insulator condition is defined in terms of the flat Hamiltonian and not the Hamiltonian itself, $F(t,H)$ is a bulk-insulator for all $t$.
	\end{proof}
	
	Clearly the path-connected components of a space and those of its retract are the same, and hence in proving \cref{thm:main 1d insulators classification theorem}, we could just as well work with $\calI^\flat_{B,\Sigma}$. This last fact makes the analysis reduce to the study of $\Lambda$-non-trivial equivariant self-adjoint unitaries.
	
	\subsection{Classification of the non-chiral classes} 
	\begin{table}
		\begin{center}
			\begin{tabular}{|c|c|c|}
				\hline
				$\FF$ & AZ Class 
				& Topological invariant\\\hline\hline
				$\CC$ & A & - \\\hline
				$\RR$ & AI & -\\\hline
				$\HH$ & AII & - \\\hline
				$\ii\RR$ & D & $\findex_2 \bbLambda \sgn(H)$\\\hline
				$\ii\HH$ & C & -\\
				\hline
			\end{tabular}
		\end{center}
		\caption{Correspondence between the operators defined in \cref{eq:real and quaternionic operators,eq:pure imaginary real or quaternionic operators} and the non-chiral AZ symmetry classes classes, and formulas for the topological invariants.}
		\label{table:AZ non-chiral classes and symmetries of projections}
	\end{table}
	The non-chiral classes are those within the AZ classes where $\Pi$ is absent: classes A,AI,AII, C and D. In this case $\calI_{B,\Sigma}^\flat$ is the same space as $\LamNTSAU_\FF$ with the appropriate correspondence between $\Sigma$ and $\FF$ as depicted in \cref{table:AZ non-chiral classes and symmetries of projections}. When $\Theta$ squares to $\pm\Id$, we have a real (resp. quaternionic) structure and that corresponds to the anti-unitary operator $C$ (resp. $J$) of \cref{eq:real and quaternionic operators}. On the other hand, the presence of a particle-hole symmetry corresponds rather to $\sgn(H)$ belonging to the purely-imaginary real or quaternionic sets of operators in \cref{eq:pure imaginary real or quaternionic operators}.
	
	We find that for the non-chiral classes our theorem is complete via \cref{sec:local SA unitaries classification} and in particular the results there which are summarized in \cref{table:equivariant Lambda-non-trivial SAUs}.
	
	\subsection{Classification of the chiral classes}
	Now we assume that $\Pi$ is present, i.e., that we are in such AZ classes where insulators obey $\{H,\Pi\}=0$. Thanks to \cref{ass:chirality has same dimensions for its eigenspaces}, it must be that $N=2W$ for some $W\in\NN_{\geq1}$, and so the Hilbert space breaks into a direct sum \eq{
		\calH = \left(\ell^2(\ZZ)\otimes\CC^W\right)\oplus\left(\ell^2(\ZZ)\otimes\CC^W\right)\,.
	} We formally refer to the left copy as ``positive chirality'' and the other as ``negative chirality'', and use $\calH_{\pm}$ for these two. Since they are isomorphic in a local way, we will actually drop the distinction between them. By a local (at the level of $\CC^N$) unitary transformation on $\calH$ we may without loss of generality assume that $\Pi$ is diagonal, i.e., acting as $\Id_W\oplus(-\Id_W)$ on each local copy of $\CC^W\oplus\CC^W$. Hence it must be that insulators which are chiral have the form \eq{ H = \begin{bmatrix}
			0 & S^\ast \\
			S & 0
	\end{bmatrix}} for some $S\in\calB(\ell^2(\ZZ)\otimes\CC^W)$ which is \emph{not necessarily self-adjoint} (note how in writing $S$ we dropped the distinction between the positive and negative chiralities). Moreover, in this chiral grading, $\Lambda$ is diagonal.
	
	Clearly, the spectral gap condition on $H$ translates to $S$ being invertible since $|H|^2=|S|^2\oplus|S^\ast|^2$, and $H$ is local iff $S$ is.
		%
	Moreover, via \cite[Lemma 2]{Graf_Shapiro_2018_1D_Chiral_BEC}, \eql{\label{eq:correspondence between chiral flat hamiltonian and off-diagonal sub-block} \sgn(H) = \begin{bmatrix}
			0 & \polar(S)^\ast \\
			\polar(S) & 0
	\end{bmatrix} } where $\polar(S)\equiv S|S|^{-1}$ is the polar part of $S$, which is in our setting unitary since $S$ is invertible.
	
	Finally, it is interesting to note that 
	\begin{lem}\label{lem:chiral insulators are automatically bulk non trivial}
		If $H\in\calI$ has chiral symmetry then it is a bulk-insulator automatically. Thus, the bulk-insulator constraint is vacuous within the chiral classes.
	\end{lem}
	\begin{proof}
		Using the same $\im(\Lambda)^\perp\oplus\im(\Lambda)$ decomposition of the Hilbert space as in \cref{eq:space decomposition of SA unitary}, we write \eq{\sgn(H) = \begin{bmatrix}X & A \\ A^\ast & Y\end{bmatrix}\,.} We caution the reader that only within this proof and unlike the rest of this section, $X,Y$ are blocks in $\sgn(H)$ and \emph{not} position operators. 
		
		Since $[\Pi,\Lambda]=0$, we conclude that $\Set{X,\Pi}=\Set{Y,\Pi}=0$. Thus, both $X$ and $Y$ have spectra which is symmetric about zero and contained in $[-1,1]$, and, is discrete on $(-1,1)$. But since, by assumption, both $\im(\Lambda)^\perp$ and $\im(\Lambda)$ are infinite-dimensional, it cannot be that either $X$ or $Y$ have only discrete spectrum. As a result, we conclude that both $X$ and $Y$ have both $\pm1$ in their essential spectrum. We conclude by \cref{lem:stability of essentially non trivial projections}, which implies that $X,Y$ are essentially non-trivial self-adjoint unitaries.
		
	\end{proof}
	
	\begin{table}
		\begin{center}
			\begin{tabular}{|c|c|c|}
				\hline
				$\FF$ & AZ Class & Topological invariant \\\hline\hline
				$\CC$ & AIII & $\findex\bbLambda\polar(S)$ \\\hline
				$\RR$ & BDI & $\findex\bbLambda\polar(S)$ \\\hline
				$\HH$ & CII&$\findex\bbLambda\polar(S)\in2\ZZ$ \\\hline
				$\star\RR$ & CI&-\\\hline
				$\star\HH$ & DIII&$\findex_2\bbLambda\polar(S)$ \\\hline
			\end{tabular}
		\end{center}
		\caption{Correspondence between the operators defined in \cref{eq:real and quaternionic operators,eq:star real and quaternionic operators} and the chiral AZ symmetry classes, and formulas for the topological invariants. Here $S$ stands for the off-diagonal block within $H$ in the presence of chiral symmetry, and $\polar$ is its polar part.}
		\label{table:AZ chiral classes and symmetries of unitaries}
	\end{table}

	\begin{lem}
		For $\Sigma$ in the chiral classes, the space $\calI_{B,\Sigma}^\flat$ is homeomorphic to $\calU_\FF^\calL$ with the correspondence between $\Sigma$ and $\FF$ as depicted in \cref{table:AZ chiral classes and symmetries of unitaries}. 
	\end{lem}
	\begin{proof}
		Most of the necessary statements for the proof have just appeared above so we really only need to focus on the correspondence between the physical symmetry classes of $\Theta$ and $\Xi$ versus the abstract real and quaternionic operator classes defined in \cref{sec:locality and symmetry}.
		
		Clearly for $\Sigma=\mathrm{AIII}$ the mapping given by \eql{\label{eq:mapping insulator to unitary}\calI_{B,\mathrm{AIII}}^\flat\ni\begin{bmatrix}
				0 & U^\ast \\ U & 0 
			\end{bmatrix}\mapsto U\in\calU^\calL_\CC} is the required homeomorphism, which is indeed a homeomorphism: well-definedness and bijectivity follow by the foregoing discussion and continuity is clear.
		
		We proceed with the other four choices of $\Sigma$. By \cref{rem:conventional choices for Pi} and \cref{ass:chirality has same dimensions for its eigenspaces}, we are left only to check what $\Theta$ squares to and whether it commutes or anti-commutes with $\Pi$: four possibilities. As was explained in \cref{rem:conventional choices for Pi}, when, in \cref{table:Kitaev}, $\Xi^2$ and $\Theta^2$ disagree we should take $\Set{\Theta,\Pi}=0$ and when they agree we take $[\Theta,\Pi]=0$. 
		
		Let us write the time-reversal symmetry operator in the chiral grading as \eq{\Theta=\begin{bmatrix}\Theta_{++} & \Theta_{+-} \\ \Theta_{-+} & \Theta_{--}\end{bmatrix}\,.} 
		\begin{enumerate}
			\item When $[\Pi,\Theta]=0$ (classes BDI and CII) we have $\Theta_{+-}=\Theta_{-+}=0$. In this case, we define $F:=\Theta_{++}=\Theta_{--}$ (they are the same by \cref{ass:chirality has same dimensions for its eigenspaces}) and we find that under the mapping \cref{eq:mapping insulator to unitary} the condition $[\sgn(H),\Theta]=0$ implies $UF=FU$, i.e., $U$ is either a real or quaternionic operator based on the value of $F^2$: for $\Theta^2=\Id$ (class BDI) we get $U\in\calU^\calL_{\RR}$ and for $\Theta^2=-\Id$ (class CII) we get $U\in\calU^\calL_{\HH}$.
			\item When $\Set{\Pi,\Theta}=0$ (classes DIII and CI) we have $\Theta_{++}=\Theta_{--}=0$. \cref{ass:chirality has same dimensions for its eigenspaces} allows us further to avoid notation overhead since $\Theta_{+-}=-\Theta_{-+}^\ast=:F$. In this case, however, $[\sgn(H),\Theta]=0$ implies $UF=FU^\ast$, which is precisely the $\star$-real or $\star$-quaternionic condition, based on $F^2=\pm\Id$, which is equal to the value of $\Theta^2$. Hence we find that for $\Theta^2=\Id$ (class CI) $U\in\calU^\calL_{\star\RR}$ and for $\Theta^2=-\Id$ (class DIII), $U\in\calU^\calL_{\star\HH}$.
		\end{enumerate}
	\end{proof}
	
	Now as a result of the statements in \cref{sec:local unitaries classification}, the proof of \cref{thm:main 1d insulators classification theorem} is complete.
	
	\subsection{Classification of exponentially local chiral insulators}\label{subsec:exp locality 1d chiral classification}
	Our theory so far has involved the one-dimensional locality condition \cref{def:1d locality}. This condition may appear somewhat contrived from the physical stand point, in the sense that all it asks is that Hamiltonians $H$ obey $\Lambda H \Lambda^\perp \in \calK$. This condition may be criticized (and we would agree, rightly so) that too much of the physics has been washed away.
	
	In this subsection we address this issue as follows: we consider one-dimensional operators with exponential locality as in \cref{eq:exp locality}, \emph{but only in class AIII} for simplicity. Indeed, this type of endeavor is somewhat perpendicular to the activity of topological classification, and is more related to a study of regularity and approximation. In the commutative setting this would be tantamount to a type of Whitney approximation theorem saying that for any two smooth manifolds $X,Y$ with $\partial Y = \varnothing$, any continuous map $X\to Y$ is continuously homotopic to a smooth map $X\to Y$. For that reason we restrict ourselves here merely to one non-trivial symmetry class rather than repeat the analysis for all the AZ classes.
	
	Hence, let us define 
	\begin{defn}[exponentially local insulators] An exponentially local insulator is a self-adjoint Hamiltonian $H=H^\ast\in\calB$ which is spectrally gapped (at zero) and for which exponential locality \cref{eq:exp locality} holds with any rate:\eq{\inf_{x,y\in\ZZ^d}-\frac{1}{\norm{x-y}}\log\left(\norm{H_{xy}}\right)>0\,.}
		We denote this space by $\calI_{\mathrm{exp}}$ and furnish it with the subspace topology from the operator norm topology. We note that at least for the chiral classes there is no need to speak of the bulk-insulator condition thanks to \cref{lem:chiral insulators are automatically bulk non trivial} (recall exponentially local operators are $\Lambda$-local).
	\end{defn}
	It is a fact that this space is strictly smaller than the one obtained with \cref{def:1d locality}. Indeed, an explicit example may be constructed \cite[Example 3.3.10]{Geib2022}.
	
	The classification result for exponentially local chiral operators is thus:
	\begin{thm}[AIII $d=1$ exp. local classification] In $d=1$, at fixed $N\in\NN$, the space $\calI_{\mathrm{exp,AIII}}$ has $\ZZ$ path components labeled by the norm continuous map \eq{\calI_{\mathrm{exp,AIII}}\ni \begin{bmatrix}
				0 & S^\ast\\S&0
			\end{bmatrix} \mapsto \findex \bbLambda S \in \ZZ\,. } 
	\end{thm}
	We note that in this theorem, it is easier for us to deform within the space of exponentially local \emph{invertible} operators rather than exponentially local \emph{unitary} operators as in the preceding proof. This is no major disadvantage though, since our ultimate goal is $\calI_{\mathrm{exp,AIII}}$ rather than $\calI_{\mathrm{exp,AIII}}^\flat$.
	\begin{proof}
		Similarly to the proof of \cref{thm:classification of RCH local unitaries}, the continuity, surjectivity and logarithmic law for $\findex\circ\bbLambda$ are established, so we are really only concerned with injectivity of the map at the level of the path-components.
		
		Hence, let $S\in\calB(\ell^2(\ZZ)\otimes\CC^W)$ be invertible and have zero index, and our goal is to continuously connect it to $\Id$ within the space of exponentially local invertibles.
		
		Let $R$ be the bilateral right shift operator on $\ell^2(\ZZ)$. Then clearly we may write \eq{ S = \sum_{l\in\ZZ} S_l R^l} where for each $l\in\ZZ$, $S_l$ is a diagonal operator given via its matrix elements \eq{ (S_l)_{xy} = \delta_{xy} S_{x,x-l}\qquad(x,y\in\ZZ)\,.} The series converges in operator norm thanks to exponential locality. Hence, given any $\ve>0$, there is some $L_\ve>0$ such that the $L_\ve$ hopping operator $S^{L_\ve} := \sum_{|l|\leq L_\ve}S_lR^l$ is $\ve$-close to $S$: \eq{ \norm{S - S^{L_\ve}}<\ve\,.} Moreover, the straight line homotopy \eq{[0,1]\ni t \mapsto (1-t) S + t S^{L_\ve}} is clearly norm continuous, and passes through locals. It passes through invertibles too if $\ve$ is chosen sufficiently small since these are open. This shows that without loss of generality we may assume that $S$ is of finite hopping. 
		
		Next, for any finite hopping operator, there is an integer $\widetilde{W}$ (in particular, $\widetilde{W}=W\times L_\ve$) and a local unitary transformation \eq{U:\ell^2(\ZZ)\otimes\CC^{\widetilde{W}}\to\ell^2(\ZZ)\otimes\CC^{W}} which does not affect the $\Lambda$-index (this is ``redimerization'') such that \eq{ U^\ast S^{L_\ve} U = A+BR+CR^\ast} where $\widetilde{A},\widetilde{B},\widetilde{C}$ are diagonal operators (so, sequences $\ZZ\to\CC^{\widetilde{W}}$). Indeed this map is \eq{U^\ast(\psi) \equiv (\dots,\begin{bmatrix}\psi_{0}\\\dots\\\psi_{L_\ve-1}\end{bmatrix},\begin{bmatrix}\psi_{L_\ve}\\\dots\\\psi_{2L_\ve-1}\end{bmatrix},\begin{bmatrix}\psi_{2L_\ve}\\\dots\\\psi_{3L_\ve-1}\end{bmatrix},\dots)\,.}
		
		Let us factor out a left shift operator \eq{A+BR+CR^\ast = (AR + BR^2 + C)R^\ast} and consider another redimerization $V:\ell^2(\ZZ)\otimes \CC^{2\ti{W}}\to \ell^2(\ZZ)\otimes \CC^{\ti{W}}$ so that we can write \eq{V^* (AR + BR^2 + C)V = \widehat{A}+\widehat{B}R} where $\widehat{A},\widehat{B}:\ZZ\to \CC^{2\ti{W}}$ are diagonal operators (the transition from $A+BR+CR^\ast$ to $A+BR$ was described in \cite[Example 1]{Graf_Shapiro_2018_1D_Chiral_BEC}). In particular \eq{\findex \bbLambda (\widehat{A}+\widehat{B}R)=-\ti{W}} due to the factoring-out of a left shift operator which has index $+\ti W$. Using \cref{lem:A+BR homotopies} below, we can deform $\widehat{A}+\widehat{B}R$ to $D^\perp+DR$ where $D:\ZZ\to\CC^{2\ti{W}}$ is diagonal with \eq{D_n=\begin{bmatrix}0_{\ti{W}} & 0 \\ 0 & \Id_{\ti{W}}\end{bmatrix}\qquad(n\in\ZZ)} and the deformation is within the space of invertible operators of the same nearest-neighbor form as $\widehat{A}+\widehat{B}R$. In conclusion, we have the path \eq{U^\ast S^{L_\ve} U = A+BR+CR^\ast\rightsquigarrow V(D^\perp +DR)V^*R^*} within the space of exponentially local invertibles. This last operator, however, is readily seen to be equal to \eq{\dots\oplus\begin{bmatrix}
				0 & \Id_{\ti W} \\ \Id_{\ti W} & 0
			\end{bmatrix}\oplus\begin{bmatrix}
				0 & \Id_{\ti W} \\ \Id_{\ti W} & 0
			\end{bmatrix}\oplus \dots\,.} Each such $2\times 2$ block may be deformed (by a \emph{local} Kuiper) to $\Id_{2\ti W}$ within the space $\calU(2\ti W)$ and hence the whole operator to $\Id$ within the infinite $\calU$ (so that, in particular, this deformation passes within invertibles and exponentially locals). We have thus exhibited a path $U^\ast S^{L_\ve} U\rightsquigarrow\Id_{\ell^2(\ZZ)\otimes\CC^{\ti W}}$. Conjugating that path with $U$ and composing with the straight line path from $S$ to $S^{L_\ve}$ we obtain a path in the original Hilbert space $S\rightsquigarrow \Id_{\ell^2(\ZZ)\otimes\CC^{W}}$.
	\end{proof}

	\begin{lem}[$A+BR$ homotopies]\label{lem:A+BR homotopies}
		On the space $\ell^2(\ZZ)\otimes \CC^K$, let $R$ be the bilateral right-shift operator on $\ell^2(\ZZ)$ and $A,B:\ZZ\to \Mat_K(\CC)$ be diagonal operators.  For any $k\in\ZZ$, the space of invertible operators on $\ell^2(\ZZ)\otimes \CC^K$ of the form $A+BR$ having $\findex_\Lambda(A+BR)=k$ is path-connected.
	\end{lem}
	
	\begin{proof}
		We first note that if $A+BR$ is invertible, then $A+\lambda BR$ is invertible for any $\lambda\in\bS^1$. To see this, for any $\lambda\in\bS^1$ we define the diagonal operator $U_\lambda:= \lambda^X\otimes \Id_K$ where $X$ is the position operator on $\ell^2(\ZZ)$. I.e., \eq{(U_\lambda)_{nm}:=\delta_{nm}\lambda^n \Id_K\qquad(n,m\in\ZZ)\,.} Since $R^\ast X R = X+\Id$ we have $R^\ast U_\lambda R  = \lambda U_\lambda$. Thus, using the fact that the diagonal operators $A,B$ commute with $U_\lambda$ we find \eq{A+\lambda BR = A + BR \underbrace{\left(R^\ast U_\lambda R U_\lambda^\ast\right)}_{=\lambda\Id}=A U_\lambda U_\lambda^\ast +B U_\lambda R U_\lambda^\ast = U_\lambda\left(A+BR\right)U_\lambda^\ast } from which we conclude $A+\lambda BR$ is invertible for all $\lambda\in\bS^1$.
		
		It is hence justified to apply 
		\cref{thm:stummel} below, with $G=BR$ to obtain idempotents $P,Q$ with respect to which $A$ and $BR$ are diagonal according to the grading \cref{eq:stummel grading}, that is, we have \eq{A&=Q^\perp A P^\perp + QAP\\ BR&= Q^\perp BR P^\perp + QBRP} (with $P^\perp\equiv\Id-P$, despite not being orthogonal) and\eql{\label{eq:stummel invertible maps}
			\forall \lambda\in\overline{\DD}\,,\qquad\begin{cases}
				Q^\perp (A+\lambda BR)P^\perp &:\im P^\perp\to\im Q^\perp\\ \quad Q(\lambda A+BR)P&:\im P\to \im Q    
			\end{cases}\text{ are both invertible.}} 
		We therefore consider the path \eql{\label{eq:path from A+BR to disjointed decomposition}[0,1]\ni t\mapsto Q^\perp(A+(1-t)BR)P^\perp + Q((1-t)A+BR)P} which is invertible by construction (each of its two diagonal blocks is separately invertible), and as we shall see, is also of the $A+BR$-hopping form. Indeed, we argue that the idempotents $P,Q$ are diagonal in space. We note in passing that this fact was observed in \cite{ben1991band}, from which we draw inspiration. Let $\bbDelta :\calB\to\calB$ be the super-operator which projects an operator to its diagonal part, i.e., for any operator $F\in\calB$, \eq{(\bbDelta F)_{nm}\equiv F_{nm}\delta_{nm}\qquad(n,m\in\ZZ)\,.} Clearly $R^\ast \left(\bbDelta F\right) R =  \bbDelta \left( R^\ast F R\right) $. Then we note the Cauchy identity \eq{\frac{1}{2\pi \ii}\oint_{\bS^1} U_\lambda FU_\lambda^* \dif{\lambda} = R^\ast\bbDelta F} which follows from the calculation on the $n,m$ matrix elements \eq{\left(\frac{1}{2\pi \ii}\oint_{\bS^1} U_\lambda FU_\lambda^* \dif{\lambda}\right)_{nm}=\left(\frac{1}{2\pi \ii}\oint_{\bS^1} \lambda^{n-m} \dif{\lambda}\right)F_{nm}=\delta_{n-m,-1}F_{nm} = \left(R^\ast \bbDelta F\right)_{nm}} and so using the definition of $P$ from \cref{eq:stummel projections} we have 
		\eq{
			P \equiv \frac{1}{2\pi \ii}\oint_{\bS^1} (A+\lambda BR)^{-1} BR\dif{\lambda} = \frac{1}{2\pi \ii}\oint_{\bS^1} U_\lambda (A+BR)^{-1}U_\lambda^\ast \dif{\lambda} BR &= R^\ast \left(\bbDelta\left(A+BR\right)^{-1} \right)BR \\&= \bbDelta\left(R^\ast \left(A+BR\right)^{-1}BR\right)\,.
		} 
		Similarly, the idempotent $Q= \bbDelta B\left(A+BR\right)^{-1}$ and is in particular diagonal. Therefore, \cref{eq:path from A+BR to disjointed decomposition} indeed passes within the space of operators of the form $A+BR$ and yields, with the diagonal idempotent $\Pi:= R P R^\ast = \bbDelta\left(\left(A+BR\right)^{-1}B\right)$, the path \eq{A+BR\rightsquigarrow Q^\perp AP^\perp + QBRP=Q^\perp AP^\perp + QB\Pi R\,.} Our goal now is to further deform the two diagonal operators $Q^\perp A P^\perp,QB\Pi$ into $D^\perp,D$ where \eql{\label{eq:def of D}D=\Id\otimes M} for some constant $K\times K$ matrix given by $M:=0_{K-p}\oplus\Id_p$ for some $p=0,\dots,K$.
		
		To get to $D^\perp+DR$, we note that $QB\Pi:\im \Pi\to\im Q$ and $Q^\perp A P^\perp:\im P^\perp \to \im Q^\perp$ are both invertible since they are the point $\lambda=0$ of \cref{eq:stummel invertible maps}. Since all these operators are diagonal, they define a sequence of invertible maps $Q_n B_n P_{n-1} : \im P_{n-1}\to\im Q_n$ and $Q_n^\perp A P_n^\perp : \im P_n^\perp \to \im Q_n^\perp$. For every $n\in\ZZ$, the invertibility of the first matrix implies that $\dim\im P_{n-1} = \dim Q_n$ and of the second $\dim \im P_{n}^\perp = \dim \im Q_n^\perp$ and hence $\dim \im P_{n} = \dim \im Q_n$ so that \eq{\ZZ\ni n\mapsto \dim \im P_n = \dim\im Q_n \in \NN} is a constant sequence, whose constant value we denote by $p\in\NN$. Our goal is to deform the two sequences of invertible maps above into the trivial ones given by $n\mapsto M,M^\perp$; to do one has to deform both the vector subspaces on which these maps are invertible so as to make them constant in $n$ (this amounts to a unitary conjugation of $P_n,Q_n$), and then deform the maps within the respective subspaces to $\Id$ resp. $0$. Indeed, the constant rank of $Q_n,P_n$ implies that we can find unitaries $U,V$ which are diagonal in space and so that $P = U^\ast D U$ and $Q = V^\ast D V$. As such, for each $n\in\ZZ$, the $K\times K$ matrices $(V A U^\ast)_n$ and $(V B R U^\ast R^*)_n$ both restrict to invertible maps $\im M^\perp \to \im M^\perp$ and $\im M \to \im M$ respectively. Using Kuiper we deform each of these matrices to $M^\perp$ and $M$ respectively to obtain a local invertible deformation  $A+BR\rightsquigarrow D^\perp+D R$.
		
		Let us see that the path we describe above indeed passes through invertibles. We have so far
		\eq{
			A+BR & \rightsquigarrow Q^\perp AP^\perp + QBRP  
			\\ & =  (V^*D^\perp V) A (U^* D^\perp U) + (V^*D V) BR (U^* D U) 
			\\ & =  V^*(D^\perp V A U^* D^\perp) U + V^* (D V BR U^* D) U 
			\\ & \rightsquigarrow D^\perp (V A U^*) D^\perp + D (V BR U^*) D
		} where the last deformation follows by a Kuiper to get $V,U\rightsquigarrow\Id$. 
		Now we note
		\eq{
			VAU^*:\ker D\to\ker D,\quad V BR U^*:\im D\to \im D
		}
		are invertible by construction. We note that $V BR U^*$ is not diagonal and shall be deformed to $DR$ rather than $D$. To this end, write $V BR U^*=(V BR U^*R^*)R$ and note $V BR U^*R^*:\im D\to\im D$ is invertible and diagonal. Deform within diagonal invertibles in each respective subspace
		\eq{
			VAU^*\rightsquigarrow D^\perp : \ker D\to\ker D,\quad V BR U^*R^* \rightsquigarrow D:\im D\to\im D
		}
		Then $(V BR U^*R^*) R \rightsquigarrow DR:\im D\to\im D$ within invertibles. Thus
		\eq{
			D^\perp (V A U^*) D^\perp + D (V BR U^*) D &\rightsquigarrow D^\perp (D^\perp) D^\perp + D (DR) D \\
			&=D^\perp + D RD = D^\perp +DR
		}
		since $RD=DR$. The path is invertible since we treat $\ker D\to\ker D$ and $\im D\to \im D$ separately.
		
		Since $\findex_\Lambda (A+BR)=k$ is preserved under norm continuous deformations, its value is equal to $\findex_\Lambda (D^\perp + DR)=p-K$, where we identified $p=K+k.$ Any two invertible operators $S,T$ of the form $A+BR$ with the same index can be deformed to the same $D^\perp+D R$, and hence can be deformed to each other.
	\end{proof}
	
	\section{Classification of one-dimensional edge systems}\label{sec:1d edge classification}
	A prominent feature of topological insulators is the \emph{bulk-edge correspondence}, which states, roughly speaking, that \quotes{the topology} of infinite systems agrees with the topology of the associated systems truncated to the half-space. This vague statement has physical content (about existence of edge modes) and two mathematical assertions: that the topological classifications of these two types of geometries agree, and moreover, that given a bulk insulator $H$, if we were to truncate it to the half-space (with largely any reasonable boundary conditions) to get $\widehat{H}$, calculating the index for $H$ or for $\widehat{H}$ (using different formulas) would yield the same number. This latter, numerical as it were, type of bulk-edge correspondence has been the subject of many papers, starting with the integer quantum Hall effect \cite{Hatsugai1993}, and continuing with the more mathematical \cite{Schulz-BaldesKellendonkRichter2000,Elbau_Graf_2002}; As far as we are aware, that the two topological classifications agree (without numerical equivalence) has been established for the entire table using KK-theory \cite{BourneKellendonkRennie2017} in the spectral gap regime.
	
	Let us make a few comments about the edge classification in the current setting. The one-dimensional edge Hilbert space is $\ell^2(\NN)\otimes\CC^N$. Now, the constraint of locality which was presented in \cref{def:1d locality} does not make sense anymore in the edge (unlike locality in the form of \cref{eq:exp locality} which would carry over directly). Moreover, generically edge systems are \emph{not} insulators: rather, they are truncations of infinite systems which are bulk insulators. In the spectral gap regime this may be encoded with or without recourse to a bulk Hamiltonian, as presented in \cite[Section 2.4]{BSS23}. In the one dimensional spectral gap setting, however, the situation is somewhat simplified for the following reason: by adding a truncation, we may only create finite-degeneracy eigenvalues but not change the essential spectrum (since the truncation is a compact perturbation of the bulk system). However, according to the RAGE theorem, eigenvalues are exponentially decaying from some center, and thus exponentially decaying from the truncation. As a result, it would appear that asking that the edge Hamiltonian is a Fredholm operator suffices for the bulk-gap requirement, because Fredholm operators are precisely those which are essentially gapped at zero. 
	
	But more is true: the Fredholm condition is a very weak notion of locality which in the edge setting is a good replacement for \cref{def:1d locality}. Indeed, if we think of the Fredholm condition as the finiteness of the kernel and the kernel of the adjoint, this is essentially asking that the operator cannot have too far away hopping, since if it did, that would violate the finite kernel condition.
	
	As such, it would appear that in one dimension, locality and the gap condition collapse into one insulator condition: 
	\begin{defn}[one-dimensional edge insulators]\label{def:1d edge insulators} $\widehat{H}=\widehat{H}^\ast\in\calB(\ell^2(\NN)\otimes\CC^N)$ is an edge insulator (with bulk gap at zero energy) iff it is a Fredholm operator.
	\end{defn}
	
	Clearly, in the edge picture we do not need to worry about the \quotes{bulk} insulator condition but we do need to make sure our systems are non-trivial in the sense that they have essential spectrum below and above zero. This corresponds to Atiyah and Singer's notion of the non-trivial component $\calF_\star$. We conclude that in the one-dimensional edge picture, if we are willing to accept a very weak notion of locality (but we emphasize it has not been completely ignored) the theory reduces to the classical Atiyah-Singer classification of Fredholm operators with symmetries \cite{AtiyahSinger1969}. Then, for example, class A corresponds to the non-trivial self-adjoint Fredholm operators $\calF^{\mathrm{sa}}_\star$ and from \cite{AtiyahSinger1969} we have \eq{\pi_0(\calF^{\mathrm{sa}}_\star)\cong \left[\Set{0}\to\calF^{\mathrm{sa}}_\star\right]\cong K_1(\Set{0})\cong\Set{0}} whereas in class AIII, the chiral off-diagonal sub-block $S$ must be Fredholm, which automatically implies that $\begin{bmatrix}
		0 & S^\ast\\S&0
	\end{bmatrix}$ is in $\calF^{\mathrm{sa}}_\star$. This then reduces to the even older Atiyah-J\"anich theorem: \eq{\pi_0(\calF) \cong \left[\Set{0}\to\calF\right]\cong K_0(\Set{0})\cong \ZZ\,.}
	
	One could then phrase an edge analog of \cref{thm:main 1d insulators classification theorem}; the formulas for the edge indices are obvious: they are the Fredholm indices or the $\ZZ_2$ Atiyah-Singer indices of the various Fredholm operators without taking $\sgn$ or polar part and without the application of $\bbLambda$, according to \cref{table:AZ chiral classes and symmetries of unitaries,table:AZ non-chiral classes and symmetries of projections}. We find:
	\begin{thm} One dimensional edge insulators as in \cref{def:1d edge insulators} have path components given by the $d=1$ column of \cref{table:Kitaev}, and hence the bulk and edge one dimensional systems have the same classifications. For any given bulk insulator $H$, the bulk index calculated from $H$ agrees with the edge index calculated from $\widehat{H}$ where $\widehat{H}$ is any edge insulator obtained by truncating $H$ 
		to the half-space such that $\widehat{H}$ is Fredholm and respects the symmetry constraint. Hence we obtain a numerical bulk-edge correspondence.
	\end{thm}
	\begin{proof}[Sketch of proof]
		As explained in the foregoing paragraphs, the classification result is covered by \cite{AtiyahSinger1969}. The numerical bulk-edge correspondence proof, at the spectral-gap level, is covered by the proof provided in \cite[Section 3]{Graf_Shapiro_2018_1D_Chiral_BEC}.
	\end{proof}

	\section{The mobility gap regime}\label{sec:mobility gap}
	As mentioned above, a more general mathematical criterion to guarantee zero electric conductance (and thus the insulator condition) is through quantum dynamics rather than via a spectral constraint. Drawing on Anderson localization, in \cite{EGS_2005} a deterministic condition was formulated for \emph{one} operator; we quote the equivalent condition given in \cite[Definition 2.5]{BSS23}:
	Let $B_1(\Delta)$ be the space of measurable functions $f:\RR\to\CC$ which are non-constant only within $\Delta$ and are bounded by $1$.
	\begin{defn}[mobility gap] A material $H=H^\ast\in\calL$ is mobility gapped at zero energy iff there exists some open interval $\Delta\ni0$ such that
		\begin{enumerate}
			\item There exists some $\mu>0$ such that for any $\ve>0$ there exists some $C_\ve<\infty$ such that \eql{\label{eq:mobility gap dynamical condition}\sup_{f\in B_1(\Delta)}\norm{f(H)_{xy}} \leq C\exp(-\mu\norm{x-y}+\ve\norm{x}) \qquad(x,y\in\ZZ^d)\,.} Hence $f(H)$ has exponentially decaying off-diagonal matrix elements whose rate of decay is however not uniform in the diagonal direction. Moreover, this statement is uniform in $f$.
			\item All eigenvalues of $H$ within $\Delta$ are uniformly finitely degenerate (the above condition implies $\sigma(H)\cap\Delta=\sigma_{\mathrm{pp}}(H)\cap\Delta$ via the RAGE theorem).
		\end{enumerate}    
	\end{defn}
	
	The type of decay condition appearing in \cref{eq:mobility gap dynamical condition} has been called \emph{weakly-local} in \cite{Shapiro2019,BSS23}. In one dimension it seems however that \cref{def:1d locality} is still weaker. 
	
	Furthermore, it is well-known from the theory of Anderson localization (see \cite{SimonWolff1986} e.g.) that any fixed deterministic energy value is almost-surely  not an eigenvalue of an Anderson localized random operator. Hence, in particular, even though in the mobility gap regime there is no reason to assume a spectral gap, or no accumulation of spectrum near zero, it wouldn't seem unreasonable to assume that zero is not an eigenvalue of $H$.
	
	Hence, if instead of taking the stronger \cref{eq:mobility gap dynamical condition} we merely setup the mobility gap condition as the minimal dynamical constraint to guarantee the existence of the index, we could come up with the following deterministic condition, which relies still on \cref{def:1d locality}: 
	\begin{defn}[tentative definition for mobility gap in $d=1$] A material $H=H^\ast\in\calL$ is mobility gapped at zero iff zero is not an eigenvalue of $H$ and if $\sgn(H)\in\calL$. We denote this space by $\calI_{\mathrm{mg,v1}}$. Its topology remains to be defined.
	\end{defn}
	
	This condition (up to strengthening the mode of locality) was the one given in \cite[Assumptions 1 and 2]{Graf_Shapiro_2018_1D_Chiral_BEC}. It is clear that such operators still have well-defined indices: the fact zero is not an eigenvalue of $H$ means that $\sgn(H)$ is actually unitary and not merely a partial isometry. But more is true: the entire proof of \cref{thm:main 1d insulators classification theorem} goes through if we skip the step connecting Hamiltonians with flat Hamiltonians! Indeed, all that is required is that operators be unitary or self-adjoint projections. 
	
	To connect Hamiltonians and flat Hamiltonians, we might employ the following ``abstract nonsense'' definition and argument:
	
	We shall make use of two different topologies on $\calI_{\mathrm{mg,v1}}$. First, let $\calT_{\iota:\calI_{\mathrm{mg,v1}}\to\calB}$ be the initial topology on $\calI_{\mathrm{mg,v1}}$ generated by the mapping $\iota:\calI_{\mathrm{mg,v1}}\to\calB$, inherited from the operator norm topology on $\calB$; this is by definition the subspace topology. Next, the functional calculus implies there is a map on operators $\sgn:\calI_{\mathrm{mg,v1}}\to\calB$ which maps $H\mapsto\sgn(H)$. Let $\calT_{\sgn:\calI_{\mathrm{mg,v1}}\to\calB}$ then be the initial topology on $\calI_{\mathrm{mg,v1}}$ which is generated by $\sgn:\calI_{\mathrm{mg,v1}}\to\calB$, where $\calB$ is understood with the operator norm topology.
	
	Importantly, any path $\calI_{\mathrm{mg,v1}}$ continuous w.r.t. $\calT_{\sgn:\calI_{\mathrm{mg,v1}}\to\calB}$ preserves an index, should one exist, as an index is always a function of the \emph{flat} Hamiltonian $\sgn(H)$: it is either $\findex_{2,\Lambda}\sgn(H)$ in the non-chiral case or, in the chiral case, it is $\findex_{(2),\Lambda}\left.P_{+} \sgn(H) P_{-}\right|_{\im P_-\to\im P_+}$ where $P_{\pm}$ are the SA projections onto the $\pm1$ eigenspaces of $\Pi$.
	
	\begin{lem}\label{lem:abstract nonsense strong deformation}
		The space $\calI^\flat_{\mathrm{mg,v1}}$ with the topology $\calT_{\iota:\calI^\flat_{\mathrm{mg,v1}}\to\calI_{\mathrm{mg,v1}}}$ is a strong deformation retract of the space $\calI_{\mathrm{mg,v1}}$ taken with the topology $\calT_{\sgn:\calI_{\mathrm{mg,v1}}\to\calB}$. 
	\end{lem}
	This statement remains true if we add the bulk-insulator constraint as well as any of the ten AZ symmetry constraints: $\calI_{\mathrm{mg,v1},B,\Sigma}^\flat$ is a strong deformation retraction of $\calI_{\mathrm{mg,v1},B,\Sigma}$ for any $\Sigma\in\mathrm{AZ}$. However, in order to avoid notational overhead we stick with the notation $\calI_{\mathrm{mg,v1}}$ where it is understood that we also take the bulk-insulator condition into the definition, as well as the appropriate symmetry.
	\begin{proof}
		Define $F:\calI_{\mathrm{mg,v1}}\times[0,1]\to\calI_{\mathrm{mg,v1}}$ via \eq{F(H,t) := (1-t)H+t\sgn(H)\qquad(H\in\calI_{\mathrm{mg,v1}},t\in[0,1])\,.} 
		We only show $F$ is continuous with respect to $\calT_{\sgn:\calI_{\mathrm{mg,v1}}\to\calB}$ (times the Euclidean topology on $[0,1]$), as the other two properties have already been shown in \cref{lem:flat Hamiltonians are strong deformation retraction}. The initial topology is generated by a sub-basis of inverse images of open sets on the co-domain. Hence, to show continuity, it suffices to start with such an open subset, and so let $U$ be open in $\calB$. Then we seek to show that $F^{-1}(\sgn^{-1}(U))$ is open. To that end, using the fact that for all $t\in[0,1]$, $\sgn(F(H,t))=\sgn(H)$, we have \eq{F^{-1}(\sgn^{-1}(U)) & =\Set{(H,t)|\sgn(F(H,t))\in U} \\ &= \Set{(H,t)|\sgn(H)\in U} \\ &= \sgn^{-1}(U)\times [0,1]\,.} The set on the last line is manifestly open.
	\end{proof}
	
	\begin{cor}[The one-dimensional column of the Kitaev table w.r.t. v1 of the mobility gap topology]\label{thm:main 1d insulators classification theorem for mgv1} At each fixed $N$, for any $\Sigma\in\mathrm{AZ}$, the path-connected components of $\calI_{\mathrm{mg,v1},N,B,\Sigma}$ considered with the initial topology $\calT_{\sgn:\calI_{\mathrm{mg,v1}}\to\calB}$, agree with the set appearing in the first column of \cref{table:Kitaev}.
	\end{cor}
	\begin{proof}
		Having the deformation retract, we know that $\pi_0(\calI^\flat_{\mathrm{mg,v1}})$ is  the same as $\pi_0(\calI_{\mathrm{mg,v1}})$ where we take the subspace topology for the former and $\calT_{\sgn:\calI_{\mathrm{mg,v1}}\to\calB}$ for the latter. That subspace topology $\calT_{\iota:\calI^\flat_{\mathrm{mg,v1}}\to\calI_{\mathrm{mg,v1}}}$ equals the topology $\calT_{\sgn:\calI^\flat_{\mathrm{mg,v1}}\to\calB}$ on $\calI^\flat_{\mathrm{mg,v1}}$. Indeed, this follows by the transitive property of the initial topology \cite[p. 2]{Grothendieck1973}: if $S\subseteq X$ and $f:X\to Y$ then the subspace topology on $S$ from $X$ taken with the initial topology generated by $f$ equals the initial topology on $S$ generated by $f\circ \iota$ where $\iota:S\to X$ is the inclusion map. But now, $\calT_{\sgn:\calI^\flat_{\mathrm{mg,v1}}\to\calB}$ and $\calT_{\iota:\calI^\flat_{\mathrm{mg,v1}}\to\calB}$ coincide. This is a consequence of the fact that $\sgn:\calI_{\mathrm{mg,v1}}^\flat\to\calB$ reduces to the inclusion map $\iota:\calI^\flat_{\mathrm{mg,v1}}\to\calB$ (since $\sgn\circ\sgn=\sgn$), and the subspace topology is precisely generated by the inclusion map.
		
		The end conclusion is that we may calculate $\pi_0(\calI^\flat_{\mathrm{mg,v1}})$ w.r.t. the topology $\calT_{\iota:\calI^\flat_{\mathrm{mg,v1}}\to\calB}$, i.e. $\pi_0(\calI^\flat)$ with the subspace topology from norm topology. This, however, is precisely the calculation already done in the proof of \cref{thm:main 1d insulators classification theorem}.
	\end{proof}

	With this, it would appear that the mobility gap problem is solved in one dimension. We maintain this is, however, not the case. Indeed, a subtlety appears from the fact we allowed ourselves to shift Hamiltonians to always place the Fermi energy at zero, which has thus made the above analysis single out zero energy. This is of course invalid because if we were to ask that \emph{all} given fixed energies are almost-surely not an eigenvalue we would constrain our operators to have a spectral gap, which we are precisely trying to avoid. So by always placing $\mu=0$ we are \emph{not} allowed to ask that zero is not an eigenvalue, and so, following the theory of Anderson localization, we would make another attempt as
	
	\begin{defn}[another tentative definition for mobility gap in $d=1$] A material $H=H^\ast\in\calL$ is mobility gapped at zero iff $\ker H$ is finite dimensional and  $\sgn(H)\in\calL$. We denote the space of all such operators as $\calI_{\mathrm{mg,v2}}$; its topology remains to be defined.
	\end{defn}
	
	We note that in this case, $\sgn(H)$ is merely a partial isometry with finite kernel (and so it is Fredholm, $\sgn(H)$ having closed range) and that still $\bbLambda \sgn(H)$ is Fredholm.
	
	Unfortunately, with the initial topology $\calT_{\sgn:\calI_{\mathrm{mg,v2}}\to\calB}$ from norm topology on $\calB$, this definition is \emph{still} not good enough, as the following counterexample demonstrates a deviation from the Kitaev table. As a result, to rectify this situation, one should either define another topology on $\calI_{\mathrm{mg,v2}}$ (e.g. the subspace topology) or start with another space entirely. The first step should be to require an open interval around zero to have finite degeneracy (uniformly). Another possibility would be to place a dynamical constraint. We postpone such investigations to future work.
	
	\begin{example}
		There exist two operators $R,S\in\calI_{\mathrm{mg,v2},\mathrm{AIII}}$ which have the same $\Lambda$-index and yet there is no continuous path connecting them in the topology $\calT_{\sgn:\calI_{\mathrm{mg,v2}}\to\calB}$.
	\end{example}
	\begin{proof}
		Define $R$ to be the unitary right shift operator on $\ell^2(\ZZ)$ and $\Lambda$ as above projects to $\szpan(\Set{\delta_n}_{n\geq1})$. With this, we choose $S := \Lambda^\perp R \Lambda^\perp + \Lambda R \Lambda$. As such, it is clear that $\findex_\Lambda S = \findex_\Lambda R = -1$ since $\bbLambda R = \bbLambda S$. Moreover, it is clear that both have finite dimensional kernels and co-kernels, and $\polar(R)=R$ and $\polar(S)=S$ which are both $\Lambda$-local. This means that both $R$ and $S$ are in $\calI_{\mathrm{mg,v2},\mathrm{AIII}}$.
		
		However, we maintain that no path continuous in $\calT_{\sgn:\calI_{\mathrm{mg,v2}}\to\calB}$ can exist in between them. Indeed, one may follow \cref{lem:abstract nonsense strong deformation} to show that the space $\calI^\flat_{\mathrm{mg,v2}}$ with the topology $\calT_{\iota:\calI^\flat_{\mathrm{mg,v2}}\to\calI_{\mathrm{mg,v2}}}$ is a strong deformation retract of the space $\calI_{\mathrm{mg,v2}}$ taken with the topology $\calT_{\sgn:\calI_{\mathrm{mg,v2}}\to\calB}$. Indeed, all that is used there is the fact that $\sgn\circ\sgn=\sgn$ and not the unitarity of the operators. However, now we could argue that, just as we did in the proof of \cref{thm:main 1d insulators classification theorem for mgv1}, that for the space $\calI^\flat_{\mathrm{mg,v2}}$, the topologies $\calT_{\iota:\calI^\flat_{\mathrm{mg,v2}}\to\calI_{\mathrm{mg,v2}}}$, $\calT_{\sgn:\calI^\flat_{\mathrm{mg,v2}}\to\calB}$ and $\calT_{\iota:\calI^\flat_{\mathrm{mg,v2}}\to\calB}$ coincide. As such, if a path from $R$ to $S$ continuous in $\calT_{\sgn:\calI_{\mathrm{mg,v2}}\to\calB}$ were to exist, flattening that path, would imply the existence of a path from $\polar(R)=R$ to $\polar(S)=S$ within $\calI^\flat_{\mathrm{mg,v2}}$ continuous w.r.t. the topology $\calT_{\iota:\calI^\flat_{\mathrm{mg,v2}}\to\calB}$. Indeed, this uses the fact that $\sgn:\calI_{\mathrm{mg,v2}}\to\calI_{\mathrm{mg,v2}}^\flat$ is continuous with respect to the initial topology. This last path, continuous w.r.t. $\calT_{\iota:\calI^\flat_{\mathrm{mg,v2}}\to\calB}$, cannot exist.
		
		Indeed, the topology $\calT_{\iota:\calI^\flat_{\mathrm{mg,v2}}\to\calB}$ is just the subspace topology from operator norm topology, on the space of partial isometries with finite kernel and co-kernel, i.e., Fredholm partial isometries. For us, $R$ has no kernel and no co-kernel, with Fredholm index zero, and $S$ has a kernel and co-kernel of dimension $+1$, again with Fredholm index zero. However, since we are working in the space of partial isometries, it is well-known \cite[Problem 130]{halmos1982hilbert} that the dimensions of the kernel and of the co-kernel are continuous in the operator norm (rather than merely lower semi-continuous as in the case of Fredholm operators). As such, having a continuous path between these two operators, would violate the continuity of the dimension of the kernels for partial isometric Fredholms.
	\end{proof}
	
	\section{Classification of bulk spectrally-gapped insulators in odd $d>1$}\label{sec:higher odd d}
	
	Our analysis so far has focused on one-dimensional structures. Let us now turn our attention to higher dimensions. We seek an analogous notion of locality as presented in \cref{def:1d locality} which would apply in higher dimensions. In their textbook, Prodan and Schulz-Baldes \cite[Chapter 6]{PSB_2016} present a construction which they ascribe to \cite{Connes1994,Garcia2000} of locality in all dimensions which proceeds as follows.
	
	Let us define $k:=d/2$ in even dimensions and $k:=(d-1)/2$ in odd dimensions. In the spirit of \cref{rem:the role of N with the compact commutator condition}, let us (without loss) \emph{assume} that $N$ is divisible by $2^k$, so that it actually carries a representation of a Clifford algebra with generators $\gamma_1,\dots,\gamma_d$ (now considered as $N\times N$ matrices). The Dirac operator is then defined as \eql{ D := \sum_{i=1}^d X_i \otimes \gamma_i} and now, in higher dimensions, we choose the locality projection $\Lambda$ to be \eql{\label{eq:Dirac projection}\Lambda := \frac{1}{2}\left(\Id+\sgn(D)\right)\,.} This operator no longer acts trivially in the internal space $\CC^N$ factor as was the case in $d=1$. Hence $D$ and both $\Lambda$ intertwine space and the internal degrees of freedom in a non-trivial way; we note that if $d=1$ we get back our choice made in \cref{def:1d locality}. It should be remarked that in our notation $\sgn(D)$ is a partial isometry which may be extended to a unitary in an obvious way.

	\begin{defn}[locality in higher odd dimensions] We define an operator  $A\in\calB(\ell^2(\ZZ^d)\otimes\CC^N)$ (with $N$ divisible by $2^k$ without loss as above and $d$ odd) to be local iff it is $\Lambda$-local as in \cref{def:Lambda-local operators} with the particular choice of $\Lambda$ made in \cref{eq:Dirac projection}.
	\end{defn}
	Going back to \cref{def:2d locality}, we identify in $d=2$\eq{\sgn(D)=\begin{bmatrix}
			0 & \exp\left(-\ii\Phi\right) \\ \exp\left(\ii\Phi\right) & 0
	\end{bmatrix}} with $\Phi\equiv\arg(X_1+\ii X_2)$ the angle-position operator. This $d=2$ pattern is typical: in even dimensions $\sgn(D)$ breaks into off-diagonal form as above \cite[Chapter 6]{PSB_2016}. It is clear that ignoring this internal structure of $\sgn(D)$ we get trivial classification for projections in even dimensions in contradiction to expectations. Hence it is clear that in even dimensions one has to contend with a different notion of locality, one which entails operators which essentially commute with a fixed \emph{unitary} (the Dirac phase) rather than the Dirac projection. This leads to rather different classification scheme which we have little to say about. 
	
	On the other hand, in higher \emph{odd} dimensions we may proceed by adopting the definition of a bulk insulator as in \cref{def:bulk insulators}, i.e., bulk insulators are operators $H=H^\ast\in\calL$ which have a spectral gap about zero and for which the Fermi projection $P$ is not merely local but also $\Lambda$-non-trivial. \emph{We emphasize that now, however, it can no longer be reasonably argued that this $\Lambda$-non-trivial requirement would correspond to bulk systems, since now $\im(\Lambda)$ cannot be identified geometrically with an edge system.} Be that as it may, one may carry on and in fact obtain all odd-dimensional columns of the Kitaev table in this way, in precisely the same manner as we did in \cref{sec:bulk insulators classification}.
	
	We thus phrase, without proof, the following
	\begin{thm}
		At each fixed $N$, for $d\in2\NN+1$, for any $\Sigma\in\mathrm{AZ}$, the path-connected components of $\calI_{N,B,\Sigma}$ considered with the subspace topology associated with the operator norm topology, agree with the corresponding set appearing in the odd-dimensional columns of \cref{table:Kitaev}. 
		
		That is, now bulk insulators are defined as in the foregoing paragraph, using the particular choice of compact-commutator locality and bulk-insulator with the choice of $\Lambda$ as in \cref{eq:Dirac projection}.
	\end{thm}
	To prove this theorem, given \cref{sec:bulk insulators classification} and the above paragraph, the missing part is explaining how the dimensions cause the symmetry classes to shift which is a shadow of the K-theoretic identity $K_i(S^2 \calA) = K_{i+2}(\calA)$ where $S$ is the suspension of a C-star algebra. The shift does not mix the chiral and non-chiral classes, and furthermore, classes A and AIII are fixed by the shift. So for either the chiral or non-chiral classes, there is a four orbit shuffle that happens as $d\mapsto d+2$. To explain this shift one has to allow $\Theta$ and $\Xi$ to act non-trivially on the Clifford space.
	
	We avoid doing so here because ultimately, we feel that the notion of locality and bulk-insulator derived from this choice of $\Lambda$ is physically contrived. Yes, one could take the point of view that locality and the gap condition may be any sufficiently strong criterion so that the indices are well-defined. But whereas in one-dimension this still made sense with respect to physical real space, in higher-dimensions, we simply cannot find a way to justify this particular choice of $\Lambda$ locality and $\Lambda$ non-triviality. We thus postpone the higher dimensional problem to future work.

	\bigskip
	\bigskip
	\noindent\textbf{Acknowledgments.} 
	We are  indebted to Gian Michele Graf and Michael Aizenman for stimulating discussions.
	\bigskip
	\appendix

	\section{The Atiyah-Singer $\ZZ_2$ index theory}\label{sec:Atiyah-Singer Z2 index theory}
	
	The material in this section was first presented by Atiyah and Singer in \cite{AtiyahSinger1969}. Different proofs appeared in \cite{Schulz-Baldes2015,FSSWY22} but for the sake of completeness we include a short presentation of the theory here, also since the context is somewhat more abstract than the $\Theta$-odd analysis which was presented in the appendix \cite{FSSWY22}.
	
	\begin{lem}[An explicit Diudonn\'e]\label{lem:strengthening of Diudonne}
		Let $T\in\mathcal{F}\left(\mathcal{H}\right)$. Then if $S\in B_{\norm{G}^{-1}}\left(T\right)$
		where $G$ is any parametrix of $T$ then $S\in\mathcal{F}\left(\mathcal{H}\right)$
		too, and 
		\eq{
			\dim\left(\ker\left(S\right)\right)=\dim\left(\ker\left(T\right)\right)-\dim\left(\im\left(\Schur\right)\right)
		}
		where $\Schur:\ker\left(T\right)\to\im\left(T\right)^{\perp}$
		is the Schur-complement of $S$ in the $T$-decomposition, i.e., 
		\eq{
			\Schur:=S_{CA}-S_{CB}\left(S_{DB}\right)^{-1}S_{DA}
		}
		with $A:=\ker\left(T\right)$, $B:=\ker\left(T\right)^{\perp}$, $C:=\im\left(T\right)^{\perp}$,
		$D:=\im\left(T\right)$.
	\end{lem}
	
	\begin{proof}
		Decomposing $\mathcal{H}= A\oplus B= C\oplus D$ we find $T:A\oplus B\to C\oplus D$
		is written in block-operator form as 
		\eq{
			T=\begin{bmatrix}0 & 0\\
				0 & T_{DB}
			\end{bmatrix}
		}
		with $T_{DB}:B\to D$ a vector space isomorphism, and we may also
		decompose $S$ as $S:A\oplus B\to C\oplus D$ in block operator form
		to get 
		\eq{
			S = \begin{bmatrix}S_{CA} & S_{CB}\\
				S_{DA} & S_{DB}
			\end{bmatrix}\,.
		}
		Now if $\norm{S-T}$ is sufficiently small then $\norm{S_{DB}-T_{DB}}$
		is sufficiently small so that $S_{DB}$ is also invertible (this may
		be verified to be true with the upper bound $\norm{G}^{-1}$), which guarantees that $\Schur$ exists. 
		
		Using an LDU decomposition we may write 
		\eq{
			S=J_{1}\left(\Schur\oplus S_{DB}\right)J_{2}
		}
		where $J_{1},J_{2}$ are two invertible operators, and as such 
		\eq{
			\dim\left(\ker\left(S\right)\right)  =  \dim\left(\ker\left(\Schur\right)\right)+\underbrace{\dim\left(\ker\left(S_{DB}\right)\right)}_{=0}\,.
		}
		Now apply rank-nullity on $\Schur:A\to C$ to get 
		\eq{
			\dim\left(\ker\left(\Schur\right)\right)+\dim\left(\im\left(\Schur\right)\right)  =  \dim\left(A\right)=\dim\left(\ker\left(T\right)\right)
		}
		which yields the result.
	\end{proof}

	\begin{thm}[Atiyah-Singer $\ZZ_2$ index]\label{thm:ASZ2 index}
		If $F\in\calF_{\star\HH}\left(\mathcal{H}\right)$ (in the sense of \cref{sec:locality and symmetry}), i.e., $J:\calH\to\calH$ is an anti-unitary that squares to $-\Id$ and we have
		\eq{
			FJ=J F^{\ast}
		}
		then 
		\eq{
			\findex_{2}\left(F\right)\equiv\dim\left(\ker\left(F\right)\right)\mod2\in\mathbb{Z}_{2}
		}
		is well-defined, in the sense that if $G\in\calF_{\star\HH}\left(\mathcal{H}\right)$
		and $\norm{F-G}$ is sufficiently small, then
		\eq{
			\findex_{2}\left(G\right) = \findex_{2}\left(F\right)\,.
		}
	\end{thm}
	
	\begin{proof}
		Using the same definitions as in the proof above, we have 
		\eq{
			\dim\left(\ker\left(G\right)\right)  =  \dim\left(\ker\left(F\right)\right)-\dim\left(\im\left(\Schur\right)\right)
		}
		with $\Schur:\ker\left(F\right)\to\im\left(F\right)^{\perp}$
		the Schur complement, given by 
		\eq{
			\Schur  =  G_{CA}-G_{CB}\left(G_{DB}\right)^{-1}G_{DA}\,.
		}
		Since $F$ is $\star$-quaternionic with respect to $J$, then \eq{J \ker F= (\im F)^\perp,\quad J(\ker F)^\perp = \im F\,.} 
		In particular, the expressions \eq{G_{CA}J=JG_{CA}^*,\quad G_{DB}J=JG_{DB}^*,\quad G_{CB}J=JG_{DA}^*,\quad G_{DA}J=JG_{CB}^*} make sense and follow directly from $G$ being $\star$-quaternionic. It follows that \eq{\Schur J=J \Schur^*}
		
		We now argue that $\im\Schur$ is even-dimensional. Let us view $\Schur:(\ker \Schur)^\perp \to \im \Schur$ as an invertible operator that is $\star$-quaternionic with respect to $J$. Since $\Schur\Schur^*:\im\Schur\to\im\Schur$ is self-adjoint, the space $\im\Schur$ decomposes into eigen-subspaces from $\Schur\Schur^*$. Let $E$ be one of the eigen-subspace and take $\vf\in E$ and write $\Schur\Schur^*\vf=\lambda\vf$. Clearly $\widetilde{\vf}\neq 0$ since $\Schur$ and $J$ are both linear invertible. Let $\widetilde{\vf}:=\Schur J\vf\in\im \Schur$. We have \eq{\ip{\vf}{\Schur J\vf}=\ip{\Schur ^*\vf}{J\vf}=\ip{J^2\vf}{J\Schur^*\vf}=-\ip{\vf}{\Schur J\vf}\,.} 
		Thus $\ip{\vf}{\widetilde{\vf}}=0$. Also \eql{\Schur\Schur^*\widetilde{\vf}=\Schur\Schur^*\Schur J\vf=\Schur J\Schur\Schur^*\vf=\Schur J\lambda\vf=\lambda \widetilde{\vf}\,.} 
		Thus $\widetilde{\vf}\in E$, and moreover, we have \eq{\Schur J\widetilde{\vf}=\Schur J\Schur J=J^2\Schur\Schur^*\vf=-\lambda\vf\,.} 
		Thus the span of $\Set{\vf,\widetilde{\vf}}$ is invariant under the action of $\Schur J$. 
		
		Pick $\psi$ in the orthogonal of the span of $\Set{\vf,\widetilde{\vf}}$ in $E$. We can form $\widetilde{\psi}:=\Schur J \psi$ similar as before, where we have $\ip{\psi}{\widetilde{\psi}}=0$ and $\widetilde{\psi}\in E$. In particular $\ip{\eta}{\widetilde{\psi}}=0$ for $\eta$ in the span of $\Set{\vf,\widetilde{\vf}}$ since \eq{\ip{\eta}{\widetilde{\psi}}=\ip{\eta}{\Schur J\psi}=\ip{\Schur^*\eta}{J\psi}=\ip{J^2\psi}{J\Schur^*\eta}=-\ip{\psi}{\Schur J\eta}=0\,.} 
		Thus the eigen-subspace $E$ is even-dimensional. This implies that $\im \Schur$ is even-dimensional.

	\end{proof}

	We may also recast the above theorem somewhat differently as follows
	\begin{thm}\label{thm:ASZ2 index v2}
		If $F\in\Schur\calF_{\ii\RR}\left(\mathcal{H}\right)$ as in \cref{sec:locality and symmetry}, i.e., $F$ is a self-adjoint Fredholm with $C$ is a real structure on $\calH$, such that \eq{\Set{F,C}=FC+CF=0} then \eq{\findex_{2}\left(F\right)\equiv\dim\left(\ker\left(F\right)\right)\mod2\in\mathbb{Z}_{2}} is well-defined, in the sense that if $G\in\Schur\calF_{\ii\RR}\left(\mathcal{H}\right)$
		and $\norm{F-G}$ is sufficiently
		small, then 
		\eq{\findex_2(G)=\findex_2(F)\,.}
	\end{thm}
	
	\begin{proof}
		
		Since $F$ is self-adjoint, then $\ker F=(\im F)^\perp=:A$ and $(\ker F)^\perp=\im F=:B$. Decompose $G$ in $A\oplus B$, we write \eq{G=\begin{bmatrix}G_{AA} & G_{AB} \\ G_{AB}^* & G_{BB}\end{bmatrix}} Since $F:(\ker F)^\perp\to\im F$ is invertible and $\norm{G-F}$ small, then $G_{BB}$ is invertible. Define the Schur complement $\Schur:A\to A$ as \eq{\Schur = G_{AA} - G_{AB} G_{BB}^{-1} G_{AB}^*} Since $G$ is self-adjoint, then $\Schur$ is, too. Since $GC=-CG$, then the subspaces $A,B$ are both invariant under the action of $C$. Thus the expressions \eq{G_{AA}C=-CG_{AA},\quad G_{BB}C=-C_{BB},\quad G_{AB}C=-CG_{AB}} make sense and hold. It follows that \eq{\Schur C=-C\Schur\,.} 
		Similar to \cref{lem:strengthening of Diudonne}, one also has \eq{\dim(\ker G)= \dim(\ker F)-\dim(\im\Schur)\,.}
		
		We argue that $\im \Schur$ is even-dimensional. Since $\Schur$ is self-adjoint, $(\ker \Schur)^\perp=\im \Schur=:V$ which is finite-dimensional. View $\Schur:V\to V$ as invertible operator. Since $\Schur^2$ is self-adjoint, the space $V$ admits an eigen-subspace decomposition with respect to $\Schur^2$. Let $E$ be one of the eigen-subspace and pick $\vf\in E$ where $\Schur^2\vf=\lambda \vf$. Let $\widetilde{\vf}:=\Schur C\vf$. Note $C\vf\in V$ since $\Schur$ is $C$-real. Thus $\widetilde{\vf}$ is well-defined. Now \eq{\ip{\vf}{\widetilde{\vf}}=\ip{\vf}{\Schur C\vf}=\ip{\Schur \vf}{C\vf}=\ip{\vf}{C\Schur\vf}=-\ip{\vf}{\Schur C\vf}=-\ip{\vf}{\widetilde{\vf}}} and hence $\ip{\vf}{\widetilde{\vf}}=0$. Also \eq{\Schur^2 \widetilde{\vf}=\Schur^2(\Schur C\vf)=\Schur C\Schur^2\vf=\lambda \Schur C\vf=\lambda\widetilde{\vf}\,.} 
		Thus $\widetilde{\vf}\in E$. Moreover, we have \eq{\Schur C\widetilde{\vf}=\Schur C\Schur C\vf=C^2\Schur^2\vf=\lambda\vf} and hence the span of $\Set{\vf,\widetilde{\vf}}$ is invariant under the action of $\Schur C$.
		
		Pick $\psi$ in the orthogonal complement of the span of $\Set{\vf,\widetilde{\vf}}$ in $E$. 
		Similarly construct $\widetilde{\psi}=\Schur C\psi$ such that $\ip{\psi}{\widetilde{\psi}}=0$ and $\widetilde{\psi}\in E$. In particular $\ip{\eta}{\widetilde{\psi}}=0$ for $\eta$ in the span of $\Set{\vf,\widetilde{\vf}}$ since \eq{\ip{\eta}{\widetilde{\psi}}=\ip{\eta}{\Schur C\psi}=\ip{\Schur\eta}{C\psi}=\ip{C^2\psi}{C\Schur\eta}=\ip{\psi}{-\Schur C\eta}=0\,.} 
		Thus the eigen-subspace $E$ is even-dimensional. This implies that $\im \Schur$ is even-dimensional.
	\end{proof}
	
	\begin{cor}\label{cor:compact perturbation preserve Z2 index} The Atiyah-Singer $\ZZ_2$ index is stable under symmetric compact perturbations:
		\begin{enumerate}
			\item If $F\in \calF_{\star\HH}$ and $K\in\calK_{\star\HH}$, then $\findex_2 (F+K)= \findex_2 F $.
			\item If $F\in\calF_{\ii\RR}^{\mathrm{sa}}$ and $K\in\calK_{\ii\RR}^{\mathrm{sa}}$, then $\findex_2(F+K)=\findex_2 F$.
		\end{enumerate}
	\end{cor}
	\begin{proof}
		Consider a straight-line homotopy from $F$ to $F+K$ and use \cref{thm:ASZ2 index} and \cref{thm:ASZ2 index v2}.
	\end{proof}
	
	\section{A child's garden of homotopies}
	In this section we employ the same notational conventions as in \cref{sec:locality and symmetry}. We are concerned with homotopies of unitaries and self-adjoint projections \emph{without} locality constraints. 
	\subsection{Equivariant homotopies of unitaries}
	The following theorem was presented in \cite{Kuiper1965}. The proof which was outlined in \cref{eq:Kuiper's path} applies only to the case $\FF=\CC$, the other two cases may be found in Kuiper's original paper.
	\begin{thm}[Kuiper]\label{thm:Kuiper} For any $\FF\in\Set{\CC,\RR,\HH}$ and any invertible operator $A\in\calG_\FF$, there is a continuous path from $A$ to $\Id$ which passes within $\calG_\FF$. If $A$ is unitary the path passes through unitaries.
	\end{thm}
	
	New (to us) is the following $\star$-variant of it:
	\begin{thm}\label{star-RH Kuiper}
		Let $\FF\in\Set{\star\RR,\star\HH}$. Then \eq{\pi_0(\calU_\FF)\simeq \Set{0}}
	\end{thm}
	\begin{proof}
		Let $F\in \Set{C,J}$. For bi-variate polynomials $p(z,\bar{z})=\alpha z^n \bar{z}^m$, we have \eq{p(U,U^\ast)F=\alpha U^n (U^\ast)^m F=\bar{\alpha} (U^\ast)^n U^m = F(p(U,U^\ast))^\ast\,.} Thus for a continuous function $f\in C(\sigma(U))$, one has $f(U)\in \calB_\FF$. 
		
		Consider the square root function \eq{h(z)=\exp(\ii\arg(z)/2)\qquad(z\in\CC\setminus\Set{0})} where we take $\arg(z)\in[0,2\pi)$ for concreteness. The function $h$ is bounded measurable on $\sigma(U)$, and clearly there exists a sequence of continuous functions $f_n$ on $\sigma(U)$ that converges point-wise to $h$, and $\|f_n\|_\infty$ is bounded. By the spectral theorem \cite[Theorem VII.2(d)]{Reed_Simon_FA_0125850506}, $f_n(U)$ converges strongly to $h(U)$. Thus $h(U)\in \calB_\FF$. In particular, since $\bar{h}(z)h(z)=1$, then $h(U)\in \calU_\FF$. Write $h(U)=Fh(U)^* F^*$, then \eq{U=h(U)^2=h(U)Fh(U)^* F^*\,.} Use \cref{thm:Kuiper} to construct a continuous path of unitaries $[0,1]\ni t\mapsto V_t$ connecting $h(U)$ to $\Id$, we let $U_t=V_tFV_t^*F^*$. Then \eq{U_t F=V_tFV_t^*=F(FV_t F^* V_t^*)=FU_t^*} Thus $U_t\in\calU_\FF$ and connects $U$ and $\Id$.
	\end{proof}
	
	We will make use of the fact that the polar part preserves symmetry constraints:
	\begin{lem}\label{lem:polar part preserve symmetry}
		Let $\FF\in\Set{\RR,\HH,\star\RR,\star\HH}$. If $A\in\calB_\FF$, then $\polar(A)\in\calB_\FF$.
	\end{lem}
	\begin{proof}
		Let $F\in \Set{C,J}$ according to $\FF$.
		
		First consider the cases $\FF\in\Set{\RR,\HH}$. Since $\calB_\FF$ is closed under the adjoint operation, $|A|^2\in\calB_\FF$ too, and hence the polar part, by writing it as the strong limit of functions which approximate $ A |A|^{-1}$. Indeed, $\polar(A)=\slim Af_n(|A|)$ where $f_n(\lambda)=1/\lambda$ if $\lambda\geq 1/n$ and $f_n(\lambda)=1/n$ if $\lambda\leq 1/n$. In particular, $f_n:\RR\to\RR$ is $\RR$-valued and hence $f_n(|A|)F=F f_n(|A|)$. Then $\polar(A)F=\slim Af_n(|A|)F=F \slim Af_n(|A|)=F \polar(A)$.
		
		Next, for the case $\FF\in\Set{\star\RR,\star\HH}$, we have \eq{A^*AF=A^*F A^*=F AA^*.} Then $f_n(|A|)F=F f_n(|A^*|)$ and \eq{\polar(A)F=\slim Af_n(|A|)F=F \slim A^*f_n(|A^*|)=F \polar(A^*)=F(\polar(A))^*.}
	\end{proof}

	\begin{lem}\label{lem:extension of zero index Fredholm RCH partial isometries to unitaries} For $\FF=\CC,\RR,\HH$, let $Z\in\calF_\FF$ be essentially unitary with zero Fredholm index. Then there is a unitary operator $Y\in\calU_\FF$ such that $Y-Z\in\calK$.
	\end{lem}
	\begin{proof}
		In what follows, we let $F=\Id,C,J$ according to the value of $\FF$. Using the fact that $T=T^\ast\in\calB$ has $T\in\calK$ iff $\sigma_{\mathrm{ess}}(T)=\Set{0}$, since $\Id-|Z|^2$ is compact, its essential spectrum equals $\Set{1}$ which implies $\Id-|Z|$ is compact as well. Let $\polar(Z)$ denote the polar part of $Z$. Then by the above, 
		\eql{ \label{eq:Z - pol(Z) is compact}
			Z-\polar(Z)=\polar(Z)|Z|-\polar(Z)=\polar(Z)(|Z|-\Id)\in\calK\,.
		} 
		Now, since $\findex Z=0$, $\ker Z$ and $(\im Z)^\perp$ are finite-dimensional and of the same dimension, we let $M:\ker Z \to \im(Z)^\perp$ be any unitary map between two finite vector spaces of the same dimension and define $Y:=\polar(Z)\oplus M$ which is now unitary and $Y-Z$ is compact using \cref{eq:Z - pol(Z) is compact} and the fact $M$ is finite rank. This settles the case $\FF=\CC$.
		
		Next, if $F\neq\Id$, we have $\polar(Z)F=F\polar(Z)$ too using \cref{lem:polar part preserve symmetry} and $ZF=FZ$ implies \eq{F\ker Z=\ker Z,\quad F(\im Z)^\perp=(\im Z)^\perp\,.} 
		
		Now the analysis divides according to the value of $\FF$. When $\FF=\RR$, we have from \cref{lem:basis in RCH symmetry} right below bases $\Set{\varphi_i}_{i=1}^m$ and $\Set{\psi_i}_{i=1}^m$ for $\ker Z$ and $(\im Z)^\perp$, respectively, such that $\varphi_i,\psi_i$ are fixed by $C$. Let $M:\ker Z\to (\im Z)^\perp$ be the unitary operator mapping $\varphi_i\mapsto \psi_i$ for  $i=1,\dots,m$. Then \eq{CM\varphi_i=C\psi_i=\psi_i=M\varphi_i=MC\varphi_i\,.} Thus the unitary direct sum $Y:=\polar(Z)\oplus M$ commutes with $C$. 
		
		When $\FF=\HH$, applying \cref{lem:basis in RCH symmetry} again, we obtain bases of Kramers pairs $\Set{\varphi_i,\varphi_{i+m}}_{i=1}^m$ and $\Set{\psi_i,\psi_{i+m}}_{i=1}^m$ for $\ker Z$ and $(\im Z)^\perp$, respectively where $m$ is half the dimension of the kernel. Let $M:\ker Z\to(\im Z)^\perp$ be defined as \eq{\varphi_i\mapsto -\psi_{i+m},\quad \varphi_{i+m}\mapsto \psi_i\,.} 
		Then \eq{JM\varphi_i=-J\psi_{i+m}=\psi_{i}=M\varphi_{i+m}=MJ\varphi_i} and similarly for $JM\varphi_{i+m}=MJ\varphi_{i+m}$. 
		
		We conclude that in all three cases, one extends $\polar(Z)$ to a unitary operator $Y:=\polar(Z)\oplus M$ with $[F,Y]=0$ so $Y\in\calU_\FF$ and moreover, $Y-Z\in \calK$.
	\end{proof}
	Above we have used the following equivariant basis assertion:        
	
	\begin{lem}\label{lem:basis in RCH symmetry}
		Let $V$ be a Hilbert space with dimension $m$, possibly infinite.
		\begin{enumerate}
			\item Suppose there is an anti-unitary $C:V\to V$ with $C^2=\Id$. Then $V$ has an ONB $\Set{\vf_i}_{i=1}^m$ such that $C\vf_i=\vf_i$. If $m=2l$ for some $l\in\NN\cup\Set{\infty}$, then $V$ has an ONB $\Set{\vf_i,\psi_i}_{i=1}^l$ such that $C\vf_i=\psi_i$.
			\item If there is an anti-unitary $J:V\to V$ with $J^2=-\Id$, $V$ has an orthonormal basis consisting of \emph{Kramers pairs}, i.e., there is an ONB $\Set{\vf_i,\psi_i}_{i=1}^m$ such that $J\vf_i=\psi_{i}$. In particular, if $\dim V<\infty$, it follows that $\dim V=2m.$
		\end{enumerate}
	\end{lem}
	\begin{proof}
		The first part is \cite[Lemma 1]{garcia2006complex} which we reproduce here for completeness. Consider the \emph{subset} $W=\Set{\varphi+C\varphi| \varphi\in V}$ consisting of elements from $V$. The elements in $W$ are fixed by $C$, and $W$ is an $\RR$-vector space. To verify this, let $\alpha\in\RR$ and $\psi=\varphi+C\varphi\in W$, then $\alpha\psi=\alpha\varphi+C(\alpha\varphi)\in W$. Let $\Set{\varphi_i}_{i=1}^{q}$ be an orthonormal basis for $W$ considered as an $\RR$-vector space. Here $q$ can be finite or infinite. Now, for any $\psi\in V$, let $\eta :=-\ii\psi$, and we have the identity 
		\eq{
			\psi=\frac{1}{2}((\psi+C\psi)+\ii(\eta + C\eta))\,.
		} 
		By expanding now $\psi+C \psi$ and $\eta+C\eta$ in the $\RR$-basis of $W$, we conclude that any element $\psi\in V$ may be written as a $\CC$-linear combination of $\Set{\varphi_i}_{i=1}^q$. In fact the set $\Set{\varphi_i}_{i=1}^q$ is orthonormal within $V$, since $W$ inherits the same inner product structure with which $\Set{\varphi_i}_{i=1}^q$ is orthonormal. Thus $\Set{\varphi_i}_{i=1}^q$ is an ONB for $V$ and $q=m$.
		
		If $m=2l$ for some $l\in\NN\cup\Set{\infty}$, define \eq{\eta_{j}^\pm:=\frac{\vf_{2j-1}\pm\ii \vf_{2j}}{\norm{\vf_{2j-1}\pm\ii \vf_{2j}}}\qquad(j\in{1,\dots,l})\,.} 
		Thus $\Set{\eta^\pm_j}_{j=1}^l$ is an ONB for $V$ with $C\eta_j^\pm=\eta_j^\mp$.
		
		For the second part, let $\varphi_1$ denote a unit-length vector from $V$. Let $\psi_1:=J\varphi_1$. By anti-unitarity, $\norm{\psi_{1}}=1$ and by $J^2=-\Id$,\eq{\ip{\varphi_1}{\psi_1}=\ip{\varphi_1}{J\vf_1}=\ip{J^2\vf_1}{J\vf_1}=-\ip{\vf_1}{\psi_1}} which implies $\ip{\varphi_1}{\psi_1}=0$. In particular, one has $J\psi_1=-\vf_1$ and the span of $\Set{\vf_1,\psi_1}$ is invariant under $J$. Pick another $\varphi_2$ in the orthogonal complement of the span of $\Set{\vf_1,\psi_1}$ and let $\psi_2=J\vf_2$. One readily verifies that $\psi_2$ is also orthogonal to the span of $\Set{\vf_1,\psi_1}$. Continue until $V$ is spanned. This construction works for the case when $m<\infty$. When $m=\infty$, we perform a so-called \emph{Zornication}.
	\end{proof}

	\begin{lem}\label{lem:extension of star fredholm partial isometries}
		For $\FF=\star\RR,\star\HH$, let $Z\in\calF_\FF$ be essentially unitary with zero Fredholm index, if applicable. Then there is a unitary $Y\in\calU_\FF$ such that $Z-Y$ is compact.
	\end{lem}
	\begin{proof}
		We have again $\polar(Z)\in \calB_\FF$ by \cref{lem:polar part preserve symmetry}. To extend $\polar(Z)$ to a unitary operator in $\calU_\FF$, the analysis divides according to the value of $\FF$. 
		
		For $\FF=\star\RR$, let $\ker Z$ be spanned by an orthonormal basis $\Set{\vf_i}_{i=1}^m$. We note that $Z\in\calB_\FF$ implies \eq{C\ker Z=\ker(Z^\ast)=(\im Z)^\perp\,.} 
		Thus $(\im Z)^\perp$ is spanned by the orthonormal basis $\Set{C\varphi_i}_{i=1}^m$. Let $M:\ker Z\to(\im Z)^\perp$ be the unitary which maps $\varphi_i$ to $C\varphi_i$. Then \eq{CM\varphi_i=CC\varphi_i=\varphi_i=M^*C\varphi} and we may extend $\polar(Z)$ by $M$ then.
		
		Next, consider the case $\FF=\star\HH$. Since the index is zero, $\dim\ker Z\in 2\NN$, so let $\Set{\varphi_i,\varphi_{i+m}}_{i=1}^m$ denote an orthonormal basis. In this case, the fact $J\ker Z=(\im Z)^\perp$ still holds. Thus, $(\im Z)^\perp$ is spanned by the orthonormal basis $\Set{J\varphi_i,J\varphi_{i+m}}_{i=1}^m$. Define the unitary map $M:\ker Z\to (\im Z)^\perp$ via\eq{\vf_i\mapsto -J\vf_{i+m},\quad \vf_{i+m}\mapsto J\vf_i\,.} 
		Then \eq{JM\vf_i=-J^2\vf_{i+m}=\vf_{i+m}=M^*J\vf_i} and similarly $JM\vf_{i+m}=M^*J\vf_{i+m}$. So again we define $Y:=\polar(Z)\oplus M$ and $Y\in\calU_\FF$.
	\end{proof}

		\subsection{Equivariant homotopies of self-adjoint unitaries}
		\begin{lem}\label{lem:classification of nontrivial projections}
			Let $\FF\in\Set{\CC,\RR,\HH,\ii\RR,\ii\HH}$. The non-trivial SAUs in $\calS_\FF$ are nullhomotopic.
		\end{lem}
		\begin{proof}
			Let $U,\ti U\in\calS_\FF$ be non-trivial. In the case when $\FF=\CC$, we can simply choose a unitary operator $W$ that maps the $\pm 1$ eigenspaces of $U$ to the respective $\pm 1$ eigenspaces of $\ti U$, since these eigenspaces are infinite dimensional by the non-triviality assumption. It is clear that \eql{\label{eq:conjugate SAU non-local case}U=W^*\ti{U}W\,.}
			We apply \cref{thm:Kuiper} to deform $W_t\rightsquigarrow \Id$ withing $\calU_\CC$ and obtain the desired path $W_t^*\ti{U} W_t\rightsquigarrow U$ within $\calS_\CC.$ 
			
			We consider $\FF=\RR$. Since $UC=CU$, it follows that \eq{C\ker(U\pm \Id) = \ker(U\pm \Id)\,.} 
			Thus, we apply \cref{lem:basis in RCH symmetry} to obtain an ONB $\Set{\vf_i^\pm}_{i=1}^\infty$ for the $\pm 1$ eigenspace of $U$ such that $C\vf_i^\pm=\vf_i^\pm$. Similarly, let $\Set{\ti{\vf}_i^\pm}_{i=1}^\infty$ be an ONB fixed by $C$ for the $\pm 1$ eigenspaces of $\ti U$. Let $W$ be the unitary operator that maps $\vf_i^\pm\mapsto \ti{\vf}_i^\pm$. It follows that \cref{eq:conjugate SAU non-local case} holds and, moreover, $WC=CW$ holds. Indeed, we have \eq{WC\vf_i^\pm=W\vf_i^\pm=\ti{\vf}_i^\pm = C\ti{\vf}_i^\pm=CW\vf_i^\pm\,.} 
			Thus, we can deform $W_t\rightsquigarrow \Id$ within $\calU_\RR$ using \cref{thm:Kuiper}, and this gives the desired path $W_t^*\ti{U} W_t\rightsquigarrow U$ within $\calS_\RR$. 
			The $\FF=\HH$ case is similar to $\FF=\RR$. Using the fact that $J\ker(U\pm \Id)=\ker(U\pm \Id)$ from $UJ=JU$ and \cref{lem:basis in RCH symmetry}, we find an ONB $\Set{\vf^\pm_i,\psi^\pm_i}_{i=1}^\infty$ for the $\pm 1$ eigenspaces of $U$ with $J\vf_i^\pm=\psi_i^\pm$, and an ONB $\Set{\ti{\vf}^\pm_i,\ti{\psi}^\pm_i}_{i=1}^\infty$ for $\ti{U}$ similarly constructed. We let $W$ maps $\vf^\pm_i,\psi^\pm_i$ to $\ti{\vf}^\pm_i,\ti{\psi}^\pm_i$ respectively. Then \cref{eq:conjugate SAU non-local case} holds, and so does $WJ=JW$. Indeed, we have \eq{WJ\vf^\pm_i=W\psi^\pm_i=\ti{\psi}^\pm_i=J\ti{\vf}^\pm_i=JW\vf_i^\pm \\ WJ\psi^\pm_i=-W\vf^\pm_i=-\ti{\vf}^\pm_i=J\ti{\psi}^\pm_i=JW\psi_i^\pm\,.} Similar to the previous case, it follows that there exists a path $U\rightsquigarrow \ti{U}$ within $\calS_\HH$ using the conjugate operator $W$.
			
			We turn to consider $\FF=\ii\RR,\ii\HH$. Let $F$ denote $C$ or $J$. Since $UF=-FU$, it follows that \eq{F \ker (U\pm \Id)=\ker (U\mp \Id)\,.} 
			Let $\Set{\vf_i^+}_{i=1}^\infty$ be an ONB for the $+1$ eigenspace of $U$. Using the above relation, then $\Set{\vf_i^-:=F\vf_i^+}_{i=1}^\infty$ is an ONB for the $-1$ eigenspace of $U$. Let $\Set{\ti{\vf}_i^+}_{i=1}^\infty$ and $\Set{\ti{\vf}_i^-:=F\ti{\vf}_i^+}_{i=1}^\infty$ be a similar construction of ONB for the $\pm 1$ eigenspaces of $\ti{U}$. Define $W$ that maps $\vf_i^\pm \mapsto \ti{\vf}_i^\pm$. Then \cref{eq:conjugate SAU non-local case} holds and, moreover, $WF=FW$ holds. Indeed, we have \eq{W F \vf_i^+=W  \vf_i^-=\ti{\vf}_i^- =F\ti{\vf}_i^+=FW\vf_i^+ \\ W F \vf_i^-=WFF\vf_i^+=\pm W\vf_i^+=\pm \ti{\vf}_i^+=F\ti{\vf}_i^-=FW\vf_i^- } where in the last line the $+1,-1$ prefactors correspond to $F=C,J$ respectively. 
			We then use \cref{thm:Kuiper} to deform $W_t\rightsquigarrow \Id$ within $\calU_\FF$, and the path $U_t:=W_t^*\ti{U} W_t\rightsquigarrow U$ will be within $\calS_\FF$. Indeed, $U_t$ is a SAU and $U_tF=W_t^*\ti{U} W_t F=-FW_t^*\ti{U} W_t=-FU_t$.
			
		\end{proof}
		
		\begin{lem}\label{lem:essential projection is compact away from genuine projection}
			Let $P$ be essentially a projection in the sense that $P^2-P\in\calK$ and $P^*-P\in\calK$. Then there exists a self-adjoint projection $Q$ such that $P-Q\in\calK$. If $P$ commutes with a given anti-unitary $F$, then so does $Q$.
		\end{lem}
		\begin{proof}
			Let $\widetilde{P}=\frac{1}{2}(P+P^*)=P+\frac{1}{2}(P^*-P)$, then $\widetilde{P}$ is self-adjoint and $\widetilde{P}^2-\widetilde{P}\in\calK$. Therefore, WLOG we assume $P$ is self-adjoint. Since $P^2-P$ is self-adjoint and compact, its spectrum can only accumulate at $0$. Thus the spectrum of $P$ can only accumulate at $0$ and $+1$. Pick any $\lambda_0\in (0,1)\setminus \sigma(P)$. Consider the self-adjoint projection $Q=\chi_{(\lambda_0,\infty)}(P)$. Now \eq{\sigma(P-Q)=\Set{\lambda-\chi_{(\lambda_0,\infty)}(\lambda) | \lambda\in\sigma(P)}\,.} 
			Thus the spectrum of $P-Q$ can only accumulates at $0$, and hence $P-Q\in\calK$.
			
			Finally, since we pick $\lambda_0$ outside the spectrum $\chi_{(\lambda_0,\infty)}$ is a continuous function which may be approximated uniformly (as $P$ is bounded) by a sequence of polynomials with real coefficients, and hence, we may guarantee that $Q$ commutes with the anti-unitary $F$ as well. 
		\end{proof}

		\section{Stummel idempotents}
		To study the space of operators of the form $A+BR$ in \cref{lem:A+BR homotopies}, we will construct a suitable grading for these operators, using so-called \emph{Stummel idempotents}. To motivate the construction, we recall the more familiar concept of the Riesz projections--actually they are only self-adjoint if the associated operator is, otherwise they are merely idempotents which is how we shall refer to them henceforth. They concern the decomposition of operators corresponding to disjoint parts of the spectrum. Let $A$ be a bounded operator whose spectrum is the disjoint union of two closed subsets $E,F$ of $\sigma(A)$. Let us specify even more, so we consider $E\subset \DD$ and $F\subset \CC\setminus\overline{\DD}$ and $\lambda\Id-A$ is invertible for $\lambda \in \bS^1$, i.e., there are two parts in the spectrum of $A$ that are separated by the unit circle. We recall the statement of the theorem concerning Riesz idempotents  (see e.g. \cite[Chapter I]{gohberg1988block}):
		\begin{thm}[Riesz]
			The operator \eq{P=\frac{1}{2\pi \ii} \oint_{\bS^1} (\lambda\Id-A)^{-1}\dif{\lambda}} is an idempotent such that $A$ decomposes as \eq{A=\begin{bmatrix}A_{11} & 0 \\ 0 & A_{22}\end{bmatrix}:\ker P\oplus \im P \to \ker P \oplus \im P\,.} In particular, the following operators are invertible \eq{\lambda \Id-A_{11}& :\ker P\to \ker P \qquad (\lambda\in \overline{\DD}) \\ \lambda\Id-A_{22}& :\im P\to \im P \qquad (\lambda\in \CC\setminus \DD) \,.}
		\end{thm}
		
		For the Riesz idempotents, we are concerned with operators of the form $A+\lambda\Id$ (note we switched from considering $\lambda\Id-A$ to $A+\lambda\Id$ for notational purposes of the later discussion). The idea can be generalized to operators of the form $A+\lambda G$, where $A,G$ are two bounded operators. Instead of considering the spectrum of $A$, we will talk about the invertibility of $A+\lambda G$ for $\lambda$ in some subset of the complex plane. For operators of the form $A+\lambda G$, we have the following 
		
		\begin{thm}(Stummel)\label{thm:stummel}
			Let $A,G\in\calB$ be given such that  $S_\lambda:=A+\lambda G$ is invertible for all $\lambda\in \bS^1$. Then the operators \eql{P=\frac{1}{2\pi \ii}\oint_{\bS^1}S_\lambda^{-1}G\dif{\lambda},\quad Q=\frac{1}{2\pi \ii}\oint_{\bS^1}GS_\lambda^{-1}\dif{\lambda} \label{eq:stummel projections}} are idempotents such that with respect to the grading \eql{\ker P\oplus \im P\to \ker Q\oplus \im Q \label{eq:stummel grading}} the operators $A,G$ decompose diagonally as \eql{\label{eq:stummel decomposition}A=\begin{bmatrix}A_{11} & 0 \\ 0 & A_{22}\end{bmatrix},\quad G=\begin{bmatrix}G_{11} & 0 \\ 0 & G_{22}\end{bmatrix} \,.} Moreover, the following operators are invertible \eql{A_{11}+\lambda G_{11}: &\ker P\to \ker Q \qquad(\lambda\in\overline{\DD}) \label{eq:stummel invertible in disk} \\ A_{22} + \lambda G_{22}:&\im P\to \im Q \qquad (\lambda\in \left(\CC\cup\Set{\infty}\right)\setminus \DD) \label{eq:stummel invertible outside disk}} where by $\lambda=\infty$, we mean the operator $G_{22}:\im P\to \im Q$.
		\end{thm}
		
		We refer the reader to \cite{stummel1971diskrete} or \cite[Chapter IV]{gohberg2013classes} for more details and context on the Stummel idempotents. Below we merely reproduce the proof of the convenience of the reader.
		\begin{proof}[Proof of \cref{thm:stummel}]
			The proof is based on the generalized resolvent identity \eql{S_\lambda^{-1}-S_{\mu}^{-1}=(\mu-\lambda)S_{\lambda}^{-1}GS_\mu^{-1}\,. \label{eq:generalized resolvent identity}} 
			
			Define an auxiliary operator \eql{K=\frac{1}{2\pi \ii}\oint_{\lambda\in\bS^1}S_\lambda^{-1}\dif{\lambda}\,. \label{eq:stummel auxilary operator}} We first show that 
			\eql{\label{eq:KGK is K}KGK=K\,.} There are contours $\Gamma_1$ and $\Gamma_2$ such that $\Gamma_2$ surrounds $\Gamma_1$, $\Gamma_1$ surrounds $\bS^1$, and $S_\lambda$ is invertible for $\lambda$ on these contours. Using the generalized resolvent identity \cref{eq:generalized resolvent identity}, we have \eq{KGK & =\frac{1}{(2\pi\ii)^2} \oint_{\Gamma_1}\oint_{\Gamma_2} S_\lambda^{-1}GS_\mu^{-1}\dif{\mu}\dif{\lambda} \\ &= \frac{1}{(2\pi\ii)^2} \oint_{\Gamma_1}\oint_{\Gamma_2} (\mu-\lambda)^{-1}(S_\lambda^{-1}-S_{\mu}^{-1}) \dif{\mu}\dif{\lambda} \\ &= \frac{1}{2\pi\ii} \oint_{\Gamma_1} \left[\frac{1}{2\pi\ii}\oint_{\Gamma_2} (\mu-\lambda)^{-1} \dif{\mu}\right] S_\lambda^{-1} \dif{\lambda} - \frac{1}{2\pi\ii} \oint_{\Gamma_2} \left[\frac{1}{2\pi\ii} \oint_{\Gamma_1} (\mu-\lambda)^{-1}\dif{\lambda}\right] S_{\mu}^{-1} \dif{\mu}\,.} Now, for the contour integrals inside the square brackets, the first one is equal to $1$ and the second one vanishes, due to the ways we construct the contours, and hence the result. Using \cref{eq:KGK is K}, we have \eq{P^2&=(KG)^2=(KGK)G=KG=P \\ Q^2&=(GK)^2=G(KGK)=GK=Q \,.} Thus $P$ and $Q$ are idempotents.
			
			Note the partitions \cref{eq:stummel decomposition} are equivalent to the expressions \eq{AP=QA,\quad GP=QG\,.} Here $GP=QG$ readily follows from the construction. For the other one, we need the identity \eql{AS_\lambda^{-1}G=GS_{\lambda}^{-1}A\,. \label{eq:ASinvG is GSinvA}} Indeed, we have \eq{AS_{\lambda}^{-1}G=(A+\lambda G)S_{\lambda}^{-1}G - \lambda G S_\lambda^{-1}G = G-GS_\lambda^{-1}(A+\lambda G) + GS_\lambda^{-1}A =  GS_\lambda^{-1}A \,.} Using \cref{eq:ASinvG is GSinvA}, it follows that \eq{AP=\frac{1}{2\pi \ii}\oint_{\bS^1}AS_\lambda^{-1}G\dif{\lambda} = \frac{1}{2\pi \ii}\oint_{\bS^1}GS_\lambda^{-1}A\dif{\lambda}=QA\,.}
			
			To show the invertibility for \cref{eq:stummel invertible in disk} and \cref{eq:stummel invertible outside disk}, we construct the inverse operators explicitly. The inverse is $S_\lambda^{-1}$ for $\lambda\in\bS^1$, and outside the unit circle, naturally we consider \eql{T_\lambda = \frac{1}{2\pi \ii}\oint_{\bS^1} (\lambda-\mu)^{-1}S_\mu^{-1}\dif{\mu} \qquad (\lambda\notin \bS^1)\,. \label{eq:the inverse T lambda}} First of all, $T_\lambda$ decomposes as \eq{T_\lambda = \begin{bmatrix} (T_\lambda)_{11} & 0 \\ 0 & (T_\lambda)_{22}\end{bmatrix}: \ker Q\oplus \im Q \to \ker P\oplus \im P} so that $T_\lambda$ can serve as the inverse for $S_\lambda$ according to the grading \cref{eq:stummel grading}. To verify the decomposition, it is equivalent to show \eq{T_\lambda Q=PT_\lambda \qquad (\lambda\notin \bS^1)\,.} Indeed, using the generalized resolvent identity \cref{eq:generalized resolvent identity}, we have \eq{T_\lambda Q &= \left[\frac{1}{2\pi \ii}\oint_{\bS^1}(\lambda-\mu)^{-1} S_\mu^{-1} \dif{\mu}\right]\left[ \frac{1}{2\pi \ii} \oint_{\bS^1} GS_\zeta^{-1}\dif{\zeta}\right] \\ &=\frac{1}{(2\pi \ii)^2}\oint_{\bS^1}\oint_{\bS^1} (\lambda-\mu)^{-1}S_\mu^{-1}  GS_\zeta^{-1}\dif{\mu}\dif{\zeta} \\ &= \frac{1}{(2\pi \ii)^2}\oint_{\bS^1}\oint_{\bS^1} (\lambda-\mu)^{-1}(\zeta-\mu)^{-1}(S_\mu^{-1} -S_\zeta^{-1})\dif{\mu}\dif{\zeta} \\ &= \frac{1}{(2\pi \ii)^2}\oint_{\bS^1}\oint_{\bS^1} (\lambda-\mu)^{-1}(\mu-\zeta)^{-1}( S_\zeta^{-1}-S_\mu^{-1})\dif{\mu}\dif{\zeta} \\ &=\frac{1}{(2\pi \ii)^2}\oint_{\bS^1}\oint_{\bS^1} (\lambda-\mu)^{-1} S_\zeta^{-1}GS_\mu^{-1} \dif{\mu}\dif{\zeta} \\ &= \left[\frac{1}{2\pi \ii}\oint_{\bS^1} S_\zeta^{-1} G\dif{\zeta}\right]\left[ \frac{1}{2\pi \ii} \oint_{\bS^1} (\lambda-\mu)^{-1} S_\mu^{-1} \dif{\mu}\right] = PT_\lambda \,.} Finally, using \eq{S_\lambda=(\lambda-\mu+\mu)G+A=(\lambda-\mu)G+S_\mu} we compute \eq{T_\lambda S_\lambda &= \frac{1}{2\pi \ii}\oint_{\bS^1}(\lambda-\mu)^{-1} S_\mu^{-1} \left[(\lambda-\mu)G+S_\mu\right] \dif{\mu} \\ &= \frac{1}{2\pi \ii}\oint_{\bS^1} S_\mu^{-1}G\dif{\mu} - \frac{1}{2\pi \ii}\oint_{\bS^1} (\mu-\lambda)^{-1}\dif{\mu} = \begin{cases}P & \lambda\in \DD \\ -P^\perp & \lambda \in\CC\setminus \DD\end{cases}\,.} 
			Similarly, we have \eq{S_\lambda T_\lambda = \begin{cases}Q & \lambda\in \DD \\ -Q^\perp & \lambda \in\CC\setminus \DD\end{cases}} The invertibility of $G_{22}:\im P\to \im Q$ heuristically follows by taking $\lambda=\infty$ in \cref{eq:stummel invertible outside disk}. In fact, its inverse is exactly the operator \cref{eq:stummel auxilary operator}. This follows from the identities $KQ=PK$, $KG=P$ and $GK=Q$. Here \eq{KQ=K(GK)=(GK)K=Q\,.}
		\end{proof}

		\begingroup
		\let\itshape\upshape
		\printbibliography
		\endgroup
	\end{document}